\documentclass[aps,pra,onecolumn, notitlepage, superscriptaddress,nofootinbib,longbibliography]{revtex4-1}

\usepackage[T1]{fontenc}
\usepackage{amsmath,amsfonts,amssymb,amsthm,bbm,bm,braket,relsize,exscale}
\usepackage{graphicx}
\usepackage{subcaption}
\usepackage{wasysym}
\usepackage[usenames,dvipsnames]{color}
\usepackage[colorlinks,bookmarks=false,citecolor=NavyBlue,linkcolor=OliveGreen,url color=blue]{hyperref}
\newcommand{\1}{\mathbbm{1}}

\newtheorem{Claim}{Claim}

\newtheorem{Lemma}{Lemma}

\allowdisplaybreaks

\begin{document}


\title{Generalized Spectral Statistics in the Kicked Ising model}

\author{Divij Gupta}
\affiliation{Department of Physics, Brandeis University, Waltham, Massachusetts, USA}
\affiliation{Deparment of Physics, University of Illinois, Urbana-Champaign, Illinois, USA}
\author{Brian Swingle}
\affiliation{Department of Physics, Brandeis University, Waltham, Massachusetts, USA}


\begin{abstract}

The kicked Ising model has been studied extensively as a model of quantum chaos. Bertini, Kos, and Prosen studied the system in the thermodynamic limit, finding an analytic expression for the spectral form factor, $K(t)$, at the self-dual point with periodic boundary conditions. The spectral form factor is the 2nd moment of the trace of the time evolution operator, and we study the higher moments of this random variable in the kicked Ising model. A previous study of these higher moments by Flack, Bertini, and Prosen showed that, surprisingly, the trace behaves like a real Gaussian random variable when the system has periodic boundary conditions at the self dual point. By contrast, we investigate the model with open boundary conditions at the self dual point and find that the trace of the time evolution operator behaves as a complex Gaussian random variable as expected from random matrix universality based on the circular orthogonal ensemble. This result highlights a surprisingly strong effect of boundary conditions on the statistics of the trace. 
We also study a generalization of the spectral form factor known as the Loschmidt spectral form factor and present results for different boundary conditions.

\end{abstract}

\maketitle
\section{Introduction}

One of the central pillars of the broad field of quantum chaos is the idea of random matrix universality: that the energy spectrum of a sufficiently generic Hamiltonian will exhibit level-repulsion like that found in a random Hermitian matrix~\cite{Casati1980,Berry1981,Bohigas1984,Haake2001,Mehta1991}. This idea can also be extended to Floquet operators, in which case the appropriate comparator random matrix is unitary rather than Hermitian. Much of the work in this field has studied correlations between pairs of eigenvalues, but random matrix universality makes predictions for higher order correlations as well. In an appropriate parameter regime, one typically expects that these correlations can be understood as those of an approximately Gaussian random variable, as has been verified in a variety of studies~\cite{Fritz_Haake_1999,Cotler_2017,PhysRevD.98.086026,Flack_2020,Forrester_2021,Cipolloni_2023,altland2025statisticsrandommatrixspectral}. The corrections to Gaussianity are more complicated, with some conflicting claims in the literature, as discussed in~\cite{altland2025statisticsrandommatrixspectral,Legramandi:2024ljn}. In this paper, we study these higher correlations in the kicked Ising model~\cite{Akila_2016,Bertini_2018}, which has previously been analyzed at the level of the second moment by Bertini et al.~\cite{Bertini_2018} and for higher moments by Flack et al.~\cite{Flack_2020} for the case of periodic boundary conditions. We find the expected approximately Gaussian results consistent with~\cite{Flack_2020}, albeit from a somewhat different perspective; we also consider generalizations to open boundary conditions and to a more general class of ``Loschmidt'' moments in which we compare spectra of Hamiltonians or Floquet operators which are correlated but not identical.  As we explain shortly, among our results, one particularly noteworthy feature is a surprising dependence of the higher order correlations on boundary conditions.

The quantities we study are higher-moment generalizations of the well-known spectral form factor (SFF), which we quickly review. Given a Floquet system with Floquet operator $U_F$, the basic object of interest is 
\begin{equation}
    Z(t)=\text{tr}(U_F^t),
\end{equation}
which represents the return amplitude after $t$ periods of Floquet evolution. In general, $Z(t)$ is a complex random variable that, when considered as a function of time, encodes all the spectral data of $U_F$. The SFF is simply
\begin{equation}
    K(t) = \mathbb{E}( |Z(t)|^2) ,
\end{equation}
where $\mathbb{E}$ denotes an average over some random parameters in the definition of $U_F$.\footnote{It is well known that $|Z(t)|^2$ is not self-averaging~\cite{PhysRevLett.78.2280}, so averaging over an ensemble of Floquet operators is essential to obtain a smooth $K(t)$. There is also a point of notation here, with some authors defining the SFF as $|Z(t)|^2$ without the average.}

We are particularly interested in higher moments of $|Z(t)|^2$ of the form
\begin{equation}
    M_{2\ell} = \mathbb{E}( |Z(t)|^{2\ell} ).
\end{equation}
Note that $M_2 = K$ by definition. If $Z(t)$ is approximately a complex Gaussian, as predicted by random matrix universality, then these moments would be
\begin{equation}
    \mathbb{E}( |Z(t)|^{2\ell} ) \approx \ell! K(t)^\ell. \label{eq:c_gauss_pred}
\end{equation}
We study these moments in the Floquet kicked-Ising model of Bertini et al. \cite{Bertini_2018} and obtain the following results.

We first consider the model at the dual unitary point with periodic boundary conditions, the situation considered in \cite{Flack_2020}. There it was pointed out that the statistics of $Z(t)$ were not in agreement with the statistics one would obtain if $U_F$ were drawn from the COE (circular orthogonal ensemble), as was assumed in \cite{Bertini_2018}. Instead, beyond the usual time reversal symmetry, an additional symmetry of the Floquet operator was found. This symmetry constrains the Floquet operators and means that a more complicated random matrix ensemble is the appropriate comparator~\cite{duenez2004random}, with the main consequence of this being that $Z(t)$ effectively behaves like a real Gaussian random variable rather than a complex one. This difference is only pronounced for higher moments, which is why the issue did not arise in \cite{Bertini_2018}. 
We provide  confirmatory numerical evidence and a different analysis of the higher moment transfer matrix which obtains the same result without explicitly invoking the symmetry. 

Second, we consider a generalization of the model to open boundary conditions. Remarkably, we find numerically that the statistics of $Z(t)$ become those of a complex Gaussian variable, even at the dual unitary point. We modify the transfer matrix setup to deal with the open boundary conditions, and sketch an approach to derive our numerical results from this modified transfer matrix setup.

Third, we compute the Loschmidt spectral form factor (LSFF), which generalizes the usual SFF (2nd moment of $Z(t)$) to allow different Hamiltonians for forward and backward evolution \cite{Winer_2022}. More concretely, the 2nd moment $M_2 = \mathbb{E}(Z(t) Z^*(t))$ computes the average of the return amplitude after $t$ periods of forward and $t$ periods of backward evolution. One may then imagine an imperfect time reversal procedure, given by a slightly different Hamiltonian $H_{\rm backward} = H_{\rm forward} + \delta H_0$ where $H_0$ is some perturbation. The LSFF is then defined as $K_{\rm E} (t) = \mathbb{E} (Z_{\rm forward} (t) Z^*_{\rm backward} (t))$, or more generally for Hamiltonians $H_1, H_2$ as
\begin{equation}
    K_{\rm E} (t) = \mathbb{E} \left[Z_1 (t) Z^*_2 (t)\right]
\end{equation}
We compute this numerically for the kicked-Ising model in a perturbative regime $H_2 = H_1 + \delta H_0$ for $\delta \ll 1$, along with higher moment analogues. We also appeal to a modified transfer matrix setup to make sense of the numerical results, although we leave a full analysis for the future.

The rest of the paper is organized as follows. In Section \ref{sec:model} we review the kicked-Ising model and the transfer operator calculation of $K$. In Section \ref{sec:periodic} we study the model with periodic boundary conditions and present an alternative set of numerical and analytical results that confirm~\cite{Flack_2020}. In Section \ref{sec:open} we study the model with open boundary conditions and show that complex Gaussian statistics are recovered even at the dual unitary point. In Section~\ref{sec:loschmidt} we consider the generalization to the Loschmidt spectral form factor. Finally, in Section \ref{sec:discussion} we discuss the results and give some outlook.

\section{Model and background}
\label{sec:model}

Following Bertini et al \cite{Bertini_2018}, we consider a Floquet Ising spin-$1/2$ chain with transverse and longitudinal fields, also known as the kicked Ising model, as an example of an exactly solvable model of quantum chaos. The time and field-dependent Hamiltonian ($H_{\rm KI}[\boldsymbol h; t]$) is given in terms of an interaction ($H_{\rm I}[\boldsymbol h]$) and kick ($H_{\rm K}$) term \cite{prosenHamiltonian1,ProsenHamiltonian2}:
\begin{equation}
    H_{\rm KI}[\boldsymbol h;t]=H_{\rm I}[\boldsymbol h]+\delta_p(t) H_{\rm K}
    \label{eq:ham}
\end{equation}
The time dependence is captured by $\delta_p(t)=\sum_{m=-\infty}^\infty\delta(t-m)$, the periodic delta function, and the interaction and kick terms are given by
\begin{equation}
    H_{\rm I}[\boldsymbol h]\equiv\! J\sum_{j=1}^{L} \sigma^{z}_j \sigma^z_{j+1} + \sum_{j=1}^{L} h_j \sigma^z_{j},\,\,\, H_{\rm K}\equiv b \sum_{j=1}^{L}  \sigma^x_j,
    \label{eq:hamiltonians}    
\end{equation}
where $L$ is the length of the chain and $\sigma_j^\alpha$, $\alpha\in\{x,y,z\}$ are the Pauli matrices at position $j$. The parameters $J, b$ are, respectively, the coupling of the Ising chain and the kick strength, while $\boldsymbol h=(h_1,\ldots, h_L)$ describes a position dependent longitudinal field. As written, the first sum for $H_{\rm I}$ ranges from $j=1$ to $L$ which leads to periodic boundary conditions with the identification $\sigma_{L+1}^{\alpha} \equiv \sigma_{1}^{\alpha}$. We may also consider open boundary conditions in which case the first sum in $H_{\rm I}$ now only ranges from $j=1$ to $L-1$.\footnote{It should be noted that our definition of the Hamiltonian deviates slightly from \cite{Flack_2020}, where the interaction Hamiltonian has an additional factor of $-\1_L$ in the first sum, where $\1_L$ is the identity operator in $(\mathbb{C}^2)^{\otimes L}$. This plays a minor role in the form of the quasienergies later.}

Next, we review the analysis of the spectral form factor in \cite{Bertini_2018}. The Floquet operator discussed in the introduction takes the form,
\begin{equation}
    U_{\rm KI}[\boldsymbol h]= 
T\!\exp\!\!\left[-i\int_0^1\!\!\!\!{\rm d}s\, H_{\rm KI}[\boldsymbol h;s]\right]\!\!=
e^{-i H_{\rm K}}e^{-i H_{\rm I}[\boldsymbol h]}\,
\label{eq:floquet}
\end{equation}
The return amplitude is then,
\begin{equation}
   Z(t) =  {\rm tr}\left(U^t_{\rm KI}[\boldsymbol h]\right).
\end{equation}
As stated in the introduction, our goal is to study the moments of $Z(t)$ for the ensemble of Floquet operators defined by \eqref{eq:floquet}. This is useful to characterize the distribution of the quasienergies $\{\varphi\}$, which are the phases of the eigenvalues of the Floquet operator. 


Each instance of the ensemble is given by a random choice of the $h_j$ fields in \eqref{eq:floquet}. We will compare averages of powers of $Z(t)$ over these random fields to averages of powers of ${\rm tr}(U_{{\rm ensemble}}^t)$ averaged over the relevant ensemble (we will shortly describe which ensemble to choose for the kicked-Ising model). The longitudinal fields on each site are independent Gaussian random variables with mean $\bar{h}$ and standard deviation $\sigma > 0$. The average of any function of the fields is thus,
\begin{equation*}
    \mathbb{E}_{\boldsymbol h} [f(\boldsymbol h)] = \int_{\infty}^{\infty} f(\boldsymbol h) \prod_{j=1}^L e^{-\frac{(h_j - \bar{h})^2}{2\sigma^2}} \frac{d h_j}{\sigma\sqrt{2\pi}}.
\end{equation*}

The goal, following \cite{Bertini_2018}, is to analytically evaluate the disorder average for the SFF, $K(t)$, in the thermodynamic limit $L \rightarrow \infty$. To do so, it is helpful to think of the trace of the Floquet operator as the partition function of a 2 dimensional Ising model on a $t \times L$ lattice, which is periodic in time,
\begin{equation}
    {\rm tr}\left(U^t_{\rm KI}[\boldsymbol h]\right) = \left[\frac{\sin{(2b)}}{2i}\right]^{Lt/2} \sum_{{s_{\tau,j}}} e^{-i\mathcal{E}
    \left[\{s_{\tau,j}\}, \boldsymbol h\right]}.
\end{equation}
Here we take $\{s_{\tau, j}\}$ to be classical spin variables in both the time and spatial directions, with the sum over all possible configurations, and energy given by
\begin{equation}
    \mathcal{E}\left[\{s_{\tau, j}, \boldsymbol h\}\right] = \sum_{\tau = 1}^t \left(J \sum_{j=1}^L s_{\tau, j} s_{\tau, j+1} + J' \sum_{j=1}^L s_{\tau, j} s_{\tau + 1, j} + \sum_{j=1}^L  h_j s_{\tau,j} \right),
\end{equation}
where $J' = -\frac{\pi}{4} - \frac{i}{2} \log{\tan{b}}$. $J$ acts as the spatial spin coupling constant, while $J'$ can be thought of as the temporal spin coupling constant. Note that this assumes (spatially) periodic boundary conditions; open boundary conditions can be imposed by changing the upper bound of the 1st sum over $j$ to $L-1$.

It is important to note that only adjacent spins are coupled both spatially and temporally; thus we can write the partition function as either the trace of a spatial transfer matrix, $U_{\rm KI}[\boldsymbol h]$, or a temporal transfer matrix, $\tilde{U}_{\rm KI}[\boldsymbol h_j]$. Here the temporal transfer matrix acts at a given point $j$ on the spatial lattice $(\boldsymbol h_j \equiv h_j \boldsymbol \epsilon, \text{where } \boldsymbol \epsilon = (1,1 \cdots 1) \text{ is a $t$-component vector})$ for all times $t$. $\tilde{U}_{\rm KI}[\boldsymbol h_j]$ is given by the same algebra as $U_{\rm KI}[\boldsymbol h]$ \eqref{eq:floquet} but acting on a spin chain of $t$ sites and with $J$ replaced by $J'$. Since the dual transfer matrix $\tilde{U}_{\rm KI}$ is not generically unitary, the analysis of \cite{Flack_2020} is restricted to the self dual points, defined by $|J| = |b| = \pi/4$, and thus $J' = \mp \pi/4$ for $b = \pm \pi/4$.

The relevant random matrix ensemble to compare this with is determined by the constraints on the Floquet operator. At the dual unitary point, the first symmetry is the time reversal symmetry, i.e. there exists an anti-unitary operator $T$ such that $T U_{\rm KI} [\boldsymbol h] T^{-1} = U^{-1}_{\rm KI} [\boldsymbol{h}]$. Only considering this symmetry would suggest the relevant ensemble is COE, however an additional symmetry was found in \cite{PhysRevE.101.052201}. Defining $F_y \equiv
\prod_{j = 1}^{L} \sigma_y^j$ and $U = U_{\rm KI} [\boldsymbol{h}] e^{i \frac{\pi}{4} L}$, we obtain the symmetry relation,
\begin{equation}
    \label{eq:conjugate symmetry}
    U^* = F_y ^\dagger U F_y 
\end{equation}
where $(\cdot)^*$ represents complex conjugation in the computational basis (the standard Pauli basis where $\sigma_x,\sigma_z$ are real). This implies that the quasienergies of $U$ come in pairs $\{\varphi,-\varphi\}$ and thus the spectrum of $U$ is real. The spectrum of $U^t_{\rm KI}[\boldsymbol{h}]$ is thus also real, up to a phase $\phi = e^{i \frac{\pi}{4} Lt}$, which suggests that the return amplitude $Z(t)$ obeys real random statistics.

Furthermore, one can deduce that the above symmetry relation reduces to the Floquet operators being orthogonal for even $L$ and symplectic for odd $L$. \cite{Flack_2020} concludes that the relevant ensemble for even $L$ is $S_+(\mathcal{N}/2) \equiv O(\mathcal{N})/(O(\mathcal{N}/2) \times O(\mathcal{N}/2))$, where $O(\mathcal{\mathcal{N}})$ is the orthogonal group of $\mathcal{N}\times\mathcal{N}$ matrices, with $\mathcal{N} = 2^L$ is the dimension of the Hilbert space $\mathcal{H}_L = (\mathbb{C}^2)^{\otimes L}$. For odd $L$, the relevant ensemble is $S_-(\mathcal{N}/2) \equiv Sp (\mathcal{N}/2)/U (\mathcal{N}/2)$ where $U(\mathcal{N}/2)$ is the unitary group and $Sp(\mathcal{N}/2) \subset U(\mathcal{N})$ is the unitary-symplectic group \cite{duenez2004random}. As a basic check, both ensembles have quasienergies appearing in pairs.

Now we return to the analysis of the SFF. At the dual unitary point, we may write the 2nd moment (SFF) in terms of either spatial or temporal transfer matrices,
\begin{align}
    K(t) &= \mathbb E_{\boldsymbol h}\left[{\rm tr}\left(U^t_{\rm KI}[\boldsymbol h]\right){\rm tr}\left(U^t_{\rm KI}[\boldsymbol h]\right)^*\right] \\
    &= \mathbb E_{h}\left[{\rm tr}\left(\tilde{U}^L_{\rm KI}[h \boldsymbol \epsilon]\right){\rm tr}\left(\tilde{U}^L_{\rm KI}[h \boldsymbol \epsilon]\right)^*\right].
\end{align}
Since $\left|{\rm tr}  (U)\right|^2 = {\rm tr}(U \otimes U^*)$, we can rewrite the 2nd expression using a single transfer matrix,
\begin{equation}
    K(t) = {\rm tr}\left(\mathbb{T}^L\right)
\end{equation}
where,
\begin{align}
\label{eq:transfer 2}
    \mathbb{T} &= \mathbb E_{h}\left[\tilde{U}_{\rm KI}[h \boldsymbol \epsilon] \otimes \tilde{U}_{\rm KI}[h \boldsymbol \epsilon]^*\right] = (\tilde{U}_{\rm KI} \otimes \tilde{U}_{\rm KI}^*)\cdot \mathbb{O}_{\sigma} \\
    \label{eq:O 2}
    \mathbb{O}_{\sigma} &= \exp{\left[-\frac{\sigma^2}{2} \left(M^{(1)}_z - M^{(2)}_z\right)^2\right]}.
\end{align}
Here we define $\tilde{U}_{\rm KI} \equiv \tilde{U}_{\rm KI}[\bar{h} \boldsymbol \epsilon]$, $M_{\alpha} \equiv \sum_{\tau = 1}^t \sigma^\alpha_\tau$ for $\alpha \in \{x,y,z\}$, and the superscripts on $M$ denote its position in a tensor product with the identity operator $\1$ (i.e. $M_z^{(1)} = M_z \otimes \1$). Thus the average over the longitudinal fields mixes conjugate copies of the system.

It turns out that the eigenvalues of $\mathbb{T}$ have at most unit magnitude and, for odd time $t$, the only unimodular eigenvalue of $\mathbb{T}$ is 1. Thus in the thermodynamic limit $L \rightarrow \infty$, the SFF is simply the total number of linearly independent unimodular eigenvectors of $\mathbb{T}$~\cite{Bertini_2018}. The eigenvectors with unimodular eigenvalues are associated with states $\ket{A}$ which must satisfy
\begin{align}
    \left(\tilde{U}_{\rm KI} \otimes \tilde{U}_{\rm KI}^* \right) \ket{A} &= e^{i\phi} \ket{A} \\
    \left(M_z^{(1)} - M_z^{(2)}\right) \ket{A} &= 0.
\end{align}
These can be rewritten as operator equations by expanding $\ket{A} = \sum\limits_{n,m}^{} {A_{n,m} \ket{n} \otimes \ket{m}^*}$, where ${\ket{n}}$ forms a basis of $\mathcal{H}_t$ and $A_{n,m} = \bra{n}A\ket{m}$ are interpreted as matrix elements,
\begin{align}
    \label{eq:m2 action}
    [A,M_z] = 0 \\ 
    \label{eq:u2 action}
    \tilde{U}_{\rm KI} A \tilde{U}_{\rm KI}^{\dagger} = e^{i\phi} A
\end{align}
The problem of computing the trace of the transfer matrix has now been reduced to counting the number of linearly independent operators satisfying the conditions above. \cite{Bertini_2018} then finds the analytic expression (for odd $t$),
\begin{equation}
    \lim_{L\rightarrow\infty}{K}(t) = \begin{cases}
        2t-1, &t \leq 5 \\
        2t, &t \geq 7
    \end{cases}
\end{equation}
This matches the COE prediction (in the thermodynamic limit $L \rightarrow \infty$) $K_{\rm COE}(t) = 2t - t \ln{(1 + 2t/\mathcal{N})}$ for $0 < t < \mathcal{N}$, where $\mathcal{N} = 2^L$ is the dimension of the Hilbert space $\mathcal{H}_L$. This, however, is a coincidence; real and complex Gaussians with identical 2nd moments will have different higher moments. By the symmetry arguments above, it is the SFF for the $S_+$ (even $L$) or $S_-$ (odd $L$) ensemble that must be compared to $K(t)$. However, these ensembles are not commonly encountered and COE matrices are easy to generate numerically. Thus for any numerical analysis, we will simply use appropriately scaled results for the COE. Higher moments $M_{2\ell} = \mathbb{E}( |Z(t)|^{2\ell} )$ of the COE return amplitude $Z_{\rm COE}(t)$ obey roughly complex Gaussian statistics, which differs from real Gaussian statistics, i.e. the statistics of $Z_{S_{\pm}} (t)$, by a factor of $(2\ell - 1)!!/\ell !$ ($!!$ denotes the double factorial). Thus, we can rescale the COE moment $M^{(\rm COE)}_{2\ell} = \mathbb{E}( |Z_{\rm COE}(t)|^{2\ell} ) \approx \ell ! K_{\rm COE}(t)^{\ell}$ by this factor, and we expect this rescaled moment $((2\ell - 1) !!/\ell!) M^{(\rm COE)}_{2 \ell}$ will be a reasonable proxy for the true moment, $M^{(S_{\pm})}_{2 \ell}$. 
\\\\
For sufficiently large finite size chains, the numerical results will be close to the analytic computation above (which is valid in the thermodynamic limit). We generate $N$ samples of the (spatial) Floquet operator with a normally distributed longitudinal field, then compute the 2nd moment for each sample and average over our ensemble.\footnote{The dual time Floquet operator can also be used; however the matrix size then depends on the timescale and limits analysis over longer timescales.} Comparing to the COE prediction in Figure~\ref{fig:order 2 figs}, we see that the system SFF matches the COE prediction (we also average over $N$ randomly generated COE matrices of appropriate size to get the COE result) for both the open and closed boundary conditions after an early time transient period. 

\begin{figure}[ht]
    \centering
    \begin{subfigure}[b]{0.49\textwidth}
        \centering
        \includegraphics[width=\textwidth]{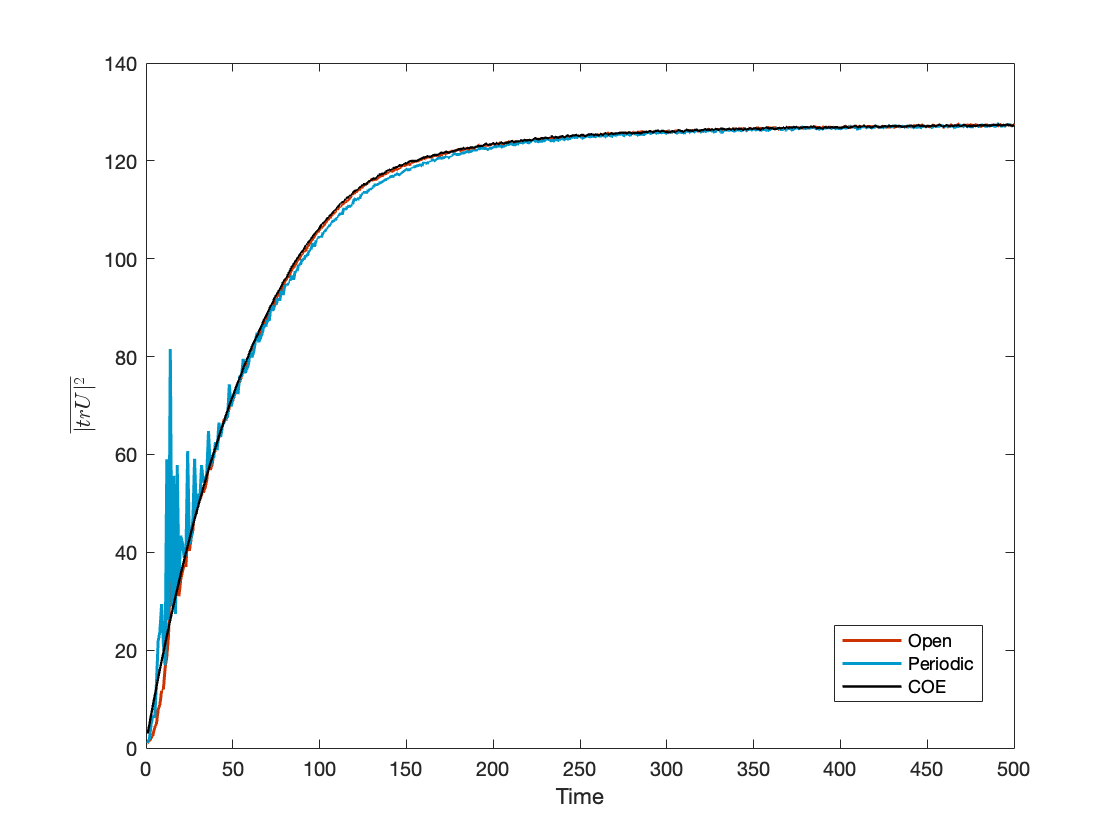}
    \end{subfigure}
    \begin{subfigure}[b]{0.49\textwidth}
        \centering
        \includegraphics[width=\textwidth]{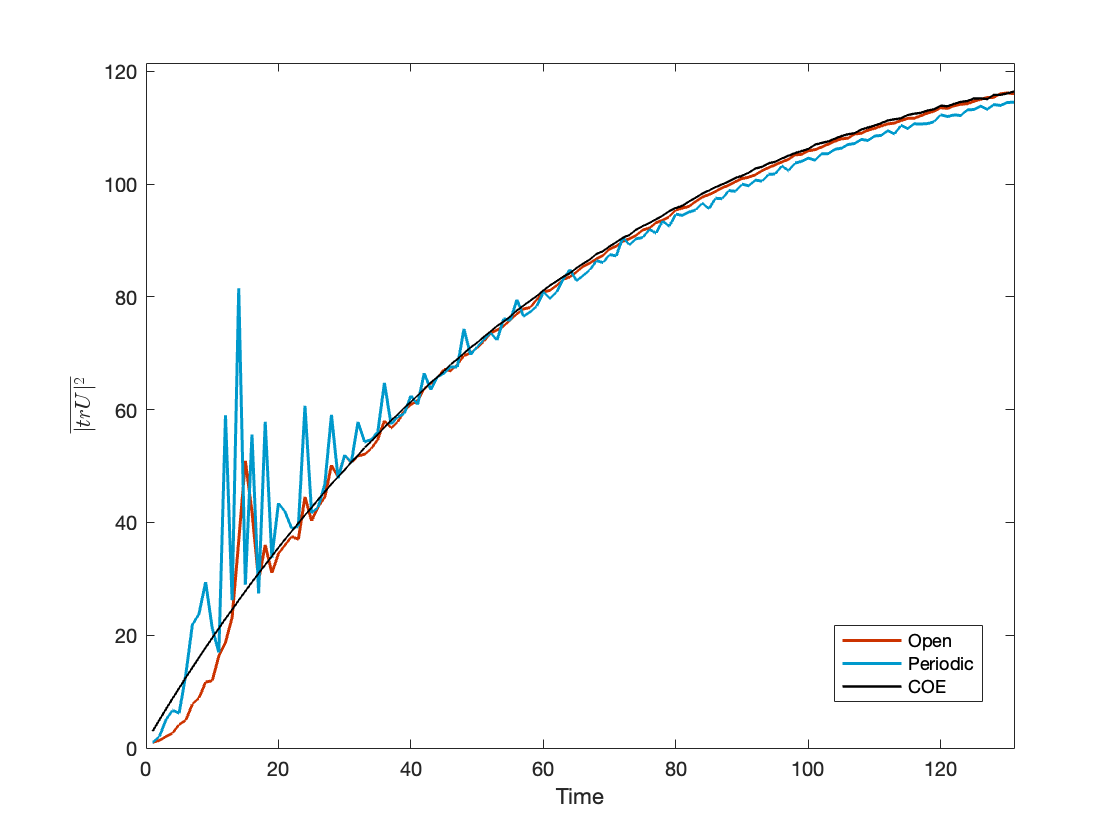}
    \end{subfigure}
    \caption{2nd moment for Floquet and COE trace of 7-qubit model with $J = b = \pi/4$, $N = 10^6$ samples and $\bar{h} = 0.6, \sigma = 100\pi$. The ramp-plateau behavior can be seen in the first plot, with the second focusing on the ramp region}
    \label{fig:order 2 figs}
\end{figure}

We emphasize that the analysis described above only applies to periodic boundary conditions. To obtain the result for open boundary conditions, we must replace the trace of $\mathbb{T}^L$ with a matrix element of $\mathbb{T}^{L-1}$ between appropriately defined boundary vectors, and the appropriate ensemble to compare to will also be different. We will discuss this in Section~\ref{sec:open}. 

\section{Periodic boundary conditions}
\label{sec:periodic}
We start our discussion of periodic boundary conditions by providing some numerical evidence in support of our statement earlier that the higher order moments for the kicked-Ising model should be well approximated by rescaled moments of the COE, even though it is the incorrect ensemble. After providing this numerical evidence, we compute the higher order moments analytically by explicitly constructing unimodular eigenvectors of higher order transfer matrices.

\subsection{Numerics}
For the numerical study in this section, all plots are generated by averaging over $N = 10^6$ samples of \eqref{eq:floquet} defined by $\bar{h} = 0.6, \sigma = 10\pi$ for a $7$-qubit chain ($L = 7$) at the self dual point $J = b = \pi/4$, unless otherwise noted.

\begin{figure}[ht]
    \centering
    \begin{subfigure}[b]{0.49\textwidth}
        \centering
        \includegraphics[width=\textwidth]{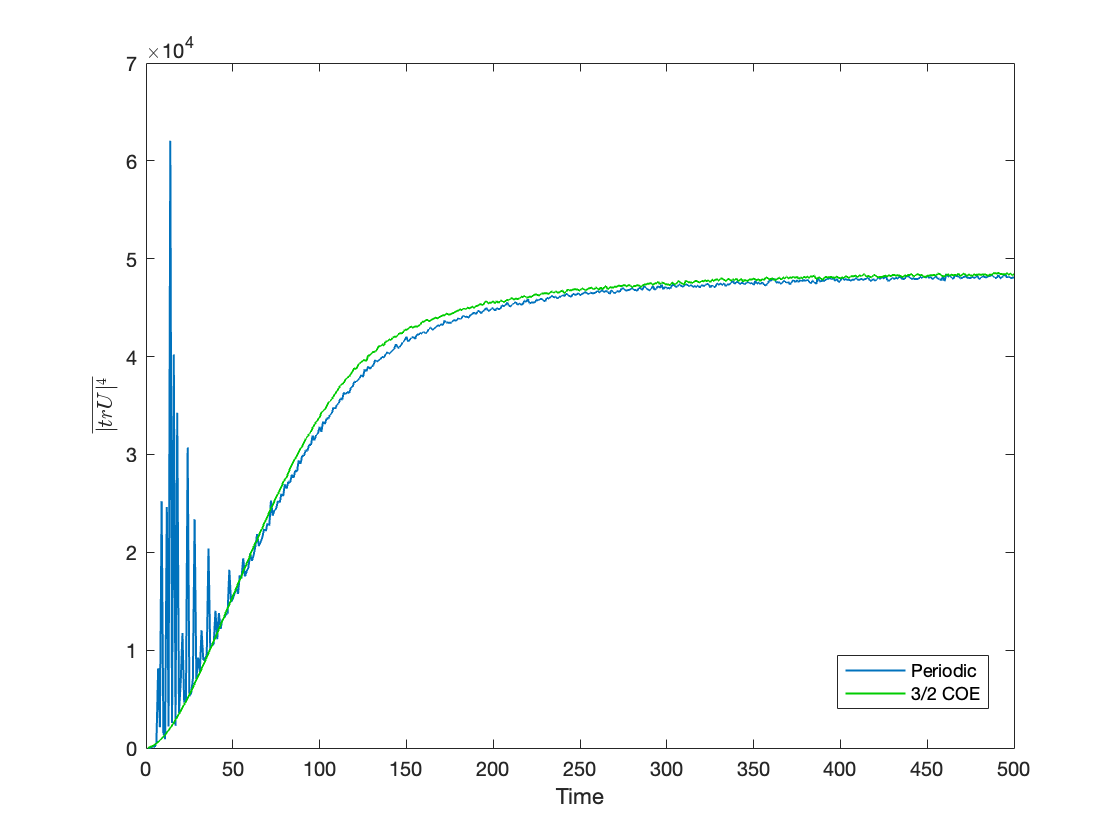}
        \label{fig:periodic 4th}
    \end{subfigure}
    \begin{subfigure}[b]{0.49\textwidth}
        \centering
        \includegraphics[width=\textwidth]{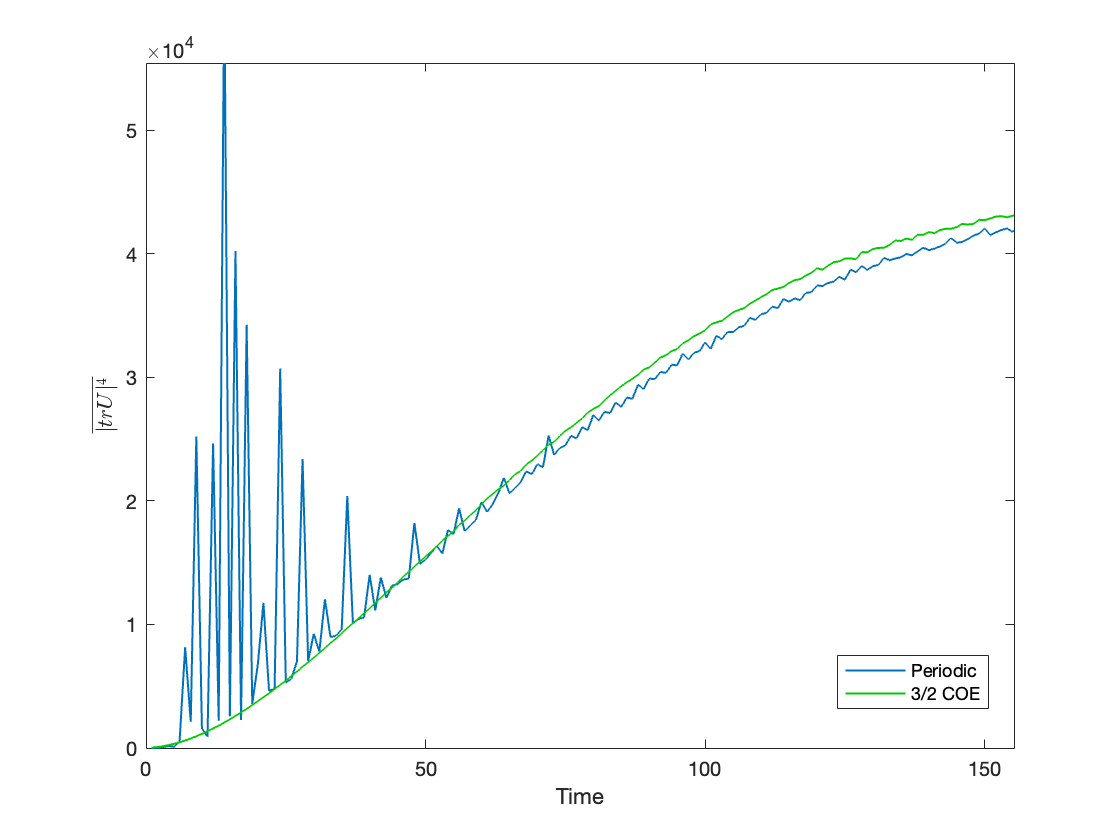}
        \label{fig:periodic 4th zoom}
    \end{subfigure}
    \caption{(Left) 4th moment of $Z(t)$ for the periodic boundary condition Floquet operator. (Right) Zooming in on the ramp region. The green curve is $3/2$ times the COE prediction.}
    \label{fig:periodic o4}
\end{figure}

We show the data for the 4th moment in Figure~\ref{fig:periodic o4}. Up to small deviations, the 4th moment for periodic boundary conditions is close to $3/2$ times the COE prediction. This is precisely $(2\ell-1)!!/\ell!$, the ratio between the 4th moment of a real and complex Gaussian variable, which is what we expect, as by the symmetry arguments $Z(t)$ must behave like a real Gaussian random variable (for periodic boundary conditions).

We can additionally verify this for higher moments, as shown in Figure~\ref{fig:periodic higher order}. The COE predictions would be $6 K^3$ (6th moment) and $24 K^4$ (8th moment), whereas the moments for a real Gaussian variable would be $15 K^3$ (6th moment) and $105 K^4$ (8th moment). Thus the periodic moments are compared to appropriately scaled predictions, i.e. $15/6$ times the COE prediction (6th moment) and $105/24$ times the COE prediction (8th moment).

\begin{figure}[ht]
    \centering
    \begin{subfigure}[b]{0.49\textwidth}
        \centering
        \includegraphics[width=\textwidth]{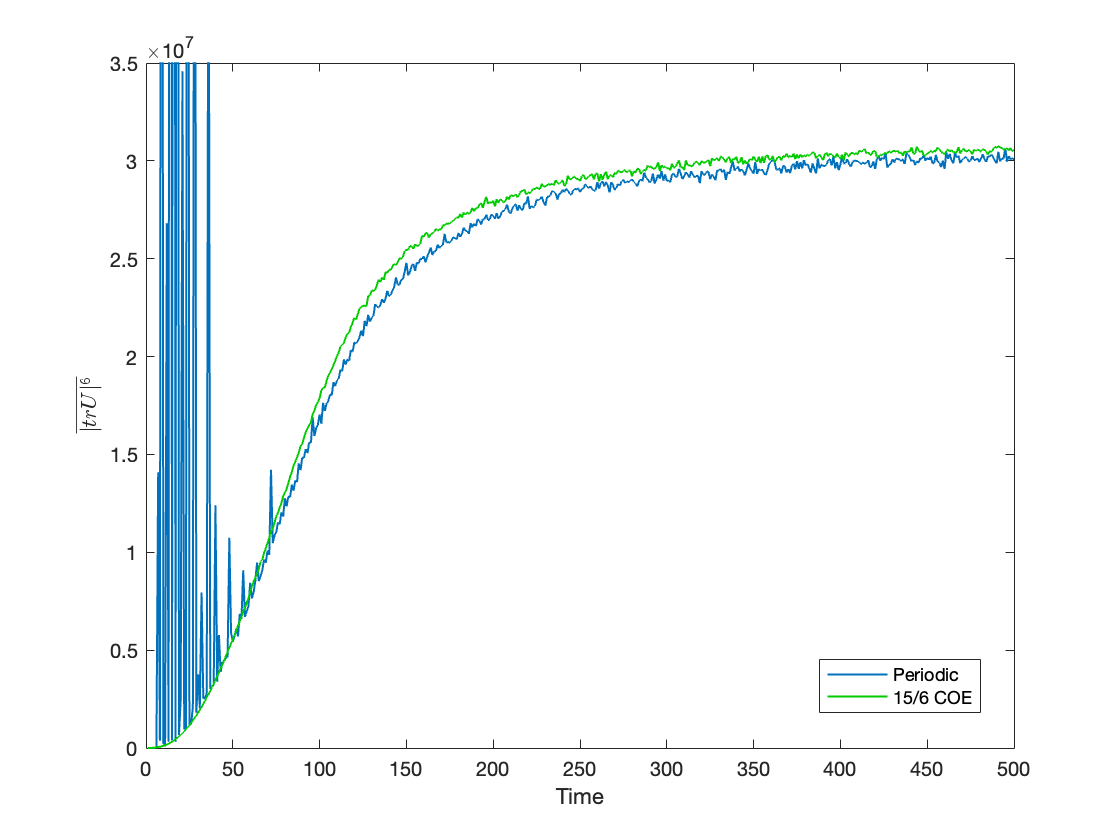}
        \subcaption{6th moment}
        \label{fig:periodic 6th}
    \end{subfigure}
    \begin{subfigure}[b]{0.49\textwidth}
        \centering
        \includegraphics[width=\textwidth]{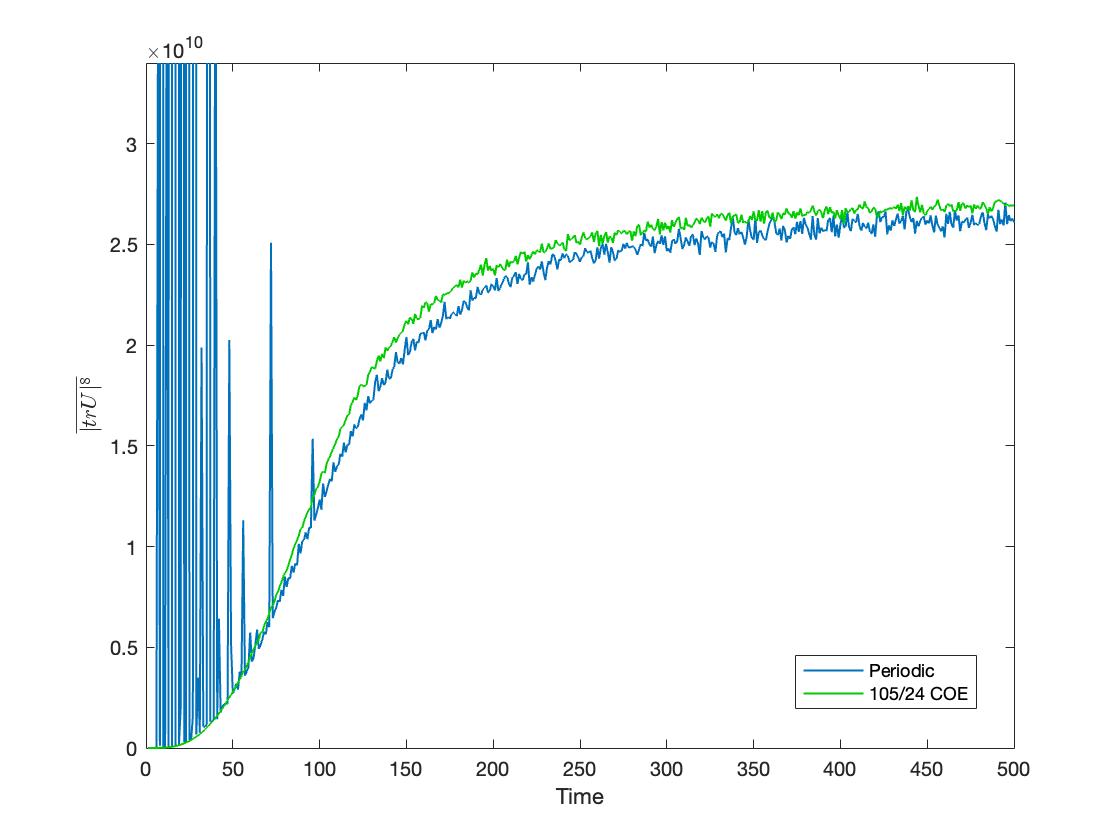}
        \subcaption{8th moment}
        \label{fig:periodic 8th}
    \end{subfigure}
    \caption{(Left, a) 6th moment of $Z(t)$ for Floquet operator. The green curve is $15/6$ times the COE prediction. (Right, b) 8th moment of $Z(t)$ for Floquet operator. The green curve is $105/24$ times the COE prediction.}
    \label{fig:periodic higher order}
\end{figure}

We conclude that up to the 8th moment, the moments are well described by a real Gaussian ansatz for $Z(t)$. The slight deviations in the fit between the scaled COE and periodic boundary conditions may get smaller as the system size increases, and may also simply be small corrections arising between the scaled COE and true ensemble prediction, $S_{\pm}$.

The underlying microscopic situation is further clarified by looking at histograms of the distribution of $Z(t)$. The result for periodic boundary conditions is shown in Figure~\ref{fig:periodic traces} for even and odd times. We see that the histogram is effectively concentrated along a real line in the complex plane, with the angle this line makes with the real axis taking special discrete values depending on $t$. Thus we see that $ e^{i \theta(t)} Z_{\text{real}}(t)$, with $\theta (t)$ periodic, provides a good model of $Z(t)$ in the case of periodic boundary conditions.  Note that this is consistent with the prediction from Sec. \ref{sec:model}, where the Floquet spectrum was predicted to be real up to a phase of $\phi = e^{i \frac{\pi}{4} L t}$.

\begin{figure}[ht]
    \centering
    \begin{subfigure}[b]{0.23\textwidth}
        \centering
        \includegraphics[width=\textwidth]{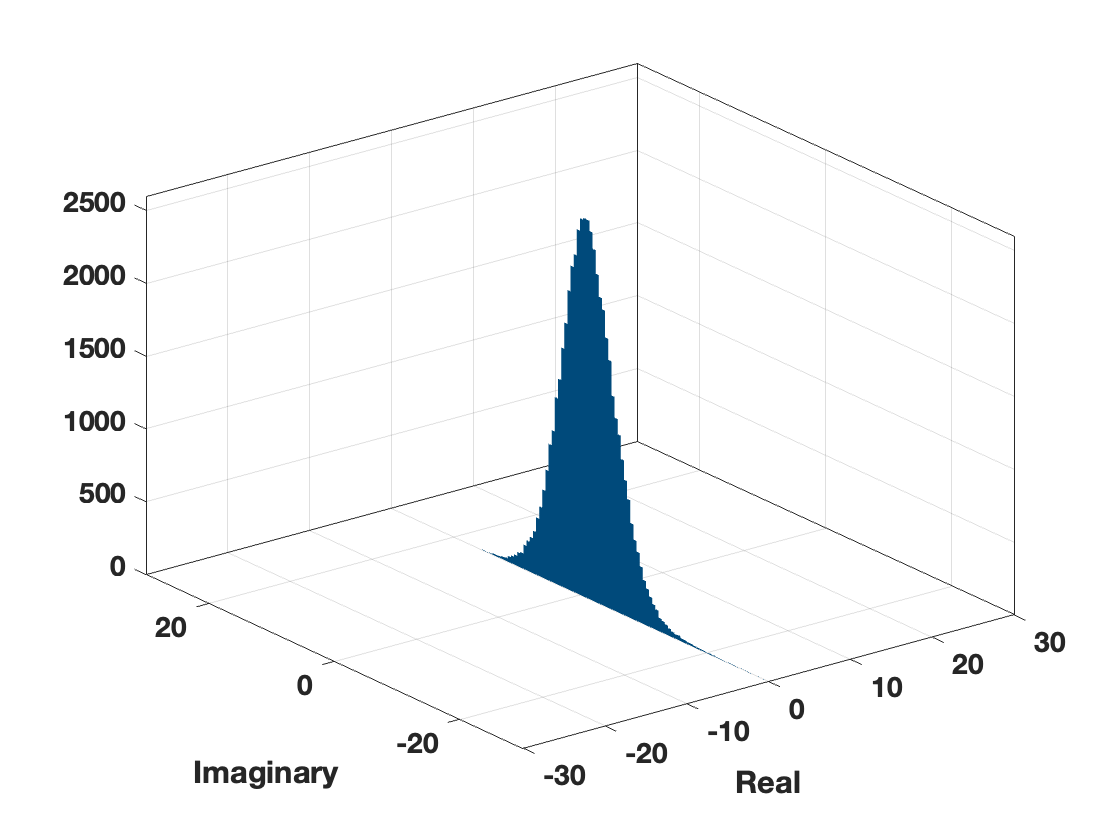}
        \subcaption{$t=10$}
    \end{subfigure}
    \begin{subfigure}[b]{0.23\textwidth}
        \centering
        \includegraphics[width=\textwidth]{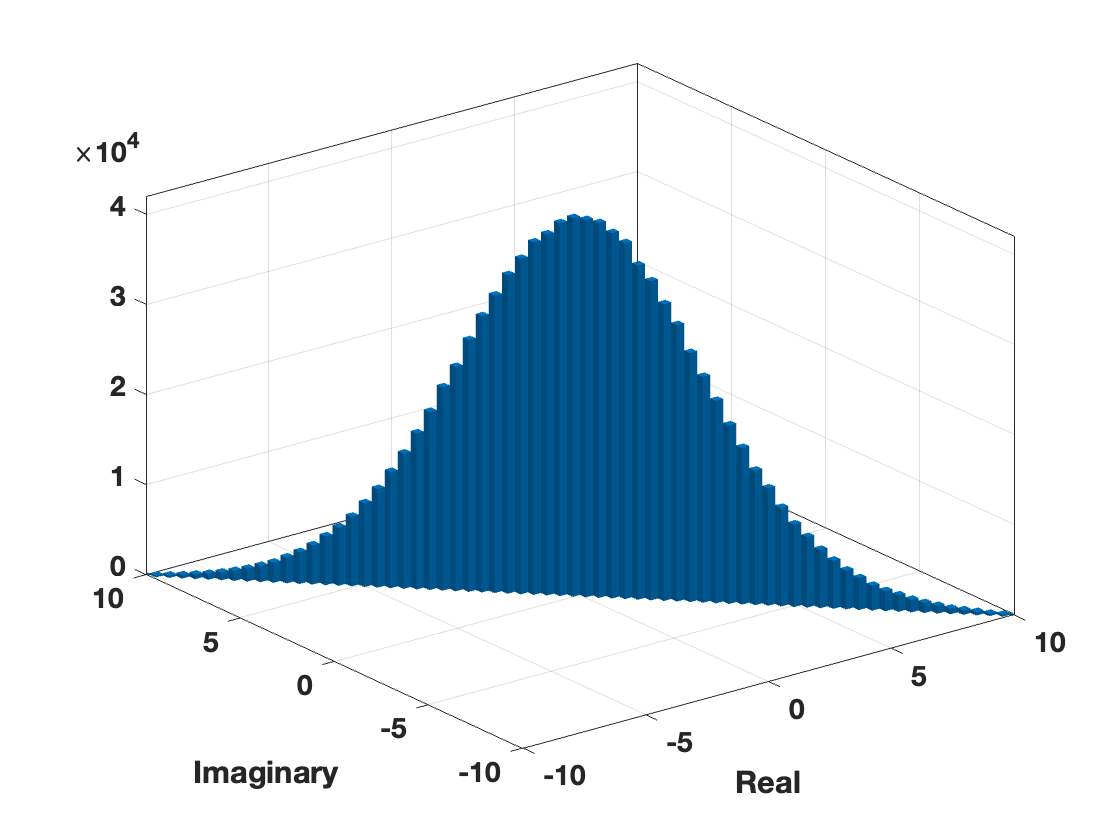}
        \subcaption{$t=11$}
    \end{subfigure}
    \begin{subfigure}[b]{0.23\textwidth}
        \centering
        \includegraphics[width=\textwidth]{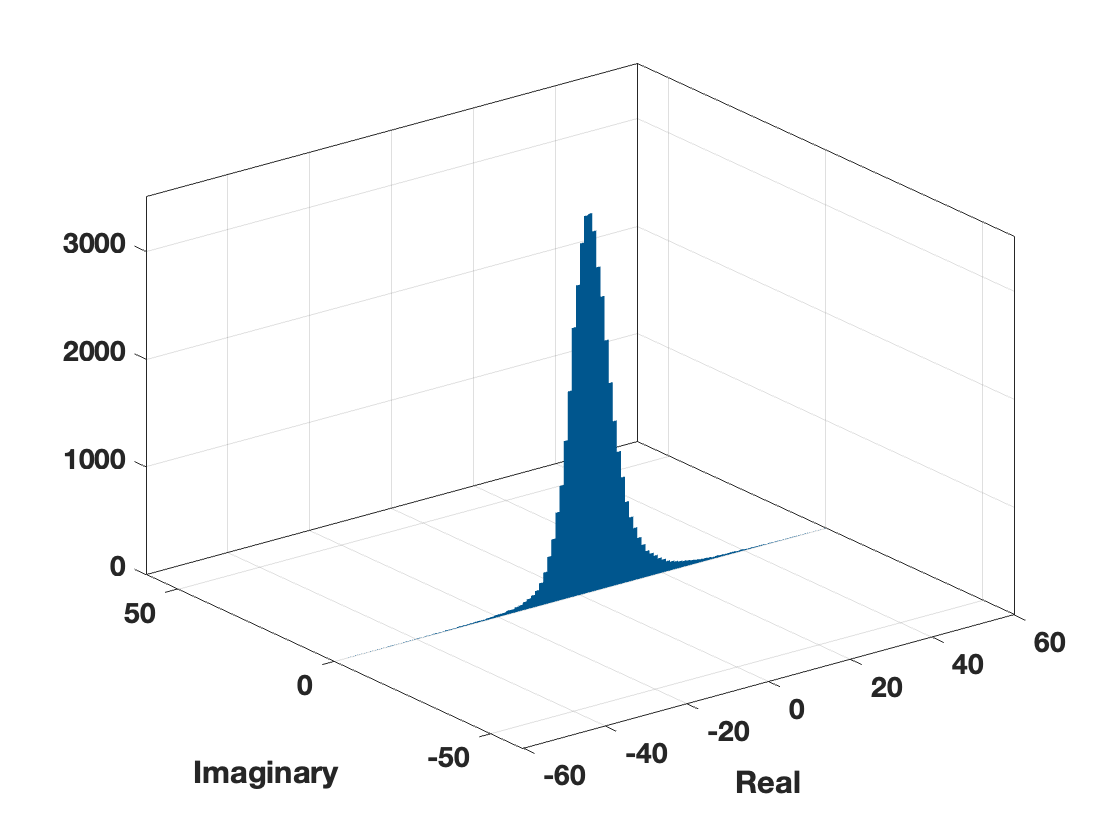}
        \subcaption{$t=12$}
    \end{subfigure}
    \begin{subfigure}[b]{0.23\textwidth}
        \centering
        \includegraphics[width=\textwidth]{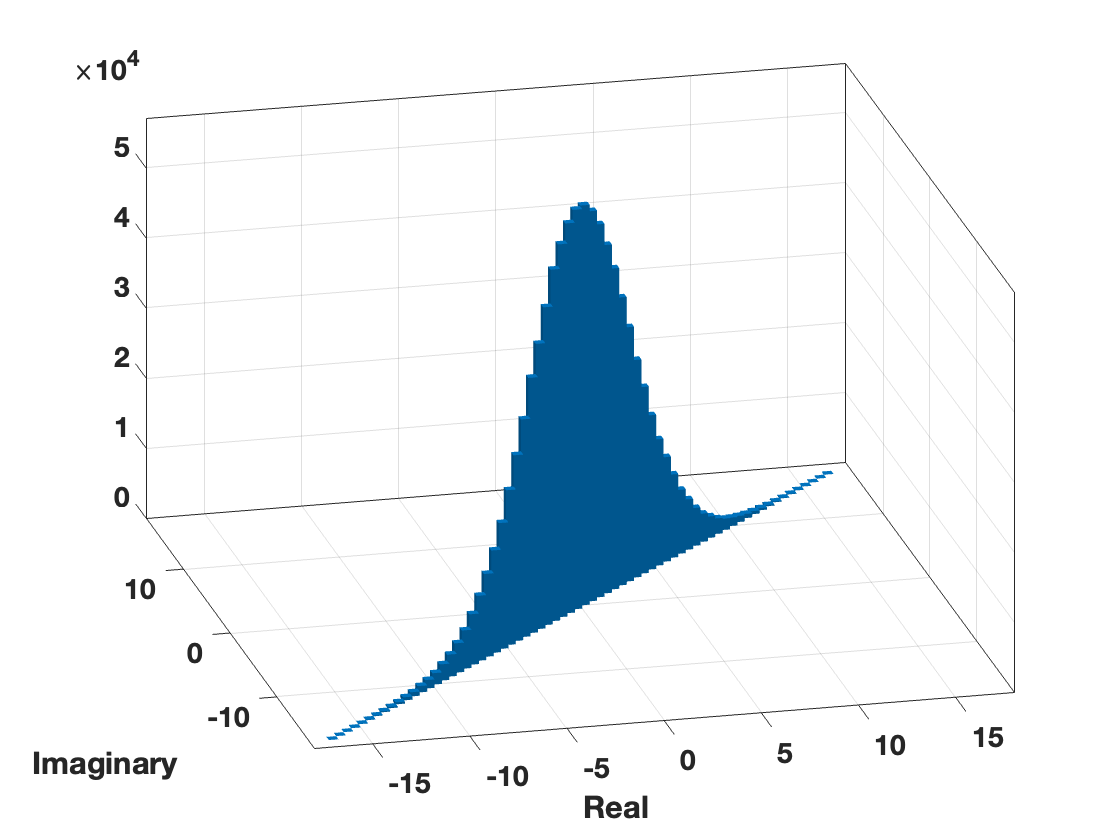}
        \subcaption{$t=13$}
    \end{subfigure}
    \caption{Distribution of $Z(t)$ for periodic boundary conditions.}
    \label{fig:periodic traces}
\end{figure}

\subsection{Transfer matrix}
With numerical evidence behind the claim that $Z(t)$ obeys real random statistics, we now construct and analyze a transfer matrix $\mathbb{T}^{(4)}$ to analytically compute the 4th moment of $Z(t)$. We show that the eigenvectors with unimodular eigenvalues that contributed to the 2nd moment can also be used to construct analogous eigenvectors that contribute to the 4th moment. The naive number of such unimodular eigenvalues is $2 K^2$. We then show that there is another class of eigenvectors that also contribute unimodular eigenvalues to the 4th moment transfer matrix. This enables us to give analytic evidence for our results in the previous subsection (in the TL), including for moments higher than 4th order. We will restrict our analysis to odd time as in \cite{Bertini_2018}.

The construction parallels the analysis reviewed in Section~\ref{sec:model}, now with 4 copies of the 2-dimensional Ising lattice: 2 copies of $U$ and 2 copies of $U^*$. This is to compute the 4th moment of ${\rm tr}\left(U^t_{\rm KI}\right)$,
\begin{equation}
    M_4(t) = \mathbb E_{h}\left[ {\rm tr}\left(\tilde{U}^L_{\rm KI}[h \boldsymbol \epsilon]\right){\rm tr}\left(\tilde{U}^L_{\rm KI}[h \boldsymbol \epsilon]\right)^* {\rm tr}\left(\tilde{U}^L_{\rm KI}[h \boldsymbol \epsilon]\right){\rm tr}\left(\tilde{U}^L_{\rm KI}[h \boldsymbol \epsilon]\right)^* \right].
\end{equation}
There are multiple ways to represent the product of traces as a trace over a tensor product. We choose to work with an alternating pattern, $U \otimes U^* \otimes U \otimes U^*$, which we refer to as a signature and denote $(-,*,-,*)$. Using $\left|{\rm tr} (U)\right|^4 = {\rm tr}(U \otimes U^* \otimes U \otimes U^*)$ we can rewrite the 4th order moment using a transfer matrix $\mathbb{T}^{(4)}$:
\begin{align}
    \mathbb{T}^{(4)} &= \mathbb E_{h}\left[\tilde{U}_{\rm KI}[h \boldsymbol \epsilon] \otimes \tilde{U}_{\rm KI}[h \boldsymbol \epsilon]^* \otimes \tilde{U}_{\rm KI}[h \boldsymbol \epsilon] \otimes \tilde{U}_{\rm KI}[h \boldsymbol \epsilon]^*\right] = (\tilde{U}_{\rm KI} \otimes \tilde{U}_{\rm KI}^* \otimes \tilde{U}_{\rm KI} \otimes \tilde{U}_{\rm KI}^* )\cdot \mathbb{O}^{(4)}_{\sigma} \\
    \mathbb{O}^{(4)}_{\sigma} &= \exp{ \left[-\frac{\sigma^2}{2} \left(M^{(1)}_z - M^{(2)}_z + M^{(3)}_z - M^{(4)}_z\right)^2\right]}
\end{align}
where $\tilde{U}_{\rm KI} \equiv \tilde{U}_{\rm KI}[\bar{h} \boldsymbol \epsilon]$, $M_{\alpha} \equiv \sum_{\tau = 1}^t \sigma^\alpha_\tau$ for $\alpha \in \{a,y,z\}$, and the superscripts on $M$ denote its position in a tensor product with the identity operator $\1$ (i.e. $M_z^{(1)} = M_z \otimes \1 \otimes \1 \otimes \1$). 

Thus we obtain the 4th moment for periodic boundary conditions, 
\begin{equation}
    M_4 = {\rm tr}\left(\mathbb{T}^{{(4)}^L}\right).
\end{equation}
Following the same reasoning as for the 2nd order transfer matrix, the eigenvalues of $\mathbb{T}^{(4)}$ will have magnitude less than or equal to one, so in the thermodynamic limit (TL) $L\to \infty$ the trace will be controlled by the unimodular eigenvalues. Thus we consider $\ket{A^{(4)}}$, a unimodular eigenvector of $\mathbb{T}^{(4)}$ that must satisfy,
\begin{align}
\label{eq:u4 action}
    \left(\tilde{U}_{\rm KI} \otimes \tilde{U}_{\rm KI}^* \otimes \tilde{U}_{\rm KI} \otimes \tilde{U}_{\rm KI}^* \right) \ket{A^{(4)}} &= e^{i\Phi} \ket{A^{(4)}} \\
    \label{eq:m4 action}
    \left(M_z^{(1)} - M_z^{(2)} + M_z^{(3)} - M_z^{(4)}\right) \ket{A^{(4)}} &= 0
\end{align}
Expanding $\ket{A^{(4)}} = \sum A_{n,m,n',m'} \ket{n,m,n',m'}$ as an arbitrary vector in terms of $\{\ket{n}\}$, which form basis vectors of $\mathcal{H}_t$, we can once again reduce the equations above to constraints on the operator $A$. Note that the sum over $\{n,m,n',m'\}$ is implied and $\ket{n,m,n',m'} \equiv \ket{n}\otimes\ket{m}^* \otimes \ket{n'} \otimes \ket{m'}^*$.
\\\\
A complete count of the unimodular eigenvectors $\ket{A^{(4)}}$ would involve repeating arguments similar to \cite{Bertini_2018} but instead now with a rank 4 tensor $A_{n,m,n'm'}$. However, we can simply use the 2nd order eigenvectors to construct some of the 4th order eigenvectors. This of course only provides a lower bound on the trace and thus the 4th moment; however showing that this lower bound is larger than the COE prediction is sufficient to show that the 4th moment for periodic boundary conditions deviates from it. The numerical data suggests that the lower bound is in fact the true number of eigenvectors, at least at sufficiently large $L$ and $t$; however, we do not have a proof of this.
\\\\
This idea of constructing 4th order eigenvectors from 2nd order ones leverages the simple relationship between the 4th ($M_4$) and 2nd ($M_2 = K$) order moments: $M_4 = 2 K^2$ for a complex random variable and $M_4 = 3K^2$ for a real random variable. The factor of $K^2$ suggests that eigenvectors of $\mathbb{T}^{(4)}$ can be formed by pairing together eigenvectors of $\mathbb{T}$. The COE prediction matches that of a complex random variable, which would suggest there are 2 such pairings. We claim that there in fact (at least) 3 pairings that form unimodular eigenvectors, and thus the 4th moment is (at least) that of a real random variable, as the numerical data suggests.

\begin{Claim}
    \label{eq:claim1}
    $\ket{A^{(4)}}  = \sum  A_{nm}  A'_{n'm'} \ket{n,m,n',m'}$ is a unimodular eigenvector of $\mathbb{T}^{(4)}$, where $A,A'$ are operator representations of unimodular eigenvectors of $\mathbb{T}$.
\end{Claim}
\begin{proof}
    Start by applying the left side of \eqref{eq:u4 action} to the eigenvector,
    \begin{align*}
        (U\otimes U^* \otimes U \otimes U^*)\ket{A^{(4)}} &= \sum_{n,m,n',m'} A_{nm} A'_{n'm'} \ket{Un,U^*m,Un',U^*m'} \\
        & = \sum_{n,m,n',m'} \sum_{k,l,k',l'} A_{nm} A'_{n'm'} U_{kn}U^{*}_{lm}U_{k'n'} U^{*}_{l'm'} \ket{k,l,k',l'} \\
        & = \sum_{n,m,n',m'} \sum_{k,l,k',l'} \left(U_{kn} A_{nm} U^{\dagger}_{ml}\right) \left( U_{k'n'} A'_{n'm'} U^{\dagger}_{m'l'}\right) \ket{k,l,k',l'} \\
        & = \sum_{k,l,k',l'} (U A U^{\dagger})_{kl} ( U A' U^{\dagger})_{k'l'} \ket{k,l,k',l'} \\
        &= \sum_{k,l,k',l'}  (e^{i\phi}A_{kl})( e^{i\phi'} A'_{k'l'}) \ket{k,l,k',l'} =  e^{i\Phi} \ket{A^{(4)}}
    \end{align*}
    where $U \equiv \tilde{U}_{\rm KI}, \Phi \equiv \phi + \phi'$ and to go to the last line we have used that $A,A'$ satisfy \eqref{eq:u2 action}. Thus \eqref{eq:u4 action} is satisfied.

Now apply the left side of \eqref{eq:m4 action} to $\ket{A^{(4)}}$,
\begin{align*}
&\left(M^{(1)} - M^{(2)} + M^{(3)} - M^{(4)}\right) \ket{A^{(4)}} \\
&= \sum_{n,m,n',m'} A_{nm} A'_{n'm'}   \left(\ket{M n,m,n',m'} - \ket{n,M m,n',m'} + \ket{n,m,M n',m'} - \ket{n,m,n',M m'}\right) \\
&= \sum_{n,m,n',m'} \sum_{k,l,k',l'} \left(M_{kn}A_{nl}  - A_{km}M^T_{ml}\right) A'_{n'm'} \ket{k,l,n',m'} + A_{nm} \left(M_{k'n'}A'_{n'l'}  - A'_{k'm'}M^T_{m'l'}\right) \ket{n,m,k',l'} \\
& = \sum_{k,l,k',l'} \left([M,A]_{kl} A'_{k'l'} +  [M,A']_{k'l'} A_{kl} \right) \ket{k,l,k',l'} = 0
\end{align*}
where we drop the subscript $z$ on $M_z$ and since $M_z$ is a real Hermitian matrix, it is symmetric ($T$ is the transpose operation). In the last line we again use that $A,A'$ satisfy \eqref{eq:m2 action}, thus showing that \eqref{eq:m4 action} is satisfied.
\end{proof}

\begin{Claim}
    \label{eq:claim2}
    $\ket{A^{(4)}}  = \sum  A_{nm'}  A'_{n'm} \ket{n,m,n',m'}$ is a unimodular eigenvector of $\mathbb{T}^{(4)}$, where $A,A'$ are operator representations of unimodular eigenvectors of $\mathbb{T}$.
\end{Claim}
\begin{proof}
    We apply \eqref{eq:u4 action} to the eigenvector just as in the previous claim, but with $n$ "paired" with $m'$ and $n'$ with $m'$,
    \begin{align*}
        (U\otimes U^* \otimes U \otimes U^*)\ket{A^{(4)}} &= \sum_{n,m,n',m'} A_{nm'} A'_{n'm} \ket{Un,U^*m,Un',U^*m'} \\
        & = \sum_{k,l,k',l'} (U A U^{\dagger})_{kl'} (U A' U^{\dagger})_{k'l} \ket{k,l,k',l'} \\
        &= \sum_{k,l,k',l'}  (e^{i\phi}A_{kl'})( e^{i\phi'} A'_{k'l}) \ket{k,l,k',l'} =  e^{i\Phi} \ket{A^{(4)}}
    \end{align*}
Applying \eqref{eq:m4 action} to $\ket{A^{(4)}}$ is similar,
\begin{align*}
\left(M^{(1)} - M^{(2)} + M^{(3)} - M^{(4)}\right) \ket{A^{(4)}} = \sum_{k,l,k',l'} \left([M,A]_{kl'} A'_{k'l} +  [M,A']_{k'l} A_{kl'} \right) \ket{k,l,k',l'} = 0
\end{align*}
Thus both \eqref{eq:u4 action} and \eqref{eq:m4 action} are satisfied.
\end{proof}

The 3rd pairing is slightly different. First we recall the form of the dual transfer matrix $\tilde{U}_{\rm KI}[\bar{h} \boldsymbol \epsilon] = \tilde{U}_K \tilde{U}_I [\bar{h} \boldsymbol \epsilon]$, where $\tilde{U}_I [\bar{h} \boldsymbol \epsilon] = \exp{\left(-i J' \sum \limits_{\tau = 1}^t \sigma_\tau^z \sigma_{\tau + 1}^z \right)} \exp{\left(-i \bar{h} M_z \right)}$ and $\tilde{U}_K = \exp{\left(-i b M_x \right)}$ are the (dual) interaction and kick terms. We also define $Q \equiv \prod \limits_{\tau = 1}^t \sigma_\tau^z$. Observe that $Q$ is diagonal and $Q_{-n,-n} = Q_{n,n} (-1)^t$. Furthermore we need the following lemma,
\begin{Lemma}
    \label{lemma: kick flip}
    $ \left(\tilde{U}_K \right)_{k,n} = \left(\tilde{U}_K \right)_{-k,-n}$
\end{Lemma}
\begin{proof}
    Here $\ket{-n}$ is defined by flipping all the spins in the state $\ket{n}$, i.e.  $\bra{n}M_z \ket{n} = - \bra{-n}M_z \ket{-n}$. Observing that $\ket{-n}$ is simply given by rotating $\ket{n}$ by $\pi$ about the x-axis, 
\begin{align*}
        \left(\tilde{U}_K \right)_{-k,-n} &= \bra{-k} \exp{\left(- i b M_x\right)} \ket{-n}\\ 
&= \bra{k} \exp{\left(- i \frac{\pi}{2} M_x\right)} \exp{\left( -i b M_x\right)} \exp{\left( i \frac{\pi}{2} M_x\right)} \ket{n} = \left(\tilde{U}_K \right)_{k,n}
\end{align*}
\end{proof}
\begin{Claim}
\label{eq:claim3}
    $\ket{A^{(4)}}  = \sum  (AQ)_{n,-n'}  (QA')_{-m,m'} \ket{n,m,n',m'}$ is a unimodular eigenvector of $\mathbb{T}^{(4)}$, where $A,A'$ are operator representations of unimodular eigenvectors of $\mathbb{T}$.
\end{Claim}
\begin{proof}
    Here we are pairing $n$ with $n'$ and $m$ with $m'$; however since there is no negative sign between the 1st and 3rd $M_z$ (and similarly between the 2nd and 4th), we need the flipped pairings $n$ with $-n'$, where $\bra{-k'}M_z \ket{-n'} = - \bra{k'}M_z \ket{n'}$ i.e. all spins from the $\ket{n'}$ state are flipped. Applying \eqref{eq:m4 action} to the eigenvector is straightforward,
\begin{align*}
&\left(M^{(1)} - M^{(2)} + M^{(3)} - M^{(4)}\right) \ket{A^{(4)}} \\
&= \sum_{n,m,n',m'} (AQ)_{n,-n'} (QA')_{-m,m'}   \left(\ket{M n,m,n',m'} - \ket{n,M m,n',m'} + \ket{n,m,M n',m'} - \ket{n,m,n',M m'}\right) \\
&= \sum_{n,m,n',m'} \sum_{k,l,k',l'}  \left(M_{kn}(AQ)_{n,-k'}  + (AQ)_{k,-n'} M^T_{n'k'}\right) (QA')_{-m,m'} \ket{k,m,k',m'} \\
& - (AQ)_{n,-n'}\left(M_{lm}(QA')_{-m,l'}  + (QA')_{-l,m'}M^T_{m'l'}\right) \ket{n,l,n',l'} \\
& = \sum_{k,l,k',l'} \left([M,AQ]_{k,-k'}(QA')_{-l,l'} +  [M,QA']_{-l,l'} (AQ)_{k,-k'} \right) \ket{k,l,k',l'} = 0
\end{align*}
where in going to the last line, we used $M_{n'k'} = M^T_{n'k'} = -M_{-n',-k'}$. The equality in the last line follows from the fact that $[Q,M_z] = 0$, which reduces the commutators to $[M,A], [M,A']$ which vanish by \eqref{eq:m2 action}.
\\\\
Applying \eqref{eq:u4 action} to $\ket{A^{(4)}}$ is more complicated however,
\begin{align}
        (U\otimes U^* \otimes U \otimes U^*)\ket{A^{(4)}} &= \sum_{n,m,n',m'} (AQ)_{n,-n'} (QA')_{-m,m'} \ket{Un,U^*m,Un',U^*m'} \nonumber \\
        & = \sum_{n,m,n',m'} \sum_{k,l,k',l'} \left(U_{kn} (AQ)_{n,-n'} U^{T}_{n'k'}\right) \left( U^*_{lm} (QA')_{-m,m'} U^{\dagger}_{m'l'}\right) \ket{k,l,k',l'}
        \label{eq:claim 3 u action}
    \end{align}
As long as $U_{kn} (AQ)_{n,-n'} U^T_{n',k'} = (AQ)_{k,-k'} e^{i \phi}$ and $U^*_{lm} (QA')_{-m,m'} U^\dagger_{m'l'} = (Q A')_{-l,l'} e^{i \phi'}$, we obtain an eigenvector with unimodular eigenvalue.

Let's work on the first of these. The important part is,
\begin{equation*}
    (AQ)_{n,-n'} U^T_{n',k'}.
\end{equation*}
We want,
\begin{equation*}
    A_{n,-n'} U^\dagger_{-n',-k'} Q_{-k',-k'} \propto A_{n,-n'} Q_{-n',-n'} U^T_{n',k'}
\end{equation*}
or,
\begin{equation}
    \label{eq: A sub}
    U^\dagger_{-n',-k'} Q_{-k',-k'} \propto Q_{-n',-n'} U^T_{n',k'}
\end{equation}
where we have used that $Q$ is diagonal.

The interaction term in $U$ will simply pull out eigenvalues on both sides, which we will deal with later. Focusing on the kick (which we denote without tilde for simplicity), we have,
\begin{equation*}
    (U_K)^\dagger_{-n',-k'} \stackrel{?}{=}  Q_{-n',-n'} (U_K^T)_{n',k'} Q_{-k',-k'} = (Q (U_{K})^T Q)_{n',k'}
\end{equation*}
which reduces to $U_K^\dagger = Q U_K^T Q$ by Lemma \ref{lemma: kick flip}. We also know that $U_K$ is symmetric, so at last we have,
\begin{equation*}
    e^{i \frac{\pi}{4} M_x} \stackrel{?}{=} Q e^{-i \frac{\pi}{4} M_x} Q
\end{equation*}
which follows from $Q M_x Q=-M_x$. 
\\\\
We can follow a similar reasoning for the 2nd term. Here we want $U^*_{lm} (QA')_{-m,m'}$ to be proportional to $(QUA')_{-l,m'}$,
\begin{equation}
    \label{eq: A' sub}
    Q_{-l,-l}U_{-l,-m} A'_{-m,m'} \propto U^*_{lm} Q_{-m,-m} A'_{-m,m'}
\end{equation}
or
\begin{equation*}
     U^*_{lm} \propto Q_{-l,-l}U_{-l,-m}  Q_{-m,-m}.
\end{equation*}
Once again we must only worry about the kick, which is symmetric and thus $U_K^* = U_K^{\dagger}$ which reduces the condition to,
\begin{equation*}
    (U^{\dagger}_{K})_{lm} \stackrel{?}{=} (Q U_K Q)_{-l,-m}.
\end{equation*}
By Lemma \ref{lemma: kick flip} we need $U^{\dagger}_K = Q U_K Q$, which we already showed.
\\\\
Now we must only deal with the eigenvalue of the interaction term when acting on the states. In particular, we need to compare the eigenvalues of $U^*_I \ket{n}$ and $U_I \ket{-n}$. Expanding these terms at the self dual points given by $J' = \mp \pi/4 , b = \pm \pi/4$ we have,
\begin{align*}
    U^*_I \ket{n} &= \exp{\left(\mp i \frac{\pi}{4} \sum_{\tau = 1}^t \sigma_\tau^z \sigma_{\tau+1}^z\right)}\exp{\left(i \bar{h} M_z\right)} \ket{n} \\
    U_I\ket{-n} &= \exp{\left(\pm i \frac{\pi}{4} \sum_{\tau = 1}^t \sigma_\tau^z \sigma_{\tau+1}^z\right)}\exp{\left(- i \bar{h} M_z\right)} \ket{-n}
\end{align*}
Eigenvalues of the 2nd exponential are the same for both $U^*, U$ and we must only worry about the eigenvalues of the 1st exponential. It is important to recall that the unimodular eigenvectors for the 2nd moment transfer matrix $\mathbb{T}$ in \cite{Bertini_2018} were only valid for odd times $t$, which is what we restrict our analysis to. To evaluate these eigenvalues we start by assuming $\ket{n}$ is the state with all $t$ spins up. Then all terms in the 1st exponential are positive and the eigenvalues of $U^*_I,U_I$ are $e^{\mp it\frac{\pi}{4}}, e^{\pm it\frac{\pi}{4}}$. Every other state can be obtained by flipping spins of the all up state. If we flip a spin that is adjacent to 2 up spins, the exponent in the eigenvalue decreases from $t$ to $t-4$. If we flip a spin adjacent to one and one down spin, the eigenvalue is unchanged. Thus each state $\ket{n}$ can be characterized by the number of relevant spin flips $f_n$ one must apply to the all up state. Importantly, $\ket{-n}$, which flips all spins in the $\ket{n}$ state, is given by performing the same number of flips on the all down state and thus has the same eigenvalue. Thus the eigenvalues of $U^*_I, U_I$ are given by $\gamma^*_n, \gamma_n = e^{\mp i \frac{\pi}{4} (t- 4f_n)}, e^{\pm i \frac{\pi}{4} (t- 4f_n)}$, irrespective of the state being $\ket{n}$ or $\ket{-n}$, and substituting into \eqref{eq: A sub} and \eqref{eq: A' sub} we have,
\begin{align*}
   Q_{-n',-n'} U^T_{n',k'} = \gamma_{n'}^2 U^\dagger_{-n',-k'} Q_{-k',-k'}\\
   U^*_{lm}  Q_{-m,-m} = (\gamma^*_m)^2 Q_{-l,-l}U_{-l,-m}
\end{align*}
where $(\gamma^*_{m})^2 = e^{\mp i t\frac{\pi}{2}} e^{\pm 2\pi f_{m} i}$ and $ \gamma^2_{n'} =  e^{\pm i t\frac{\pi}{2}} e^{\mp 2\pi f_{n'} i}$. Since $f_{n'}, f_m$ are integers, the second exponential is 1 in both cases and the resulting expressions are independent of index. We can use this to simplify \eqref{eq:claim 3 u action}, where we substitute the expressions above,
\begin{align*}
    (U\otimes U^* \otimes U \otimes U^*)\ket{A^{(4)}} &= \sum_{n,m,n',m'} \sum_{k,l,k',l'} (\gamma_{n'}\gamma_m^*)^2 \left(U_{kn} A_{n,-n'} U^{\dagger}_{-n',-k'} Q_{-k',-k'}\right) \left(Q_{-l,-l} U_{-l,-m} A'_{-m,m'} U^{\dagger}_{m'l'}\right) \ket{k,l,k',l'} \\
&= \sum_{k,l,k',l'} (U A U^{\dagger} Q)_{k,-k'} (Q U A' U^{\dagger})_{-l,l'} \ket{k,l,k',l'} \\ 
&= \sum_{k,l,k',l'} \left(e^{i\phi} (AQ)_{k,-k'}\right) \left(e^{i\phi'} (QA')_{-l,l'}\right) \ket{k,l,k',l'} =  e^{i\Phi} \ket{A^{(4)}}
\end{align*}
Since the eigenvalues have no dependence on an index, they cancel each other and we once again use that $A,A'$ satisfy \eqref{eq:u2 action} to show that \eqref{eq:u4 action} is satisfied.
\end{proof}

Thus we have constructed $3$ classes of unimodular eigenvectors from $3$ ways of pairing the $Z$s and $Z^*$s among themselves. There are 4 copies of the return amplitude: 2 of $Z$ and 2 of $Z^*$. There are $2$ ``natural'' pairings between the $Z$'s and $Z^*$'s, and one non-trivial ``self'' pairing that pairs $Z$ with $Z$ and $Z^*$ with $Z^*$. Since our eigenvectors are formed from vectors of $\mathbb{T}^{(4)}$ with eigenvalue 1, all $\ket{A^{(4)}}$ unimodular eigenvectors of $\mathbb{T}^{(4)}$ also have eigenvalue 1 and the 4th moment in the TL is simply the number of linearly independent $\ket{A^{(4)}}$. Our 3 independent pairings thus show that the 4th moment $K_4$ is (at least) $3K^2$, the expected value for ${\rm tr}\left(U^t_{\rm KI}\right)$ behaving as a real random variable.

Extending this construction to higher order moments is simple. Continuing to work in the alternating signature $(-,*,\cdots,-,*)$, we can construct unimodular eigenvectors of the $(2\ell)$-th order transfer matrix $\mathbb{T}^{(2
\ell)}$ by summing matrix elements of $2n$ eigenvectors of $\mathbb{T}$ over the indices $\{m_1,m_1',\cdots ,m_\ell, m'_\ell\}$, where the primed indices denote complex conjugate basis states of $\mathcal{H}_t$. The primed and unprimed indices can be paired with each other as shown in claims 1 and 2 above, while self-pairings between two primed/unprimed indices are to be done as in claim 3. Since there are an equal number of primed and unprimed indices, for each pair of primed indices there must be a pair of unprimed indices, which is required since as shown in claim 3 these pairings have conjugate eigenvalues that cancel each other. Thus, each set of pairings contributes $K^n$ to the moment (as higher order eigenvectors constructed by pairing 2nd order eigenvectors also have eigenvalue 1), and since any 2 indices can be paired we get $M_{2\ell} \geq (2\ell-1)!!K^\ell$, i.e. the higher order moments of a real Gaussian.

This provides an alternate computation of the lower bound computed in \cite{Flack_2020}. It should be noted that in \cite{Flack_2020}, the pairings are among operators that span the dihedral group $\mathcal{G}_t$, which are in a one-to-one correspondence with the unimodular eigenvectors $\ket{A}$. Thus the procedure outlined above follows the same essential idea as \cite{Flack_2020}, albeit being more explicit with the steps. The explicit nature of our procedure sheds some light on the nature of the self-dual point (see Sec. \ref{sec:discussion}) and helps with computations related to the Loschmidt SFF (see Sec. \ref{sec:loschmidt}).

\section{Open boundary conditions}
\label{sec:open}
\subsection{Numerics}
Much like the periodic case, we start our discussion of open boundary conditions with numerical evidence. In line with the data in the previous section, all plots are generated by averaging over $N = 10^6$ samples of \eqref{eq:floquet} defined by $\bar{h} = 0.6, \sigma = 10\pi$ for a $7$-qubit chain ($L = 7$) at the self dual point $J = b = \pi/4$, unless otherwise noted. The modification required for open boundary conditions is to change the upper limit of the first sum in \eqref{eq:hamiltonians} for the interaction term $H_{\rm I}$ to $L-1$.

\begin{figure}[ht]
    \centering
    \begin{subfigure}[b]{0.49\textwidth}
        \centering
        \includegraphics[width=\textwidth]{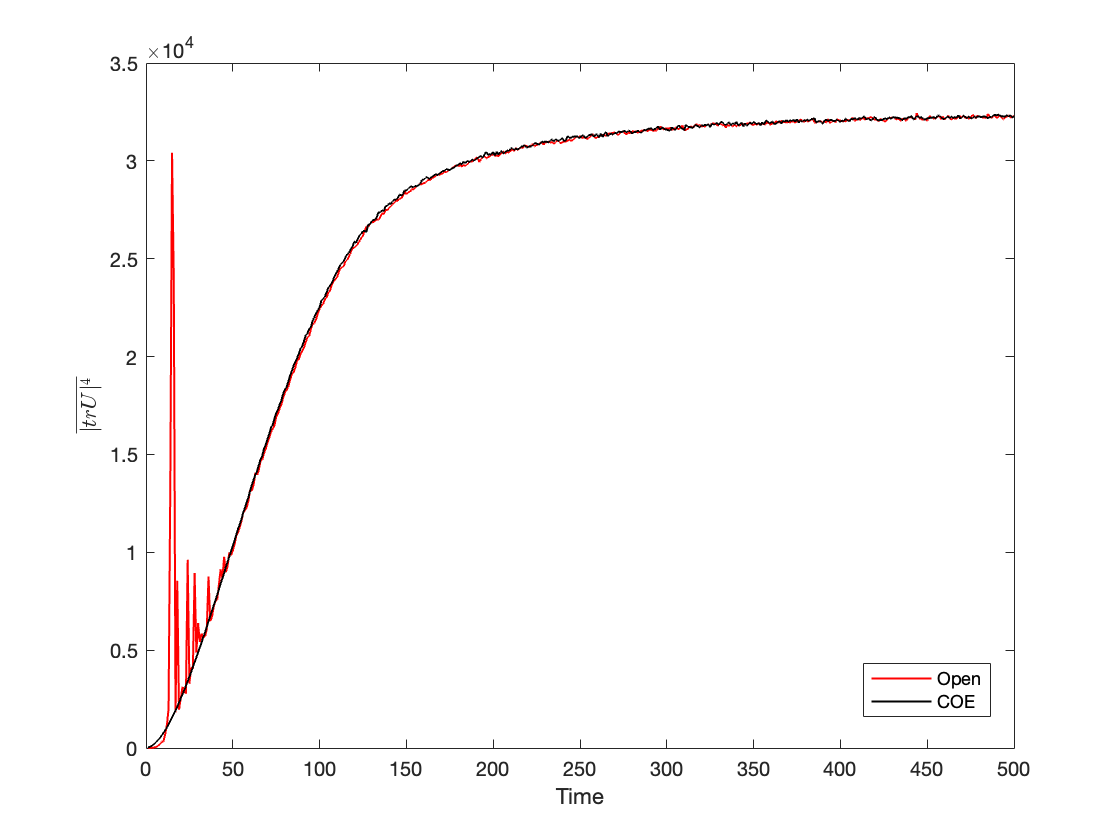}
    \end{subfigure}
    \begin{subfigure}[b]{0.49\textwidth}
        \centering
        \includegraphics[width=\textwidth]{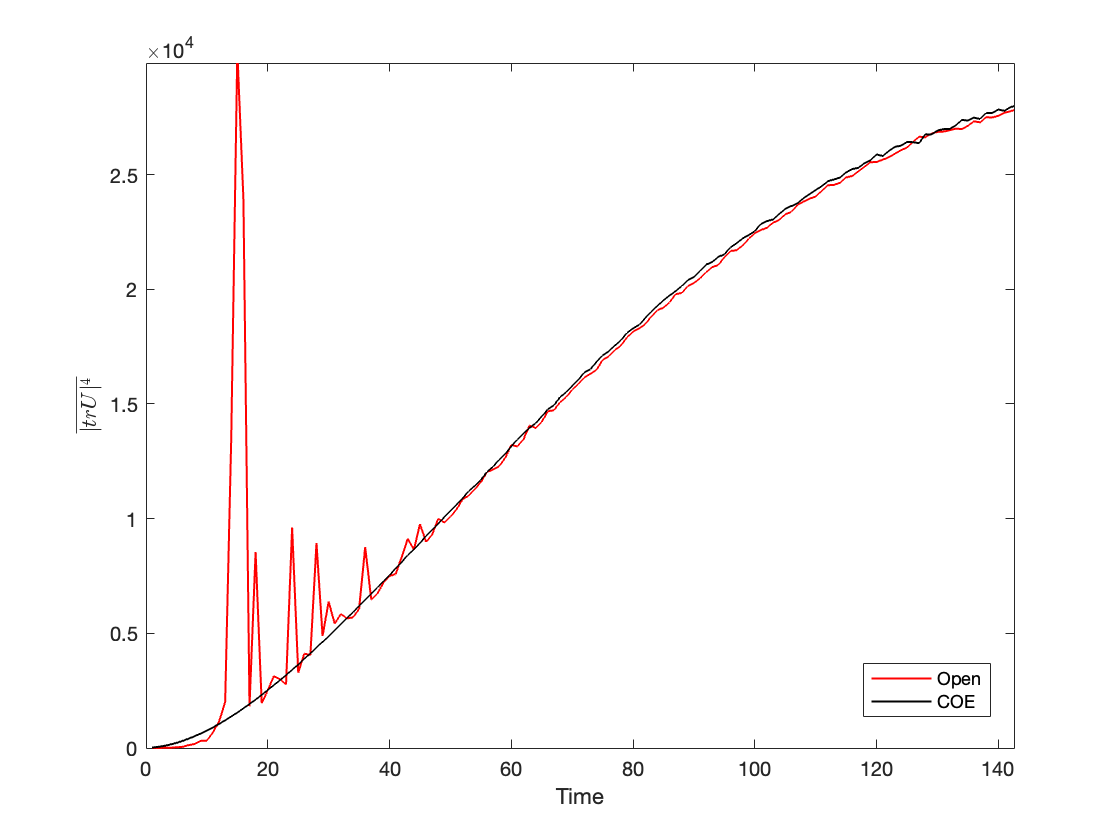}
    \end{subfigure}
    \caption{(Left) 4th moment of $Z(t)$ for the open boundary condition Floquet operator. (Right) Zooming in on the ramp region. The black curve is the COE prediction.}
    \label{fig:open o4}
\end{figure}

We show the data for the 4th moment in Figure~\ref{fig:open o4}. The 4th moment agrees almost perfectly with the COE prediction, which leads us to hypothesize that $Z(t)$ behaves like a complex Gaussian random variable for open boundary conditions. Thus we predict the higher moments for open boundary conditions to line up with the COE prediction, i.e. $M_{2\ell}/M_{2}^\ell = \ell!$.

We subject this to additional checks by looking at higher moments, as shown in Figure~\ref{fig:open higher order}. We see that open boundary conditions agree with the COE prediction, which is close to $6 K^3$ (6th moment) and $24 K^4$ (8th moment), i.e. the moments for a complex Gaussian variable.

\begin{figure}[ht]
    \centering
    \begin{subfigure}[b]{0.49\textwidth}
        \centering
        \includegraphics[width=\textwidth]{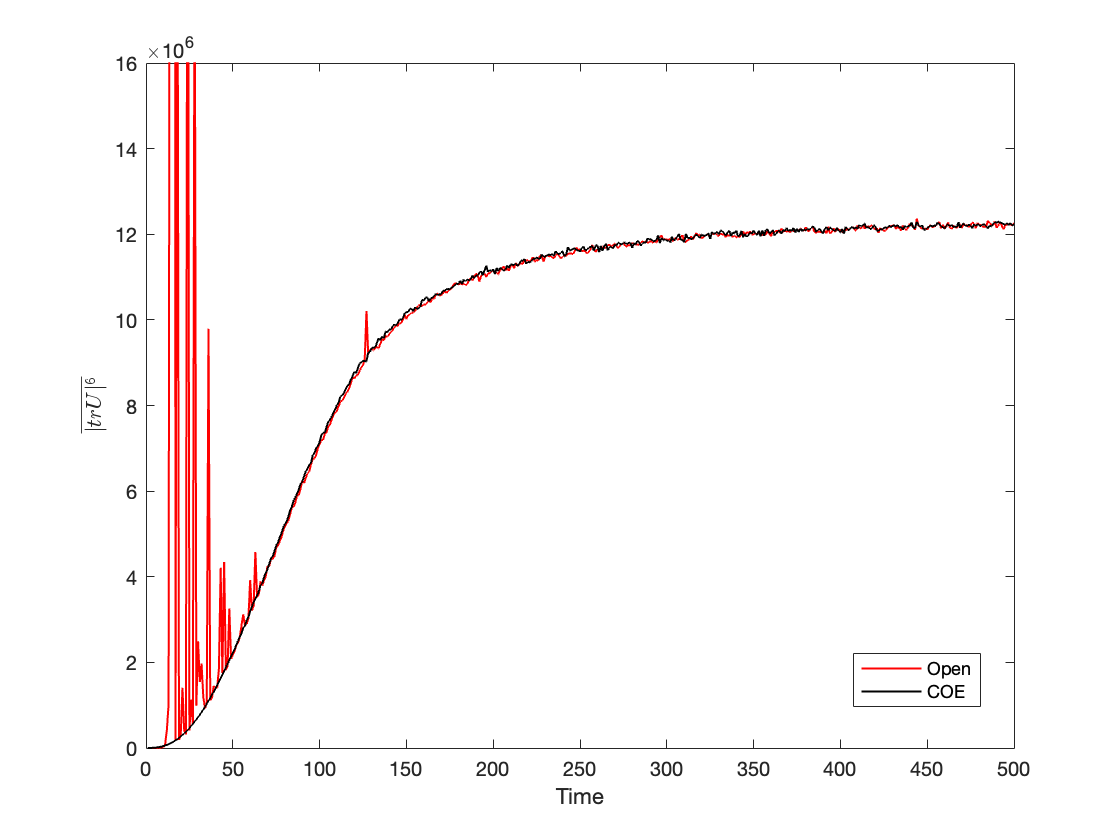}
        \subcaption{6th moment}
    \end{subfigure}
    \begin{subfigure}[b]{0.49\textwidth}
        \centering
        \includegraphics[width=\textwidth]{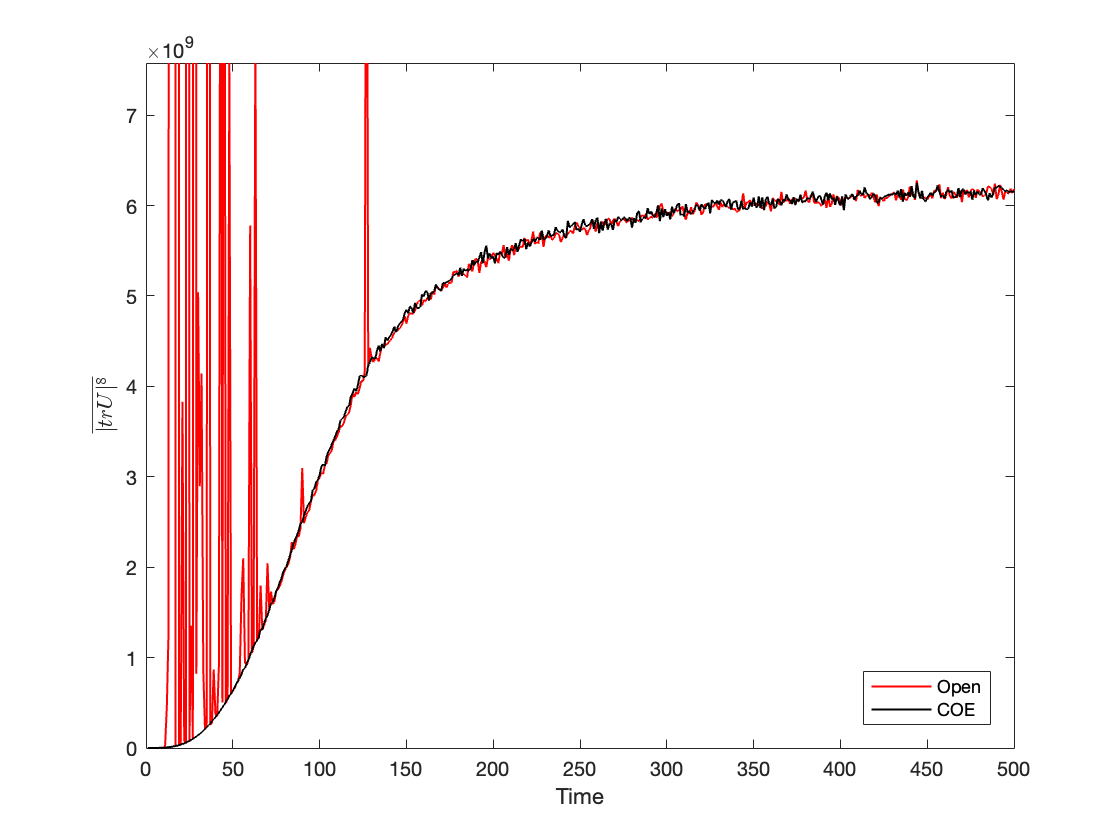}
        \subcaption{8th moment}
    \end{subfigure}
    \caption{(Left, a) 6th moment of $Z(t)$ for Floquet (red) compared with COE (red). (Right, b) 8th moment of $Z(t)$ for Floquet (red) compared with COE (black).}
    \label{fig:open higher order}
\end{figure}

Similar to the situation for periodic boundary conditions, we consider histograms of the magnitude and phase of $Z(t)$. The open boundary condition case shown in Figure~\ref{fig:open traces} shows a full 2d histogram in the complex plane. Hence, $Z(t)$ has a (near) normal distribution over the complex plane in the open case, suggesting it behaves as a complex Gaussian.

\begin{figure}[ht]
    \centering
    \begin{subfigure}[b]{0.23\textwidth}
        \centering
        \includegraphics[width=\textwidth]{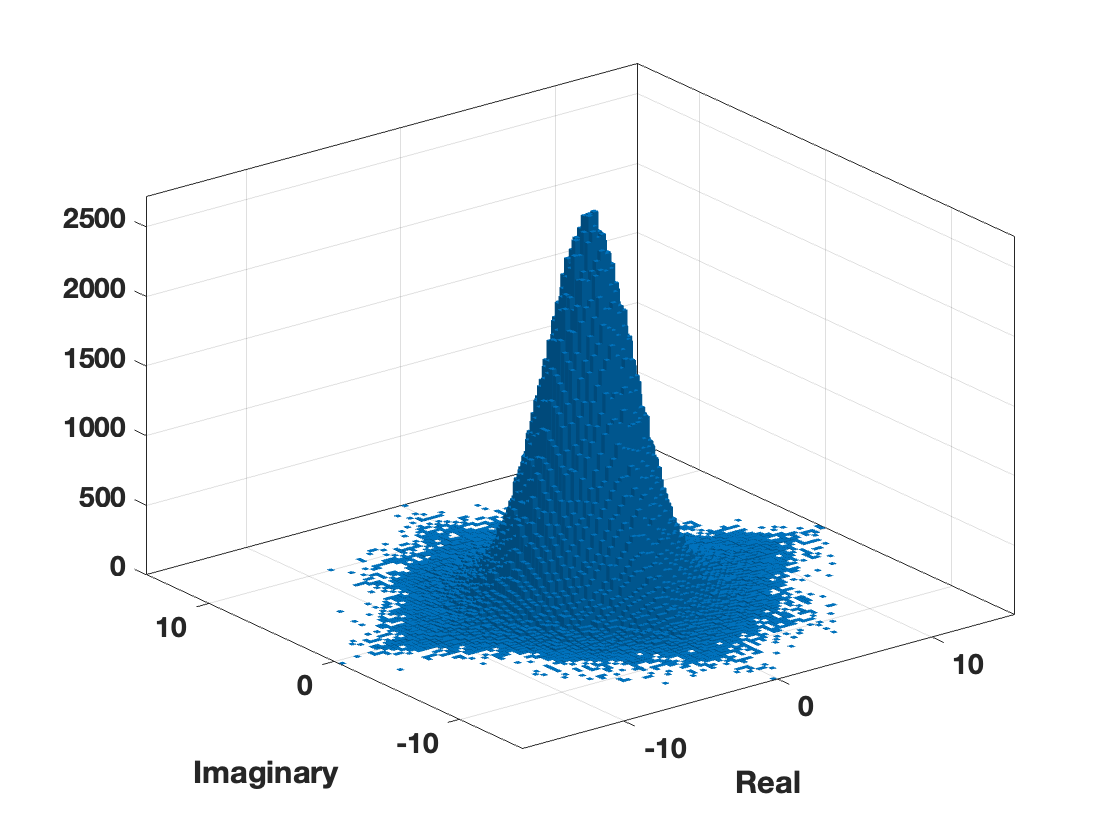}
        \subcaption{$t=10$}
    \end{subfigure}
    \begin{subfigure}[b]{0.23\textwidth}
        \centering
        \includegraphics[width=\textwidth]{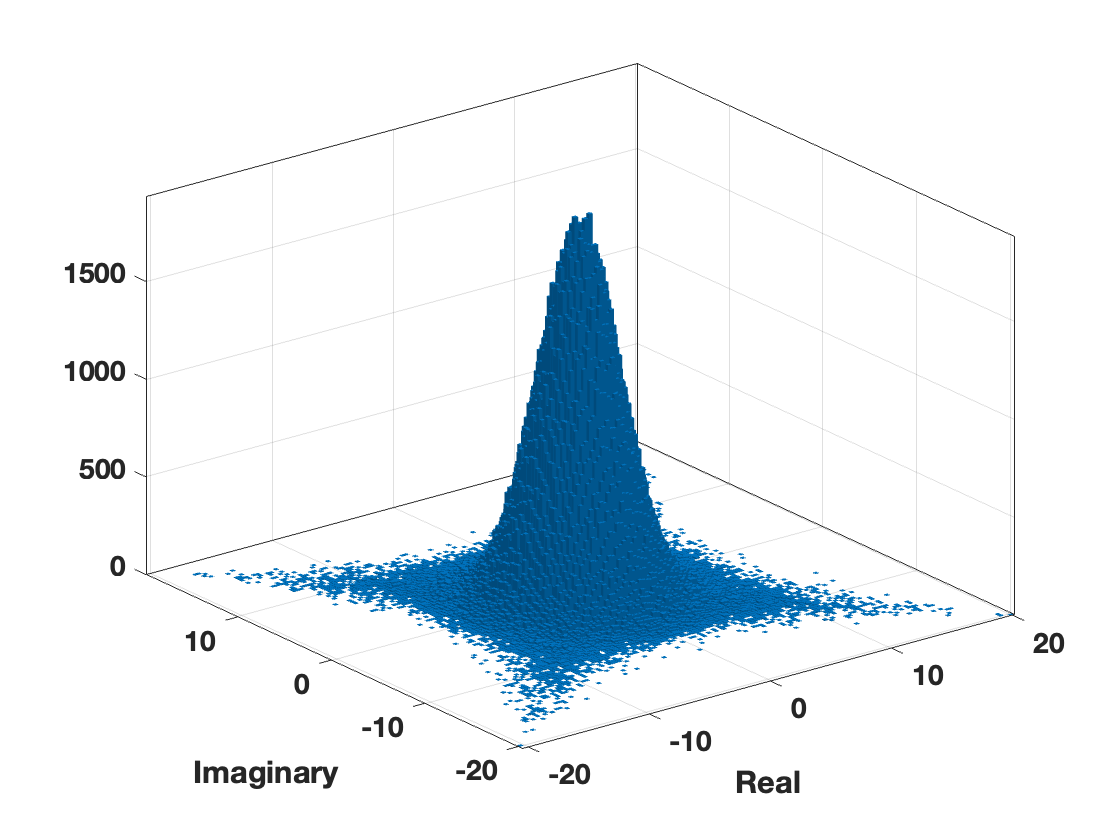}
        \subcaption{$t=11$}
    \end{subfigure}
    \begin{subfigure}[b]{0.23\textwidth}
        \centering
        \includegraphics[width=\textwidth]{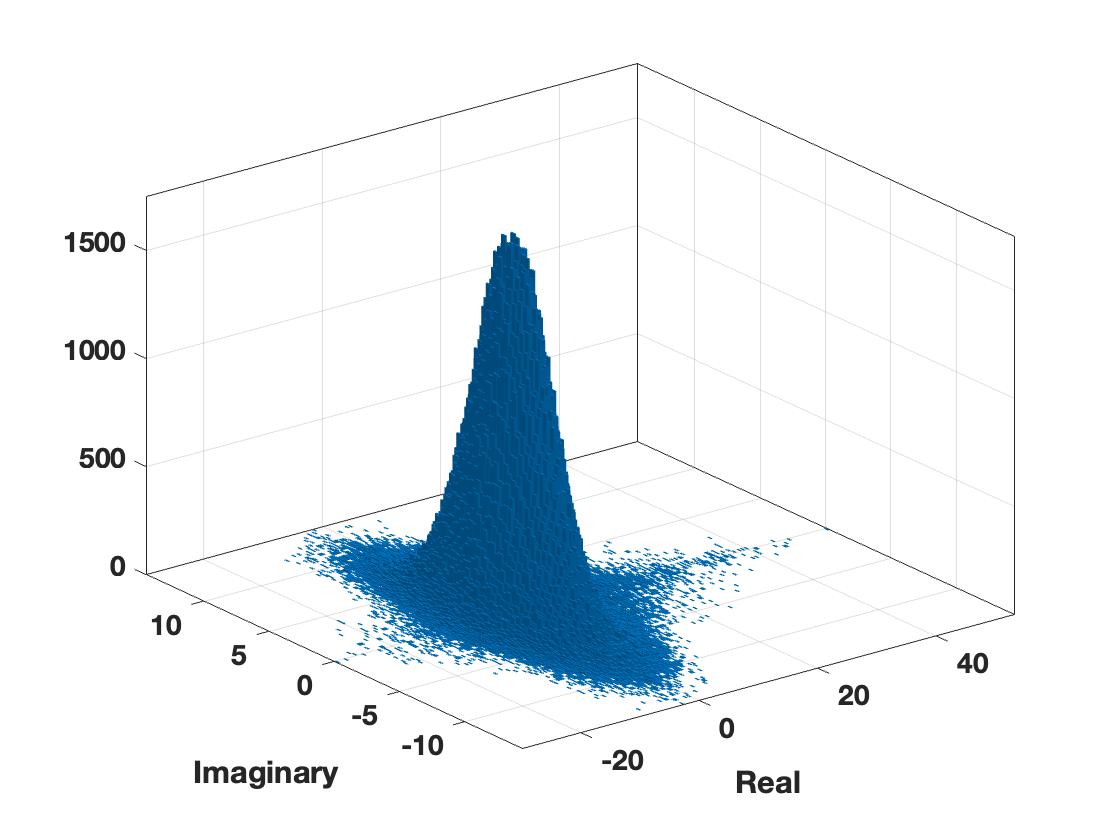}
        \subcaption{$t=12$}
    \end{subfigure}
    \begin{subfigure}[b]{0.23\textwidth}
        \centering
        \includegraphics[width=\textwidth]{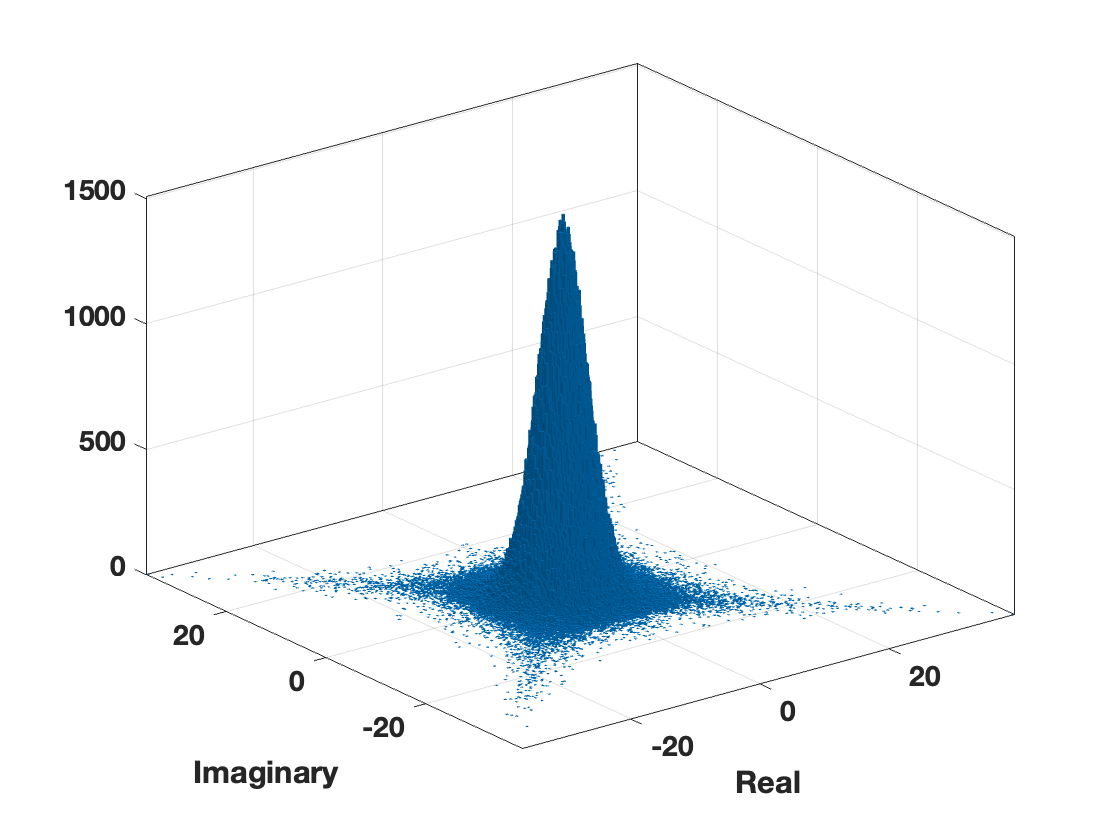}
        \subcaption{$t=13$}
    \end{subfigure}
    \caption{Distribution of $Z(t)$ for open boundary conditions.}
    \label{fig:open traces}
\end{figure}

\subsection{Boundary vectors}

Analytic verification of the results shown above is challenging. Symmetry arguments similar to those applied to periodic boundary conditions do not work for open boundary conditions. The Floquet operator in both situations possesses the time-reversal symmetry; however the `conjugate symmetry' \eqref{eq:conjugate symmetry} is only present for periodic boundary conditions. This, however, is not sufficient to conclude that the relevant ensemble for open boundary conditions is COE and thus the relevant statistics are complex Gaussian, as this does not rule out additional symmetries that may further constrain the ensemble.

Alternatively, an analytic computation of the moments for open boundary conditions can be conceived of in a similar fashion to the periodic case. We can consider modifying the transfer matrix approach by introducing boundary vectors. As the trace of the Floquet operator can be rewritten as the partition function of 2 dimensional Ising model on a $t \times L$ lattice, we may use the same modification as is used for the open boundary Ising model. Breaking the periodicity in $L$ (but not $t$) makes it so the partition function is no longer the trace of the temporal transfer matrix $\mathbb{T}$, but instead a matrix element defined by an inner product with appropriately defined boundary vectors $\bra{v_1} \mathbb{T}^{L-1} \ket{v_L}$. 

Thus, starting with the modified 2-dimensional partition function, for open boundary conditions we have
\begin{equation}
    {\rm tr}\left(U^t_{\rm KI}[\boldsymbol h]\right) = \left[\frac{\sin{(2b)}}{2i}\right]^{Lt/2} \sum_{{s_{\tau,j}}} e^{-i\mathcal{E}
    \left[\{s_{\tau,j}\}, \boldsymbol h\right]}
\end{equation}
as before, but now with the energy function given by
\begin{equation}
    \mathcal{E}\left[\{s_{\tau, j}, \boldsymbol h\}\right] = \sum_{\tau = 1}^t \left(J \sum_{j=1}^{L-1} s_{\tau, j} s_{\tau, j+1} + J' \sum_{j=1}^L s_{\tau, j} s_{\tau + 1, j} + \sum_{j=1}^L  h_j s_{\tau,j} \right)
\end{equation}
where we have changed the upper bound of the first sum. We can rewrite this using the dual (temporal) Floquet operator,
\begin{align}
\label{eq:open Z}
    {\rm tr}\left(U^t_{\rm KI}[\boldsymbol h]\right) &= \sum_{\{\boldsymbol{s}_j\}} \bra{\boldsymbol{s}_L} B[h_L\boldsymbol{\epsilon}] \ket{\boldsymbol{s}_L} \prod_{j = 1}^{L-1} \bra{\boldsymbol{s}_{j+1}} \tilde{U}_{\rm KI} [h_j\boldsymbol{\epsilon}] \ket{\boldsymbol{s}_j}
\end{align}
where we sum over the spin configurations $\ket{\boldsymbol{s}_j}$ defined by $\sigma_{\tau}^z \ket{\boldsymbol{s}_j} = s_{\tau,j} \ket{\boldsymbol{s}_j}$ and we define the additional operator,
\begin{equation}
    B[h_j\boldsymbol{\epsilon}] = \left(\frac{\sin{(2b)}}{2i}\right)^{t/2} e^{-i \tilde{H}_{\rm I}[h_j\boldsymbol{\epsilon}]}, \quad{} \tilde{H}_{\rm I}[h_j\boldsymbol{\epsilon}] = J' \sum_{\tau = 1}^t \sigma_\tau^z \sigma_{\tau + 1}^z + \sum_{\tau = 1}^t h_j \sigma_{\tau}^z.
\end{equation}
Here, $\tilde{H}_{\rm I} [h_j\boldsymbol{\epsilon}]$ is the dual interaction Hamiltonian evaluated on the magnetic field at the $j$-th spin.

Using \eqref{eq:open Z}, we can write the SFF,
\begin{align}
    K_{\rm open} (t) &= \mathbb{E}_{h} \left[{\rm tr}\left(U^t_{\rm KI}[\boldsymbol h]\right) {\rm tr}\left(U^t_{\rm KI}[\boldsymbol h]\right) ^* \right] \\
    &= \mathbb{E}_{h} \left[ \sum_{\{\boldsymbol{s}_1, \boldsymbol{s}_L,\boldsymbol{s}'_1, \boldsymbol{s}'_L\}} \bra{\boldsymbol{s}_L} B[h \boldsymbol{\epsilon}] \ket{\boldsymbol{s}_L} \bra{\boldsymbol{s}'_L} B^*[h \boldsymbol{\epsilon}] \ket{\boldsymbol{s}'_L} \bra{\boldsymbol{s}_L} \tilde{U}_{\rm KI}^{L-1}[h \boldsymbol{\epsilon}] \ket{\boldsymbol{s}_1} \bra{\boldsymbol{s}'_1} (\tilde{U}_{\rm KI}^*)^{L-1}[h \boldsymbol{\epsilon}] \ket{\boldsymbol{s}'_L} \right] \\
    &= \bra{1,1} \mathbb{E}_{h} \left[\left(\tilde{U}_{\rm KI}[h \boldsymbol{\epsilon}] \otimes \tilde{U}_{\rm KI}^* [h \boldsymbol{\epsilon}]\right)^{L-1}\left(B[h\boldsymbol{\epsilon}] \otimes B^*[h\boldsymbol{\epsilon}]\right)\right] \ket{1,1} \\
    &= \bra{1,1} \mathbb{T}^{L-1} \left(B \otimes B^*\right) \cdot\mathbb{O}_\sigma \ket{1,1}
\end{align}
Here we define $\ket{1,1} \equiv \ket{1} \otimes \ket{1}$, where $\ket{1}$ is the vector with all entries equal to 1 in the computational basis. We also use the convention where $B \equiv B[\bar{h} \boldsymbol{\epsilon}]$ and $\mathbb{O}_\sigma$ is the same variance matrix as for the periodic case. We have also picked the self-dual point i.e. $J = b = \pi/4$. $\mathbb{T}$ may not be generically diagonalizable; however \cite{Bertini_2018} shows that the algebraic and geometric multiplicities of any unimodular eigenvalue of $\mathbb{T}$ must coincide. Writing $\mathbb{T}$ in Jordan normal form $P D P^{-1}$, $P$ consists of (generalized) eigenvectors and $D$ is an upper triangular matrix consisting of Jordan blocks. Since the algebraic and geometric multiplicities of unimodular eigenvalues coincide, they all form a diagonal Jordan block with size equal to the number of unimodular eigenvectors (i.e. the periodic SFF, $K(t)$). Similarly, we may write $P = [\ket{A_1}, \cdots \ket{A_{K(t)}}, \cdots]$ in terms of the unimodular eigenvectors $\ket{A_n}$ followed by the remaining eigenvectors. Defining $\ket{v_L} \equiv \ket{1,1}$ and $\ket{v_R} \equiv (B\otimes B^*) \cdot \mathbb{O}_\sigma \ket{1,1}$, we can now expand $\mathbb{T}$ in terms of unimodular eigenvectors $\ket{A}$ to compute,
\begin{align}
\lim_{L\rightarrow \infty} K_{\rm open} (t) = \lim_{L \rightarrow \infty} \bra{v_L} \mathbb{T}^{L-1} \ket{v_R} = \sum_{\ket{A}}^{K(t)}  \braket{v_L|A} \braket{A|v_R}
\end{align}
since only the unimodular eigenvectors contribute in the TL. From this we can see that the SFF for open boundary conditions is computed by the sum over the components of $\ket{A}$ in the computational basis, $\braket{1,1|A}$ and the sum of the components of $\bra{A} (B\otimes B^*)\cdot \mathbb{O}_\sigma$. We were unable to find an analytic expression for these, and numerical data was limited due to the exponential increase in the transfer matrix size with $t$. We expect $\braket{v_L|A}\braket{A|v_R} = 1$ such that the SFF for open boundary conditions matches the periodic SFF.

Extending to the 4th moment, one could imagine a scenario where the moment no longer matches the periodic boundary condition case. Here we have the 4th order transfer matrix $\mathbb{T}^{(4)}$, appropriately defined boundary vectors $\ket{v^{(4)}_L},\ket{v^{(4)}_R}$ and unimodular eigenvectors $\ket{A_n^{(4)}}$ such that
\begin{align*}
    M_4^{(\rm open)}(t) = \lim_{L \rightarrow \infty} \bra{v^{(4)}_L} \left(\mathbb{T}^{(4)}\right)^{L-1} \ket{v^{(4)}_R} = \sum_{n = 1}^{3 K^2(t)}  \braket{v_L^{(4)}|A_n^{(4)}} \braket{A_n^{(4)}|v^{(4)}_R}.
\end{align*}
If suppose only eigenvectors formed by pairings of the kind in claim 1 and 2 had a non-zero inner product with both $\ket{v^{(4)}_L},\ket{v^{(4)}_R}$, as opposed to eigenvectors formed by pairings of the kind in claim 3 whose inner product with one (or both) boundary vectors is 0, then the 4th moment would match the COE prediction $M_4 (t) = 2 K^2 (t)$. We suspect some cancellation of this kind also holds for higher order moments, with eigenvectors formed by pairings of primed/unprimed indices with themselves having an inner product of $0$ with the boundary vectors. This would be a possible explanation for the open boundary condition moments matching the COE results, which could be investigated in future work.

\section{Loschmidt echo SFF}
\label{sec:loschmidt}
Here we consider a generalization of the SFF called the Loschmidt SFF (LSFF)~\cite{Winer_2022}. The SFF, given by an ensemble average of $|Z(t)|^2$, considers evolving the system forward and backwards in time by the same Hamiltonian $H$, which determines the return amplitude $Z$. To define the LSFF, we instead consider 2 (correlated) Hamiltonians $H_1, H_2$ which each define a return amplitude $Z_1, Z_2$. This defines the LSFF,
\begin{equation}
    K_{\rm E} (t) = \mathbb{E}[Z_1(t) Z_2^*(t)],
\end{equation}
where the subscript E denotes echo. The Loschmidt SFF has its power when the Hamiltonians are closely correlated, and thus often has the interpretation of evolving the system forward in time by a Hamiltonian $H_1$ and then applying a close, but imperfect time reversal procedure given by $H_2 = H_1 + \delta H_0$, where $\delta \ll 1$ is a perturbation parameter for some perturbation Hamiltonian $H_0$.

For the kicked Ising model, one may consider Hamiltonians $H_1, H_2$ defined by \eqref{eq:ham} and determined by longitudinal fields $\boldsymbol{h}_1,\boldsymbol{h}_2$ drawn from correlated normal distributions, with the same mean and variance $\bar{h}, \sigma^2$, and correlation $C$. Thus we have for the kicked Ising model,
\begin{equation}
    K_{\rm E} (t) = \mathbb{E}_{\boldsymbol{h}_1,\boldsymbol{h}_2} \left[{\rm tr}\left(U^t_{\rm KI}[\boldsymbol{h}_1]\right){\rm tr}\left(U^t_{\rm KI}[\boldsymbol{h}_2]\right)^*\right]
\end{equation}
Following \cite{Bertini_2018} and the previous discussion, one may be interested in this quantity in the thermodynamic limit ($L \rightarrow\infty$). Repeating the transfer matrix procedure, we can rewrite
\begin{equation}
    K_{\rm E} (t) = \rm{tr} \left(\mathbb{T}_E^L\right)    
\end{equation}
where,
\begin{align}
    \mathbb{T}_{\rm E} &= \mathbb E_{h_1,h_2}\left[\tilde{U}_{\rm KI}[h_1 \boldsymbol \epsilon] \otimes \tilde{U}_{\rm KI}[h_2 \boldsymbol \epsilon]^*\right] = (\tilde{U}_{\rm KI} \otimes \tilde{U}_{\rm KI}^*)\cdot \mathbb{O}_{\sigma,C} \\
    \mathbb{O}_{\sigma,C} &= \exp{\left[-\frac{\sigma^2}{2} \left\{\left(M_z^{(1)}\right)^2 + \left(M_z^{(2)}\right)^2  - 2C M_z^{(1)} M_z^{(2)}\right\}\right]}.
\end{align}
Here we have defined $\tilde{U}_{\rm KI} \equiv \tilde{U}_{\rm KI}[\bar{h} \boldsymbol{\epsilon}]$ and $M_z^{(1)} = M_z \otimes \1$ ($M_z^{(2)}$ is defined similarly) as before. Note that for $C = 1$, this reduces to the transfer matrix setup from \eqref{eq:transfer 2}, \eqref{eq:O 2} , as expected.
\\\\
It is convenient now to introduce $M \equiv M_z^{(1)} + M_z^{(2)}$, $\Delta \equiv M_z^{(1)} - M_z^{(2)}$, in terms of which we have,
\begin{align}
    \label{eq:O_C exp correction form}
    \mathbb{O}_{\sigma,C} &= \mathbb{O}_{\sigma} \exp{\left[-\frac{\delta\sigma^2}{2} \left(M^2 - \Delta^2\right)\right]}
\end{align}
where $\delta = \frac{1-C}{2}$. 

One could now find unimodular eigenvectors of this new transfer matrix $\mathbb{T}_{\rm E}$ (if they exist) and use them to predict the behavior of the LSFF in the thermodynamic limit (TL). The approach we follow to do so is to relate the behavior of $\mathbb{T}_{\rm E}$ to the old transfer matrix $\mathbb{T}$. \eqref{eq:O_C exp correction form} is suggestive that the LSFF could be computed as a perturbative correction to the SFF with $\delta \ll 1$ as the perturbation parameter, i.e. the distributions the longitudinal fields are drawn from are nearly perfectly correlated ($C \rightarrow 1$). These corrections should be approximately exponential.

We first note that $\mathbb{O}_{\sigma,C}$ defined in \eqref{eq:O_C exp correction form} is a positive semi-definite operator with (operator) norm $\leq 1$. Thus a non-zero contribution in the TL is only expected for a vector $\ket{v} \in \mathcal{H}_t \otimes\mathcal{H}_t$ satisfying $\Delta \ket{v} = M \ket{v} = 0$ and $(\tilde{U}_{\rm KI} \otimes \tilde{U}_{\rm KI}^*) \ket{v} = e^{i \phi}\ket{v}$. The first 2 conditions imply $M_z^{(1)} \ket{v} = \pm M_z^{(2)} \ket{v}$, which is only true when $M_z^{(1)} \ket{v} = M_z^{(2)} \ket{v} = 0$. Restricting to odd times $t$, we see this cannot be true for any vector composed of $\sigma_z$-basis vectors, as the sum of the spins cannot vanish. Thus to leading order, the LSFF trivially vanishes in the TL as there exist no unimodular eigenvectors (for odd $t$).
\\\\
It is more interesting to analyze the large yet finite $L$ regime, where subleading behavior ($\mathcal{O}(\delta)$) appears. To motivate our ansatz for analyzing the numerical data, we start by expanding the LSFF in $\delta$,
\begin{align}
    \rm tr \left(\mathbb{T}_E^L\right) &= \rm tr \left(\mathbb{T} \exp{\left[-\frac{\delta \sigma^2}{2} \left(M^2 - \Delta^2\right)\right]}\right)^L \\
    &= \rm tr \left(\mathbb{T} \left[ 1 - \frac{\delta \sigma^2}{2} \left( M^2 - \Delta^2\right) + \mathcal{O}\left(\delta^2\right)\right]\right)^L \\
    &= \rm tr \left( \mathbb{T}^L\right) - \frac{\delta \sigma^2 L}{2} \rm tr \left(\mathbb{T}^{L} \left[M^2 - \Delta^2\right]\right) + \mathcal{O}\left(\delta^2\right)
\end{align}
where in the last line we used the cyclic property of the trace to collect the $\mathcal{O}\left(\delta\right)$ terms.

When expanding the trace in the 2nd term, we only consider the contributions from the unimodular eigenvectors $\ket{A}$ of the transfer matrix $\mathbb{T}$. This is because since we cannot find unimodular eigenvectors of $\mathbb{T}_{\rm E}$, the subleading contribution to its trace is given by vectors $\ket{v}$ such that $\Delta \ket{v} = 0$ and $(\tilde{U}_{\rm KI} \otimes \tilde{U}_{\rm KI}^*) \ket{v} = \ket{v}$, which are exactly the relations \eqref{eq:m2 action}, \eqref{eq:u2 action} which define $\ket{A}$.

Expanding the trace and only keeping contributions from the unimodular eigenvectors $\ket{A}$ yields,
\begin{align*}
    \rm tr \left(\mathbb{T}^{L} \left[M^2 - \Delta^2\right]\right) &= \sum_{\ket{A}}^{K(t)} \left(\bra{A} \mathbb{T}^{L} M^2 \ket{A} - \bra{A} \mathbb{T}^{L} \Delta^2 \ket{A}\right) \\
    &= \sum_{\ket{A}}^{K(t)} 4 \bra{A} \left(M_z^{(1)} \right)^2 \ket{A} = 4 K(t) \overline{\bra{A}\left(M_z^{(1)}\right)^2 \ket{A}}
\end{align*}
where we used $\Delta \ket{A} = 0$ from \eqref{eq:m2 action} and $\mathbb{T}\ket{A} = \ket{A}$. Furthermore, we replace the sum in the second line with the average value of the inner product $\bra{A}\left(M_z^{(1)}\right)^2 \ket{A}$ over all unimodular eigenvectors $\ket{A}$ (note this is not an ensemble average). This average is time dependent, and should scale with the number of terms in $M_z^{(1)}$, i.e. $\overline{\bra{A} \left(M_z^{(1)}\right)^2 \ket{A}} \sim t$. Combining with the previous result we obtain,
\begin{align}
    \rm tr \left(\mathbb{T}_E^L\right) &= K(t) \left[1 - 2 \delta\sigma^2 L t + \mathcal{O}\left(\delta^2\right)\right] \\ 
    & \sim K(t) e^{-2 \delta \sigma^2 L t} \\
    \label{eq:lsff conjecture}
    & \sim  2t e^{-2\delta\sigma^2 L t}
\end{align}
In the last line, we have conjectured an approximate re-exponentiated form of the LSFF; it is this form we will compare to numerical data shortly. As a first check, we see that indeed in the TL the exponential damping causes the LSFF to be 0 for all time. However, for large but finite L we expect some kind of exponential decay, i.e. the subleading behavior we were after.
\begin{figure}[ht]
    \centering
    \begin{subfigure}{0.49\textwidth}
        \centering
        \includegraphics[width = \textwidth]{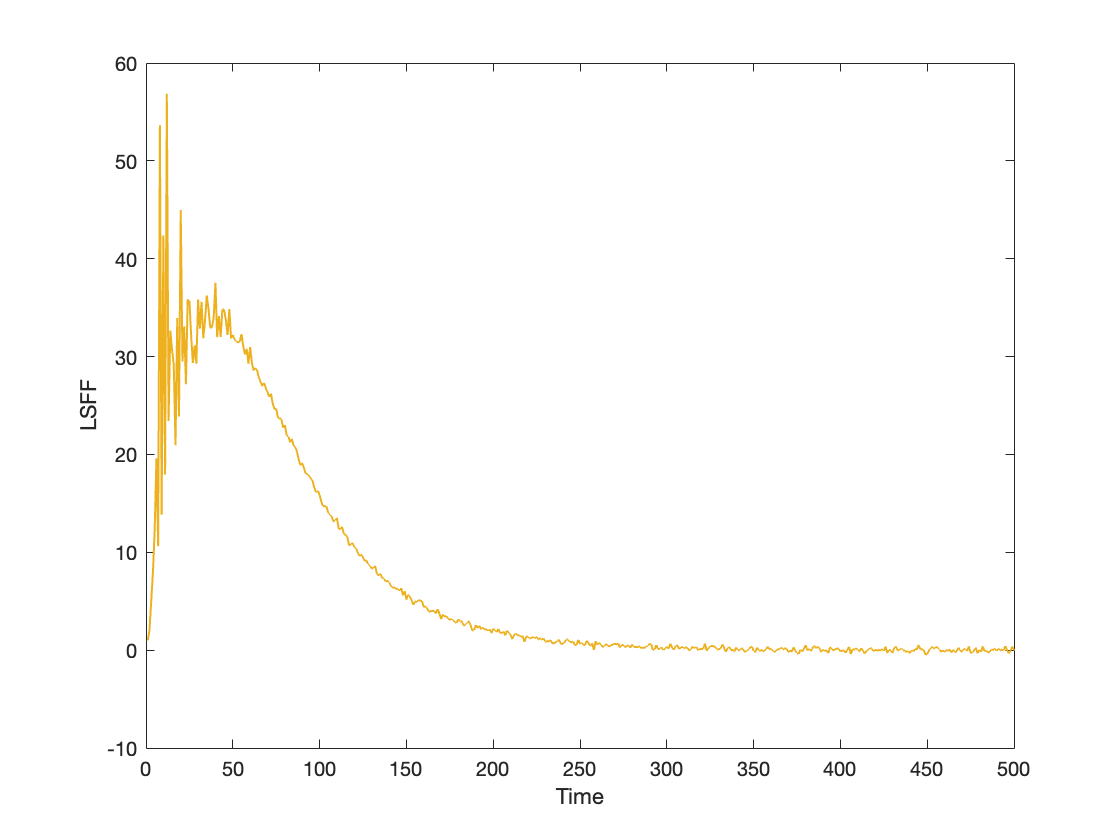}
    \end{subfigure}
    \begin{subfigure}{0.49\textwidth}
        \centering
        \includegraphics[width = \textwidth]{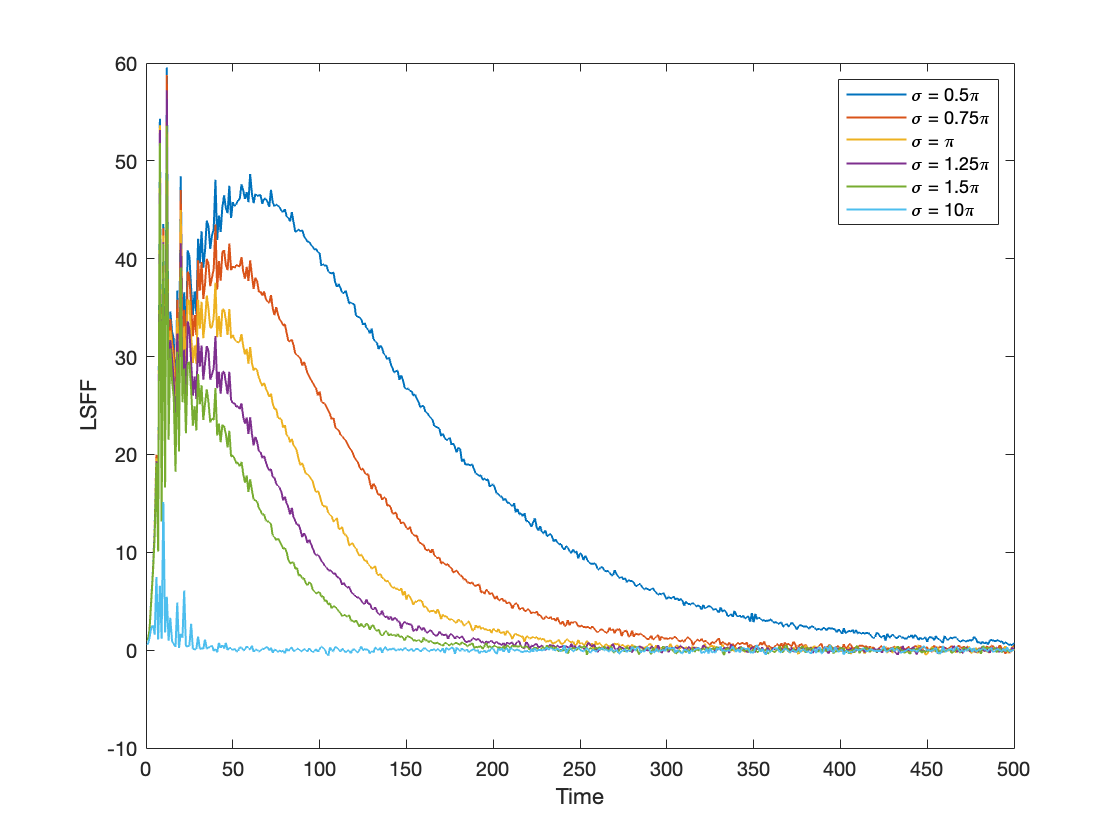}
    \end{subfigure}
    \caption{(Left) LSFF over time for $L = 6, \sigma = \pi, C = 0.9999$, (Right) Plot of LSFF for various values of $\sigma$. Ensemble average over $N = 10^5$.}
    \label{fig:lsff decay}
\end{figure}

Numerical analysis as shown in Fig. \ref{fig:lsff decay} supports this behavior, showing a ramp for early times (ignoring fluctuations) followed by an exponential decay. Furthermore, the decay is faster with increasing $\sigma$, as \eqref{eq:lsff conjecture} suggests. Unless mentioned otherwise, the data in the rest of this section is for $L = 6, \sigma = \pi,\delta = 5 \cdot 10^{-5}$ and averaged over $N = 2\cdot10^5$.

We can go a step further and fit our data to the sample curve $f(t) = 2t e^{-a t}$, where $a$ is the fitting parameter. We must be careful with the ranges we are fitting over, as this behavior is only subleading in $\delta$. Beyond subleading corrections will affect the behavior for intermediate times, and large early time fluctuations (similar to the ones we saw for the SFF) make the fit poor for early times. Thus, we only expect our ansatz to be a good fit for the late time decay. For fixed $\sigma$ and sufficiently small $\delta$ we should have a linear relationship with the fitting parameter and $\delta$, i.e. $a \propto \delta$ as our ansatz \eqref{eq:lsff conjecture} predicts, which we verify numerically.
\begin{figure}[ht]
    \centering
    \begin{subfigure}{0.49\textwidth}
        \centering
        \includegraphics[width=\linewidth]{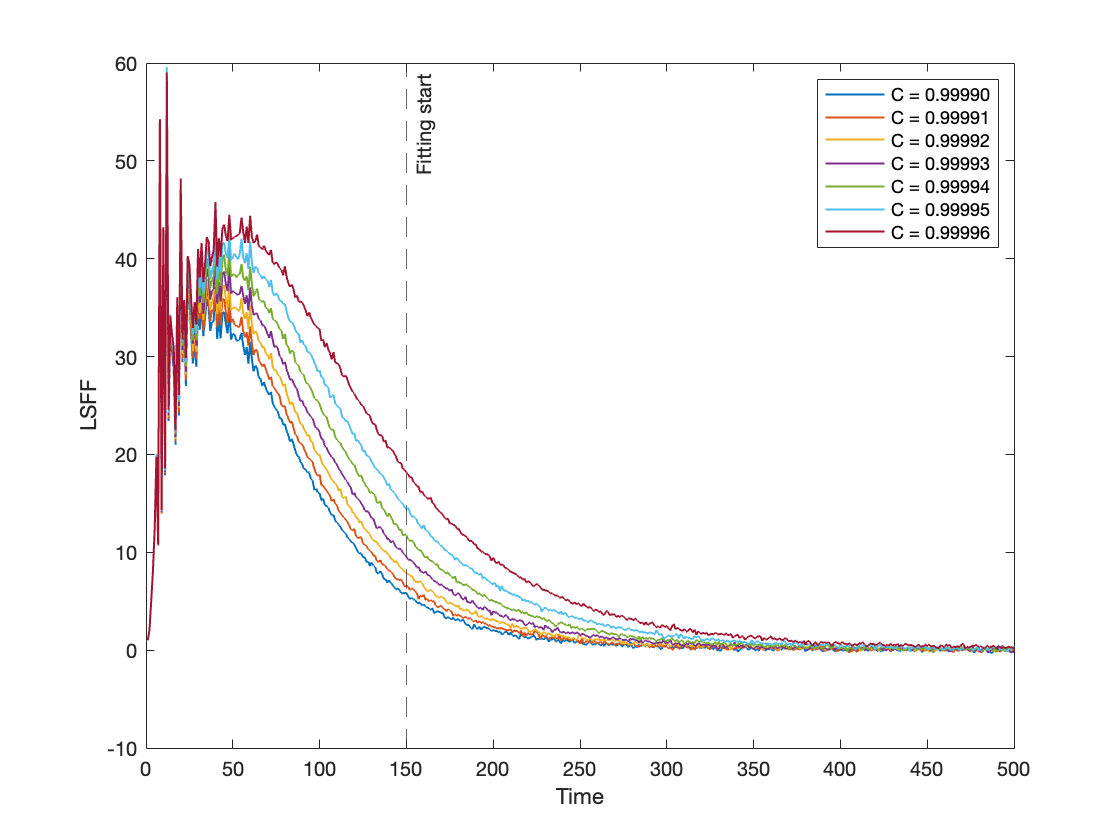}
    \end{subfigure}
    \begin{subfigure}{0.49\textwidth}
        \centering
        \includegraphics[width=\linewidth]{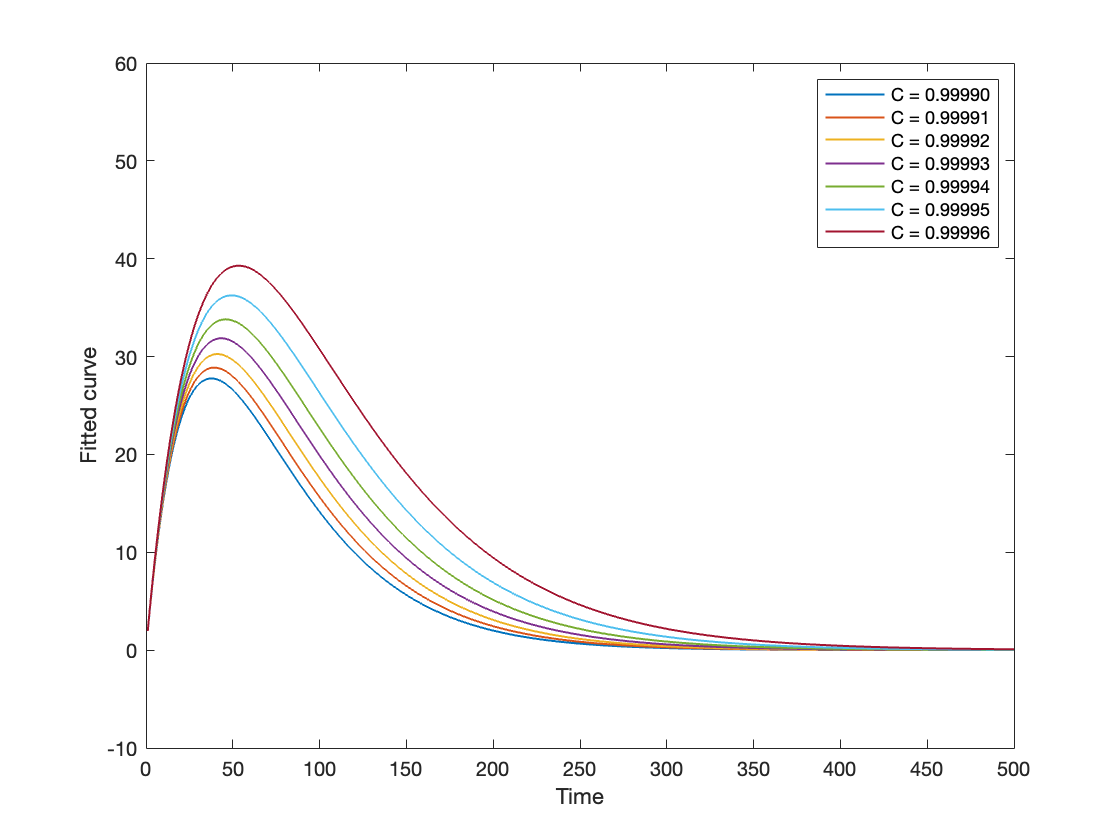}
    \end{subfigure}
    \caption{(Left) LSFF for various $\delta = \frac{1-C}{2}$ values with line showing where they were all fit to $f(t)$. (Right) Fits of the curves on the left starting at $t = 150$.}
    \label{fig:varying C}
\end{figure}

\begin{figure}[ht]
    \centering
    \begin{subfigure}{0.49\textwidth}
        \centering
        \includegraphics[width=\linewidth]{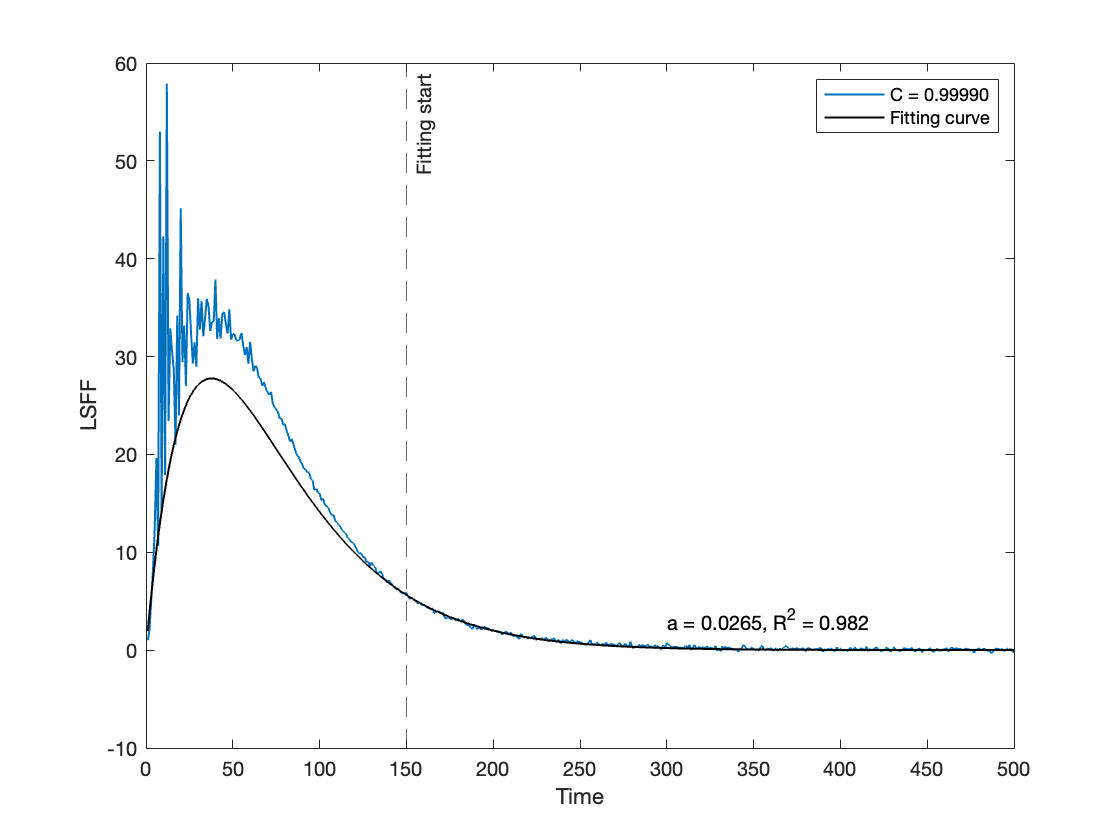}
    \end{subfigure}
    \begin{subfigure}{0.49\textwidth}
        \centering
        \includegraphics[width=\linewidth]{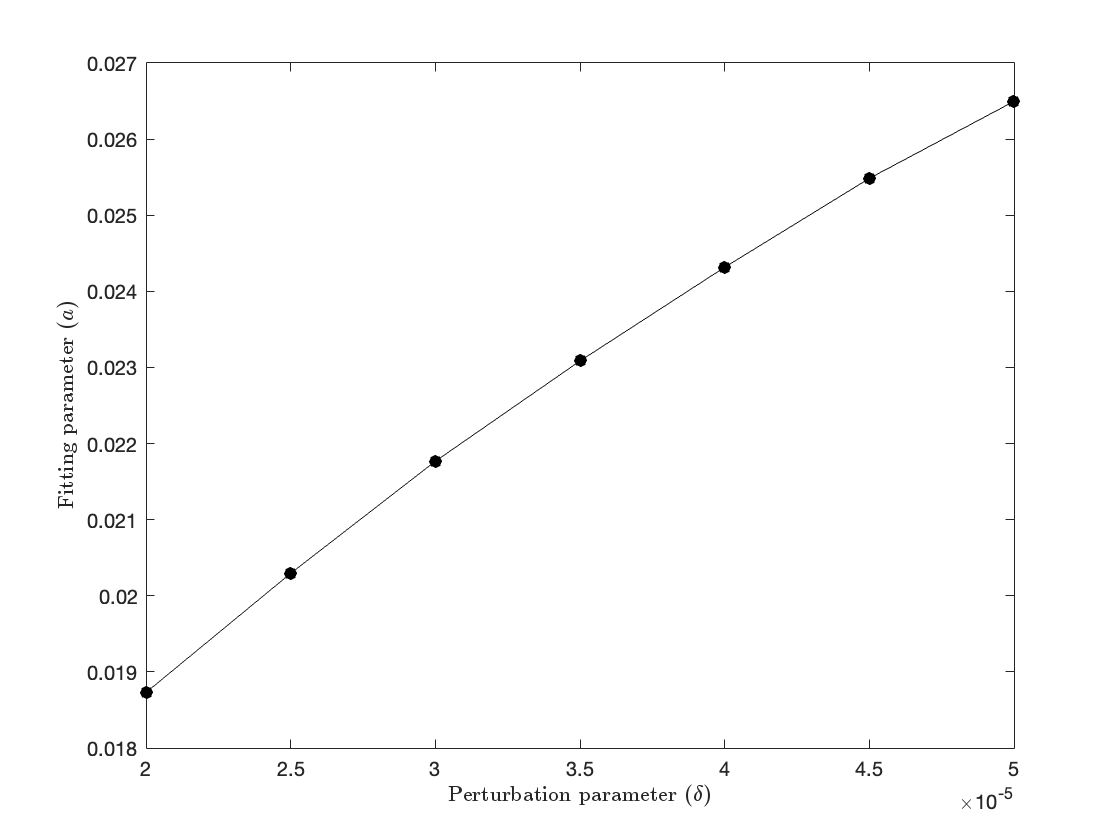}
    \end{subfigure}
    \caption{(Left) Fit shown for $\delta = 5\cdot10^{-5}$ along with value for fitting parameter $a$ and coefficient of determination $R^2$. (Right) Dependence of fitting parameter $a$ on perturbation parameter $\delta$.}
    \label{fig:fitting C}
\end{figure}

Fig. \ref{fig:varying C} shows the $\delta$-dependence of the data, along with the fits performed (at late times). Fig. \ref{fig:fitting C} shows the results of this fitting, yielding a nearly linear relationship between $a$ (fitting parameter) and $\delta$, as expected from \eqref{eq:lsff conjecture}.

\subsection{Higher order LSFF}
Following the approach we followed with the SFF, it is natural to consider higher moment analogues of the LSFF. The LSFF is no longer the 2nd moment of some return amplitude; however, we can instead consider deformations of the 4th moment of $Z(t)$ to define the 4th order LSFF,
\begin{equation}
    M_{\rm 4E} = {\rm tr}\left(\mathbb{T}_{\rm E}^{{(4)}^L}\right)
\end{equation}
where,
\begin{align}
    \mathbb{T}_{\rm E}^{(4)} &= \mathbb E_{h_1,h_2}\left[\tilde{U}_{\rm KI}[h_1 \boldsymbol \epsilon] \otimes \tilde{U}_{\rm KI}[h_2 \boldsymbol \epsilon]^* \otimes \tilde{U}_{\rm KI}[h_2 \boldsymbol \epsilon] \otimes \tilde{U}_{\rm KI}[h_1 \boldsymbol \epsilon]^*\right] = (\tilde{U}_{\rm KI} \otimes \tilde{U}_{\rm KI}^* \otimes \tilde{U}_{\rm KI} \otimes \tilde{U}_{\rm KI}^* )\cdot \mathbb{O}^{(4)}_{\sigma, C} \\
    \mathbb{O}^{(4)}_{\sigma,C} &= \mathbb{O}_{\sigma}^{(4)} \exp{\left[-\frac{\delta\sigma^2}{2} \left(M^2 - \Delta^2\right)\right]}.
\end{align}
Here we again define $\delta = \frac{1-C}{2}$ to be our perturbation parameter and $M_z^{(1)} = M_z \otimes \1 \otimes \1 \otimes \1$. We have used the variables $M \equiv M_z^{(1)} + M_z^{(2)} - M_z^{(3)} - M_z^{(4)}$ and $\Delta \equiv M_z^{(1)} - M_z^{(2)} + M_z^{(3)} - M_z^{(4)}$, and once again $\mathbb{O}_\sigma^{(4)} = \exp{\left(-\frac{\sigma^2}{2} \Delta^2\right)}$ is recovered for $C = 1$.

Similar to the LSFF, the trace has non-zero contributions in the TL only if there exist unimodular eigenvectors of $\mathbb{T}_{\rm E}^{(4)}$. These vectors must be unimodular eigenvectors $\ket{A^{(4)}}$ of $\mathbb{T}^{(4)}_{\rm E}$, satisfying the additional constraint $M \ket{A^{(4)}} = 0$. While no such eigenvectors could exist for the LSFF, they do for the 4th moment. Sec. \ref{sec:periodic} describes how the eigenvectors $\ket{A^{(4)}}$ are constructed by pairings of the 2nd order eigenvectors $\ket{A}$, where we obtain 3 different pairings: vectors that pair (1) $M_z^{(1)}$ with $M_z^{(2)}$ and $M_z^{(3)}$ with $M_z^{(4)}$, (2) $M_z^{(1)}$ with $M_z^{(4)}$ and $M_z^{(2)}$ with $M_z^{(3)}$, and (3) $M_z^{(1)}$ with $-M_z^{(3)}$ and $M_z^{(2)}$ with $-M_z^{(4)}$. For eigenvectors of pairing (2) we have $M_z^{(1)} \ket{A^{(4)}} = M_z^{(4)}\ket{A^{(4)}}$ and $M_z^{(2)} \ket{A^{(4)}} = M_z^{(3)}\ket{A^{(4)}}$, which imply $M \ket{A^{(4)}} = 0$. Since each pairing contributes a factor of $K(t)^2$ to $M_4 (t)$, we see that the leading order ($\mathcal{O}(\delta^0)$) contribution to $M_{\rm 4E}$ must be $K(t)^2$. It should also be noted that the fact that eigenvectors corresponding to pairing (2) have the leading order contribution is an artifact of the choice of signature we used to define $\mathbb{T}^{(4)}$, $(-,*,-,*)$. Changing the signature to $(-,-,*,*)$ would instead result in the eigenvectors corresponding to pairing (1) having the leading order contribution. However, pairing (3) eigenvectors do not have the leading order contribution for any choice of signature.

To obtain this result, along with the subleading contribution that appears at finite size ($\mathcal{O}(\delta)$), we repeat the procedure for the LSFF and expand $M_{\rm 4E}$ in $\delta$,
\begin{align}
    \rm tr \left(\mathbb{T}_E^{{(4)}^L}\right) &= \rm tr \left(\mathbb{T}^{(4)} \exp{\left[-\frac{\delta \sigma^2}{2} \left(M^2 - \Delta^2\right)\right]}\right)^L \\
    &= \rm tr \left(\mathbb{T}^{(4)} \left[ 1 - \frac{\delta \sigma^2}{2} \left( M^2 - \Delta^2\right) + \mathcal{O}\left(\delta^2\right)\right]\right)^L \\
    \label{eq:lsff4 ansatz}
    &= \rm tr \left( \mathbb{T}^{(4)^L}\right) - \frac{\delta \sigma^2 L}{2} \rm tr \left(\mathbb{T}^{(4)^{L}} \left[M^2 - \Delta^2\right]\right) + \mathcal{O}\left(\delta^2\right)
\end{align}
Now we expand the 2nd trace in terms of unimodular eigenvectors of $\mathbb{T}^{(4)}$, since only these have a subleading contribution by the same argument as for the LSFF,
\begin{align*}
    \rm{tr} \left(\mathbb{T}^{(4)^{L}} \left[M^2 - \Delta^2\right]\right) &= \sum_{\ket{A^{(4)}}}^{3K(t)^2} \left(\bra{A^{(4)}} \mathbb{T}^{(4)^{L}} M^2 \ket{A^{(4)}} - \bra{A^{(4)}} \mathbb{T}^{(4)^{L}} \Delta^2 \ket{A^{(4)}}\right) \\
    &= \sum_{\ket{A^{(4)}}}^{3K(t)^2} \bra{A^{(4)}} M^2 \ket{A^{(4)}}
\end{align*}
where we take the sum over eigenvectors formed from all 3 pairings. To evaluate this inner product, we must now split up the sum over the eigenvectors $\ket{A^{(4)}_1},\ket{A^{(4)}_2},\ket{A^{(4)}_3}$ corresponding to the 3 different pairings discussed earlier. As we mentioned, $M\ket{A^{(4)}_2} = 0$ so only 2 of the eigenvectors contribute. Note that for $\ket{A^{(4)}_1}$ we set $M_z^{(2)} \ket{A^{(4)}_1}  = M_z^{(1)} \ket{A^{(4)}_1}$, $M_z^{(3)} \ket{A^{(4)}_1}  = M_z^{(4)} \ket{A^{(4)}_1}$ and for $\ket{A_3^{(4)}}$ we set $M_z^{(2)} \ket{A^{(4)}_3}  = -M_z^{(4)} \ket{A^{(4)}_3}$, $M_z^{(3)} \ket{A^{(4)}_3}  = -M_z^{(1)} \ket{A^{(4)}_3}$. Thus we obtain,
\begin{align*}
    \rm{tr} \left(\mathbb{T}^{(4)^{L}} \left[M^2 - \Delta^2\right]\right) &= \sum_{\ket{A^{(4)}}}^{2K^2(t)} 4\bra{A^{(4)}}  \left(M_z^{(1)} - M_z^{(4)}\right)^2 \ket{A^{(4)}} \\
    &= 8K^2(t) \overline{\bra{A^{(4)}}  \left(M_z^{(1)} - M_z^{(4)}\right)^2 \ket{A^{(4)}}}
\end{align*}
where we take the overline to mean the average over the non-zero contributions (i.e. from $\ket{A^{(4)}_1}$ and $\ket{A^{(4)}_3}$). Once again, we expect of order $t$ non-zero terms in this inner product, suggesting again $\overline{\bra{A^{(4)}}  \left(M_z^{(1)} - M_z^{(4)}\right)^2 \ket{A^{(4)}}} \sim t$. Substituting $\rm tr \left( \mathbb{T}^{(4)^L}\right) = 3K^2(t)$ and collecting terms in the original expression \eqref{eq:lsff4 ansatz},
\begin{align}
    \rm tr \left(\left(\mathbb{T}_E^{(4)}\right)^L\right) &= K(t)^2 + 2K(t)^2  \left[1 - 2 \delta\sigma^2 L t + \mathcal{O}\left(\delta^2\right)\right] \\ 
    & \sim K(t)^2 + 2K(t)^2 e^{-2 \delta \sigma^2 L t} \\
    \label{eq:lsff4 conjecture}
    & \sim  4t^2 + 8t^2 e^{-2\delta\sigma^2 L t}
\end{align}
Similar to the SFF, there is an exponential decay for large finite $L$ but now we also have $M_{\rm 4E}(t) = 4t^2$ in the TL, which is simply the 2nd moment squared of the COE. For finite $L$ we thus expect some decay behavior, but instead of decaying to 0 it plateaus at the same value the 2nd moment squared of the COE plateaus as shown in Fig. \ref{fig:lsff4 var}.
\begin{figure}[ht]
    \centering
    \begin{subfigure}{0.49\textwidth}
        \centering
        \includegraphics[width=\linewidth]{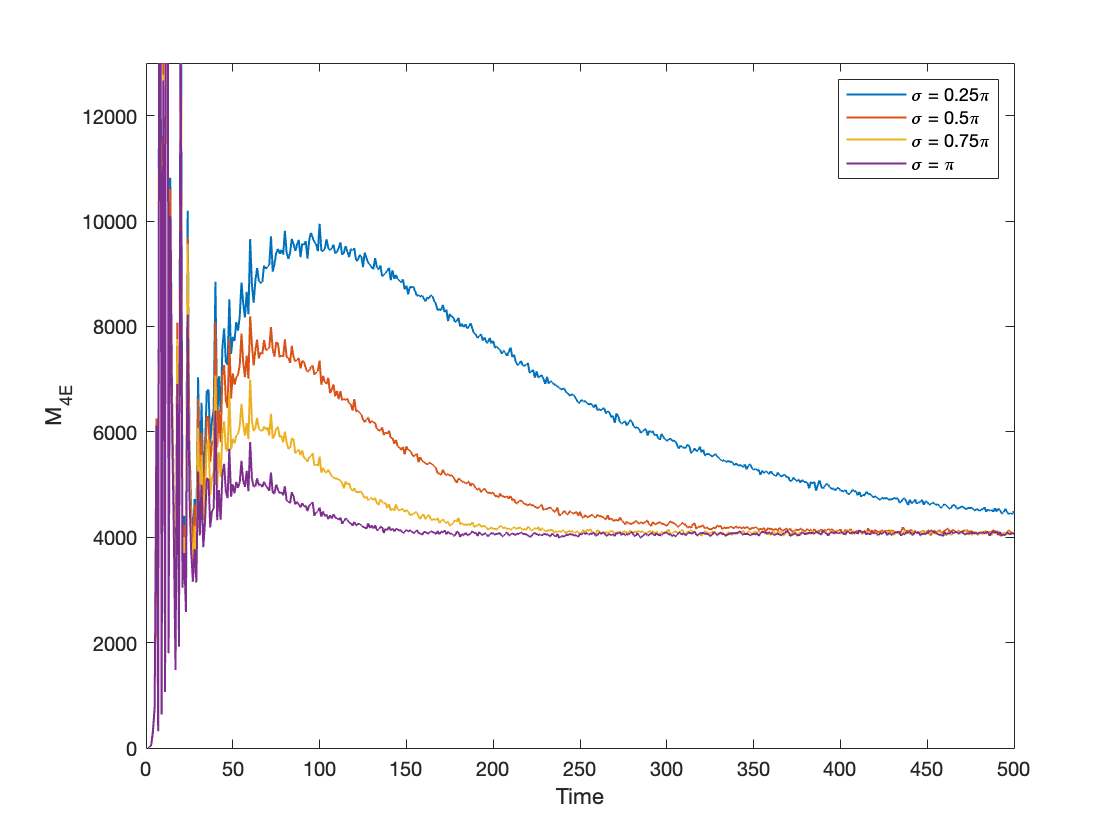}
    \end{subfigure}
    \begin{subfigure}{0.49\textwidth}
        \centering
        \includegraphics[width=\linewidth]{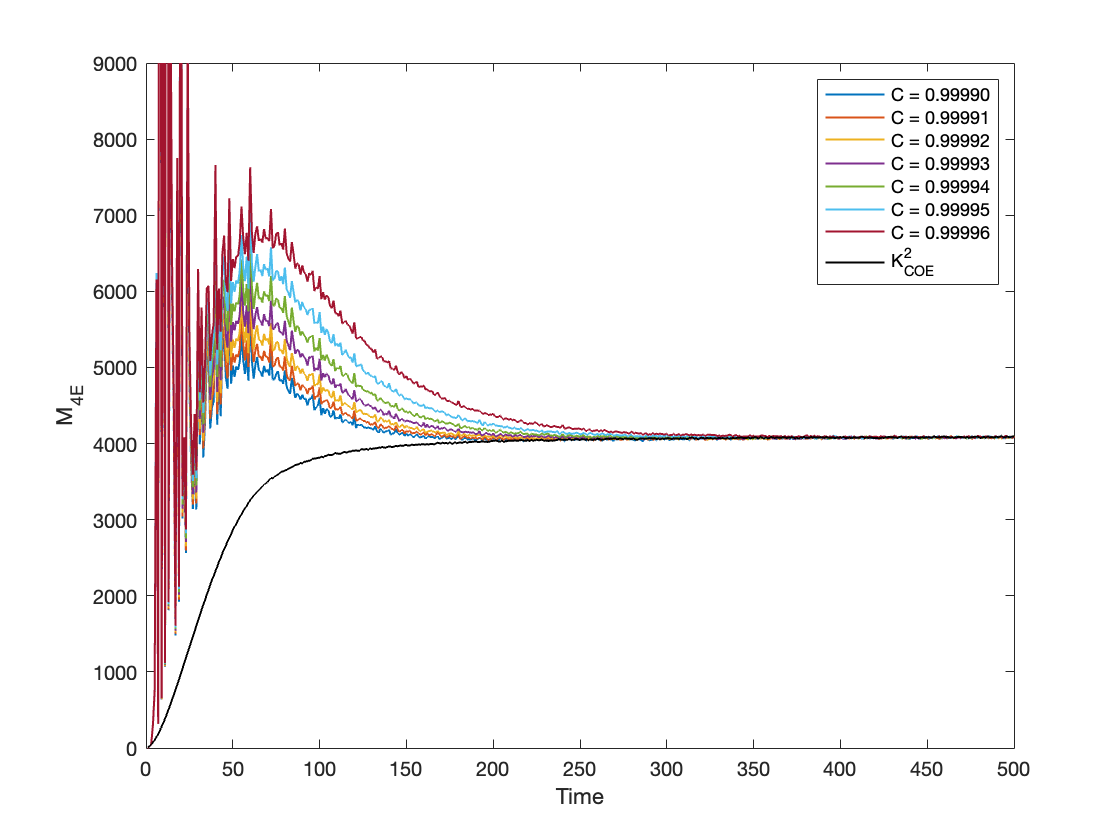}
    \end{subfigure}
    \caption{Plots of $M_{\rm 4E}(t)$ varying (left) $\sigma$ and (right) $\delta$. The black curve is the leading order behavior given by $K^2_{\rm COE}$.}
    \label{fig:lsff4 var}
\end{figure}
We can once again fit the data to a sample curve. To perform the fits, we subtract off the leading order behavior of $K^2_{\rm COE}$ and fit to the sample curve $f(t) = 8t^2 e^{-at}$. Once again, we test the relationship between the fitting parameter $a$ and the perturbation parameter $\delta$.
\begin{figure}[ht]
    \centering
    \begin{subfigure}{0.49\textwidth}
        \centering
        \includegraphics[width=\linewidth]{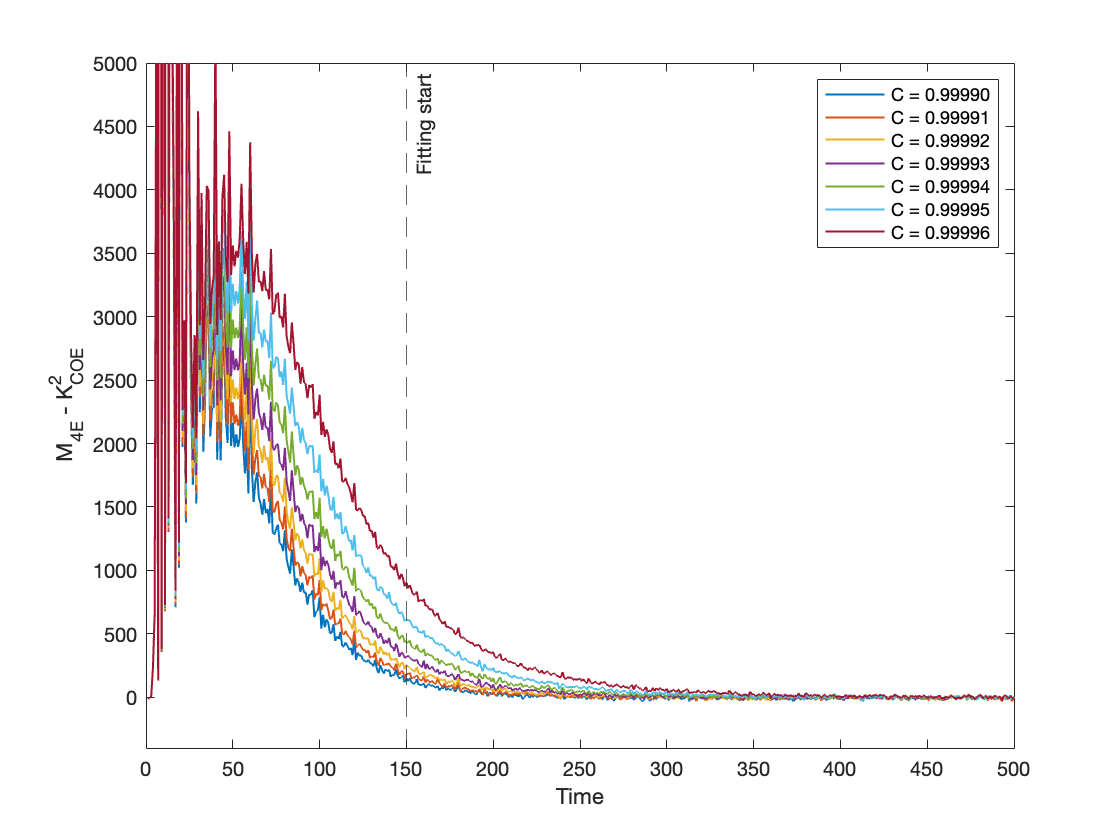}
    \end{subfigure}
    \begin{subfigure}{0.49\textwidth}
        \centering
        \includegraphics[width=\linewidth]{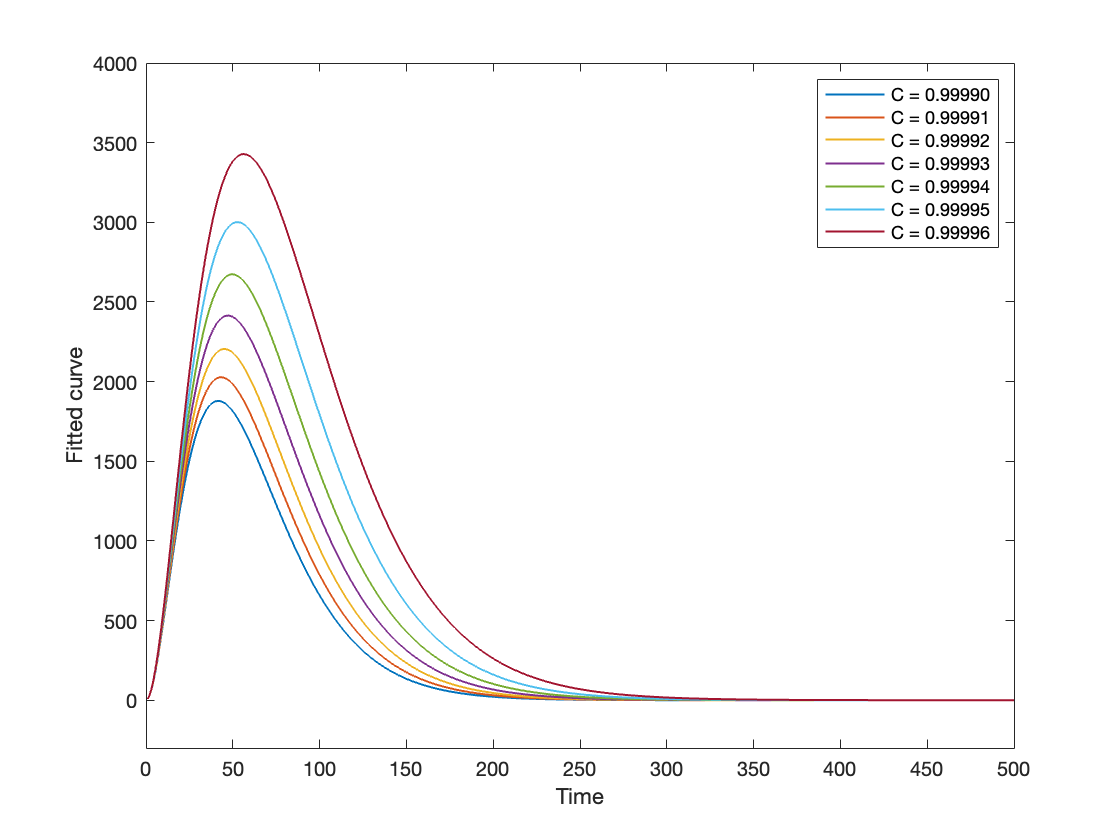}
    \end{subfigure}
    \caption{(Left) $M_{4E} - K_{\rm COE}^2$ for various $\delta$ values with line showing where they were all fit to $f(t)$. (Right) Fits of curves on the left starting at $t = 150$. Averaged over $N = 2 \cdot 10^6$}
    \label{fig:varying C4}
\end{figure}

\begin{figure}[ht]
    \centering
    \begin{subfigure}{0.49\textwidth}
        \centering
        \includegraphics[width=\linewidth]{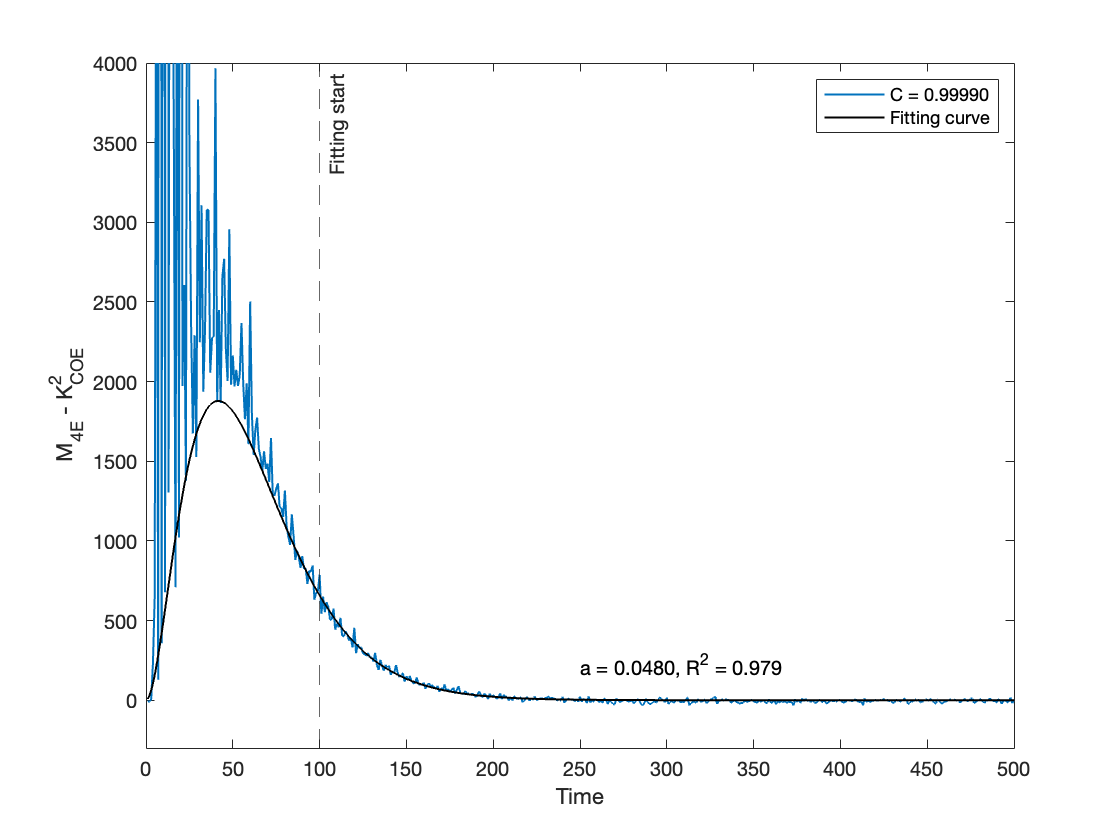}
    \end{subfigure}
    \begin{subfigure}{0.49\textwidth}
        \centering
        \includegraphics[width=\linewidth]{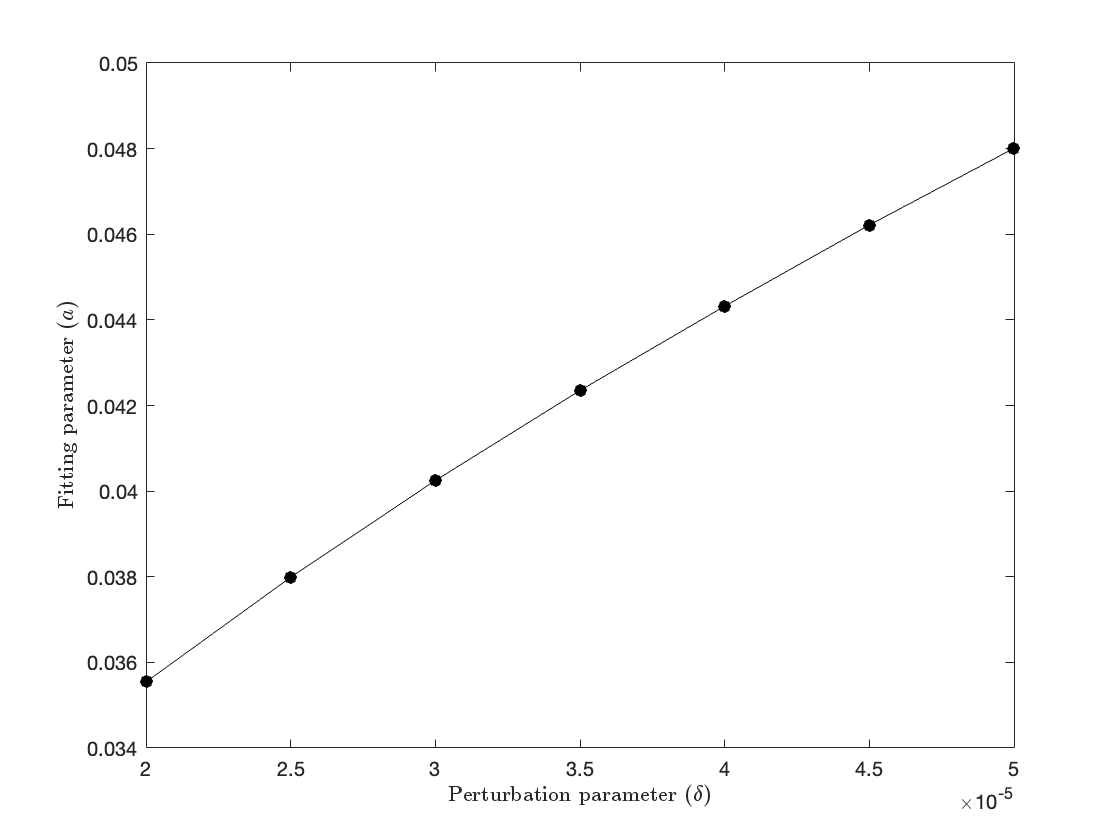}
    \end{subfigure}
    \caption{(Left) Fit shown for $\delta = 5\cdot10^{-5}$ along with value for fitting parameter $a$ and coefficient of determination $R^2$. (Right) Dependence of fitting parameter $a$ on perturbation parameter $\delta$}
    \label{fig:fitting C4}
\end{figure}

Fig. \ref{fig:varying C4} shows the $\delta$-dependence of the data, along with the fits performed (at late times). Fig. \ref{fig:fitting C4} shows the results of this fitting, yielding a nearly linear relationship between $a$ (fitting parameter) and $\delta$, as expected from \eqref{eq:lsff4 conjecture}.

\subsection{Open boundary conditions}
The LSFF (and its higher order variants) can also be computed for open boundary conditions. One method would involve computing the appropriate overlap between the boundary vectors from Sec. \ref{sec:open} and the relevant unimodular eigenvectors. Another simply makes use of our knowledge of the numerical data for open boundary conditions to verify the numerical data for the LSFF for open boundary conditions.

The periodic and open case agree for the 2nd moment, i.e. $K(t) = K_{\rm open}(t)$. Since the perturbative expansion of the LSFF can be written as $K(t) \left[1 - 2 \delta \sigma^2 L t + \mathcal{O}\left(\delta^2\right)\right] \sim K(t) e^{-2\delta\sigma^2 Lt}$, we use the same ansatz for open boundary conditions i.e. $K_{\rm E}^{(\rm open)} = K_{\rm open}(t) e^{-2 \delta \sigma^2 L t} = 2t e^{-2 \delta \sigma^2 Lt}$. It should be noted that while the LSFF for both open and periodic boundary conditions averages to a real number for large $N$, the return amplitudes $Z_1,Z_2$ for open boundary conditions have complex random statistics. Thus for finite $N$ the LSFF for open boundary conditions has a non-trivial imaginary part; however this can be ignored as this goes to 0 for large $N$.
\begin{figure}[ht]
    \centering
    \begin{subfigure}{0.49\textwidth}
        \centering
        \includegraphics[width=\linewidth]{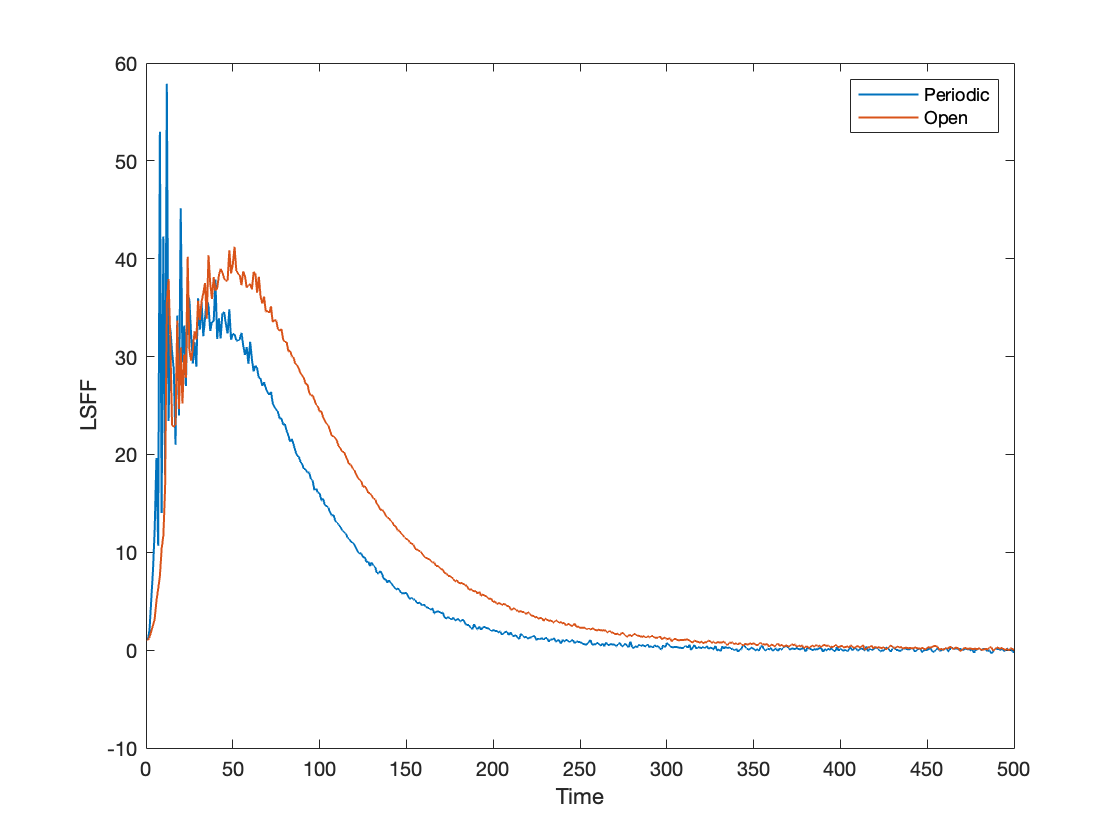}
    \end{subfigure}
    \begin{subfigure}{0.49\textwidth}
        \centering
        \includegraphics[width=\linewidth]{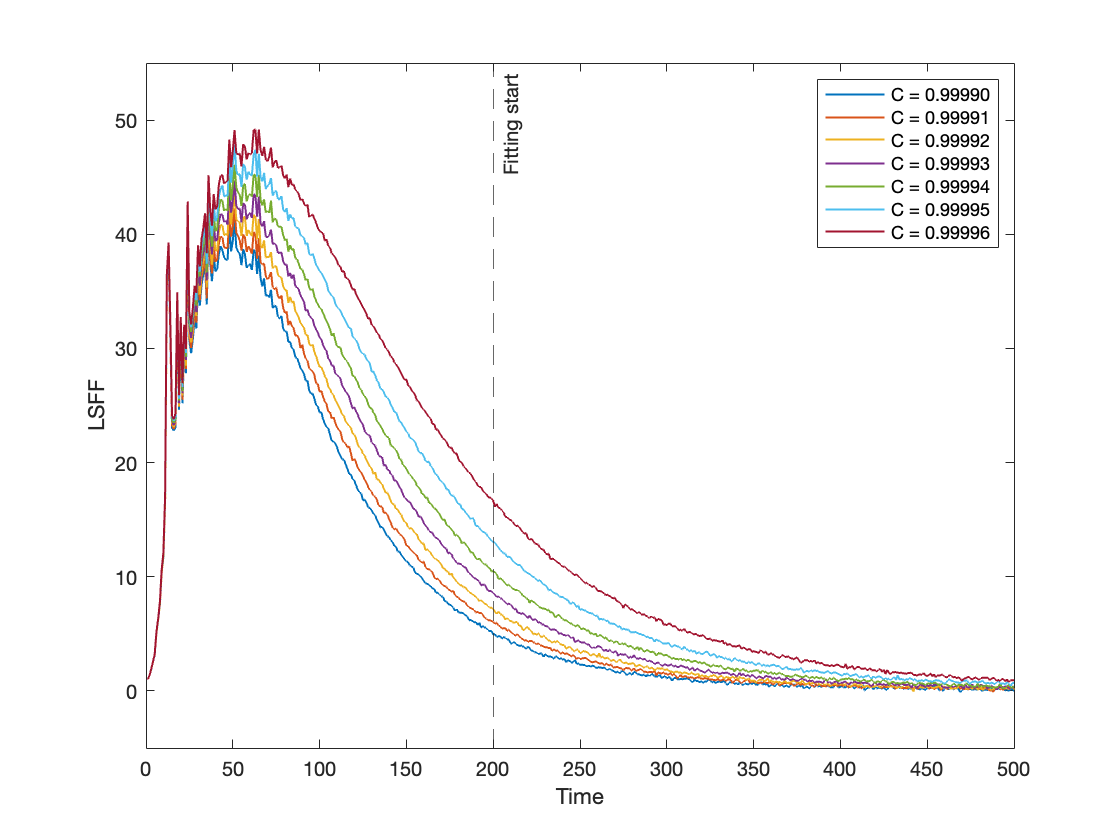}
    \end{subfigure}
    \caption{LSFF for (left) periodic vs open boundary conditions, (right) open boundary conditions for various $\delta$ values}
    \label{fig:lsff open}
\end{figure}

\begin{figure}[ht]
    \centering
    \begin{subfigure}{0.49\textwidth}
        \centering
        \includegraphics[width=\linewidth]{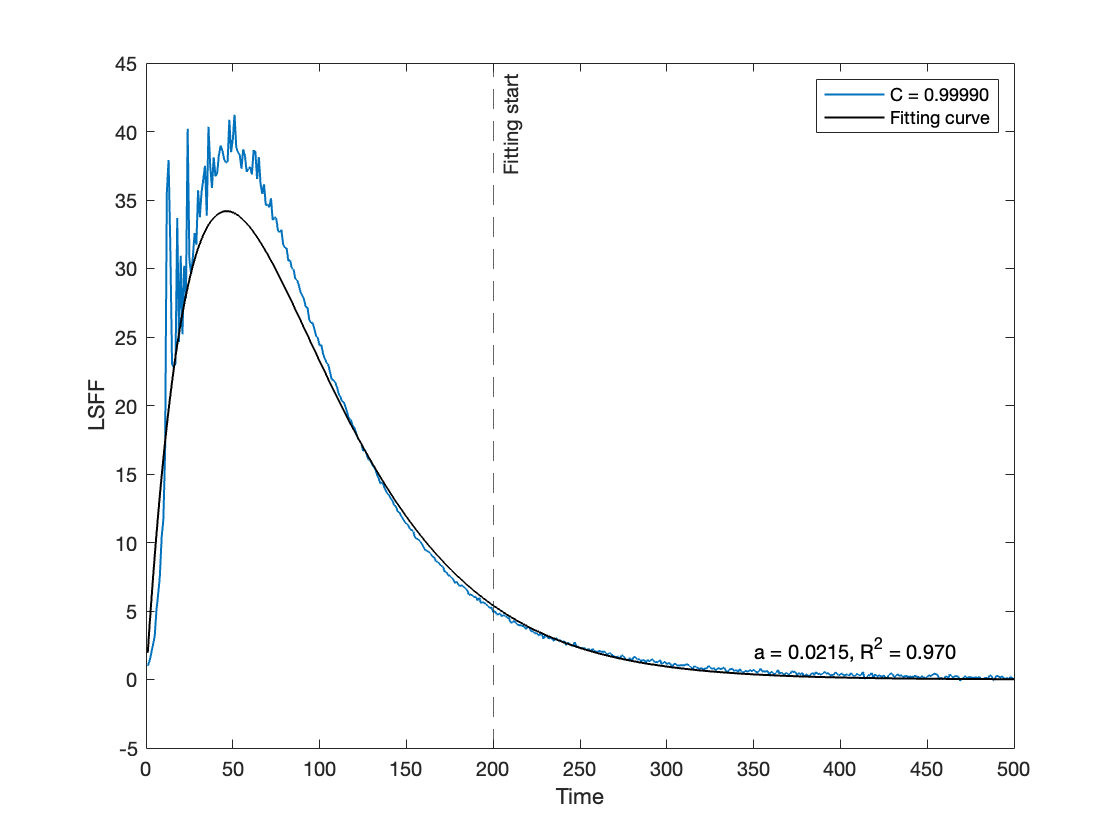}
    \end{subfigure}
    \begin{subfigure}{0.49\textwidth}
        \centering
        \includegraphics[width=\linewidth]{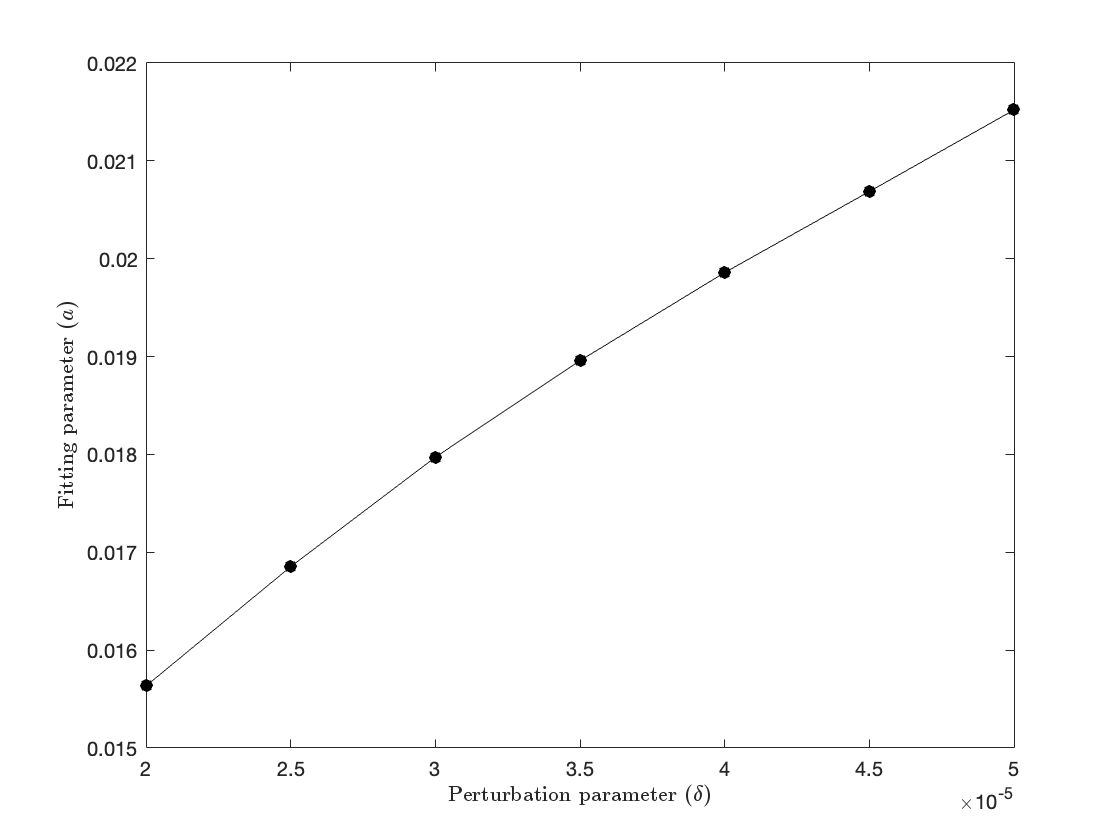}
    \end{subfigure}
    \caption{(Left) Fit shown for $\delta = 5\cdot10^{-5}$ along with value for fitting parameter $a$ and coefficient of determination $R^2$. (Right) Dependence of fitting parameter $a$ on perturbation parameter $\delta$}
    \label{fig:fitting C open}
\end{figure}

Fig. \ref{fig:lsff open} shows plots of the periodic and open LSFF together for the same parameters. Here we can see that the exponential decay ansatz is valid for both, albeit for differing decay parameters. We can also verify that when fitting to a sample curve $f(t) = 2t e^{-at}$ while varying $\delta$ (see Fig. \ref{fig:lsff open}), we obtain a linear relationship between the fitting parameter $a$ and $\delta$, like in the periodic case (see Fig. \ref{fig:fitting C open}).

For the 4th order LSFF, we must account for the difference between $M_4$ and $M_4^{(\rm open)}$. The 4th order LSFF is written as the perturbative expansion,
\begin{align*}
    M_{\rm 4E}(t) &= \rm tr \left( \mathbb{T}^{(4)^L}\right) - \frac{\delta \sigma^2 L}{2} \rm tr \left(\mathbb{T}^{(4)^{L-1}} \left[M^2 - \Delta^2\right]\right) + \mathcal{O}\left(\delta^2\right) \\
    &= M_4(t) - 2 \delta \sigma^2 L \sum_{\ket{A^{(4)}}}^{M_{4}(t)}  \bra{A^{(4)}}  \left(M_z^{(1)} - M_z^{(4)}\right)^2 \ket{A^{(4)}}
\end{align*}
Modifying appropriately for open boundary conditions we replace $M_4(t) = 3K(t)^2$ with $M_4^{(\rm open)}(t) = 2K(t)^2$. Furthermore, the sum over the 2nd term is over all 3 pairings, but only the 1st and 3rd pairings contribute. For open boundary conditions, we anticipate the 3rd pairing to be absent and so the sum is only over the 1st and 2nd pairings, of which only the 1st pairing has a non-zero contribution. We again expect the averaged value $ \overline{\bra{A^{(4)}}  \left(M_z^{(1)} - M_z^{(4)}\right)^2 \ket{A^{(4)}}} \sim t$ to have the same time dependence. Thus we obtain,
\begin{align}
    M_{\rm 4E}^{(\rm open)}(t) &= K(t)^2 + K(t)^2 \left[1 - 2 \delta\sigma^2 L t + \mathcal{O}\left(\delta^2\right)\right] \\ 
    & \sim K(t)^2 + K(t)^2 e^{-2 \delta \sigma^2 L t} \\
    \label{eq:lsff4 open conjecture}
    & \sim  4t^2 + 4t^2 e^{-2\delta\sigma^2 L t}
\end{align}

\begin{figure}[ht]
    \centering
    \begin{subfigure}{0.49\textwidth}
        \centering
        \includegraphics[width=\linewidth]{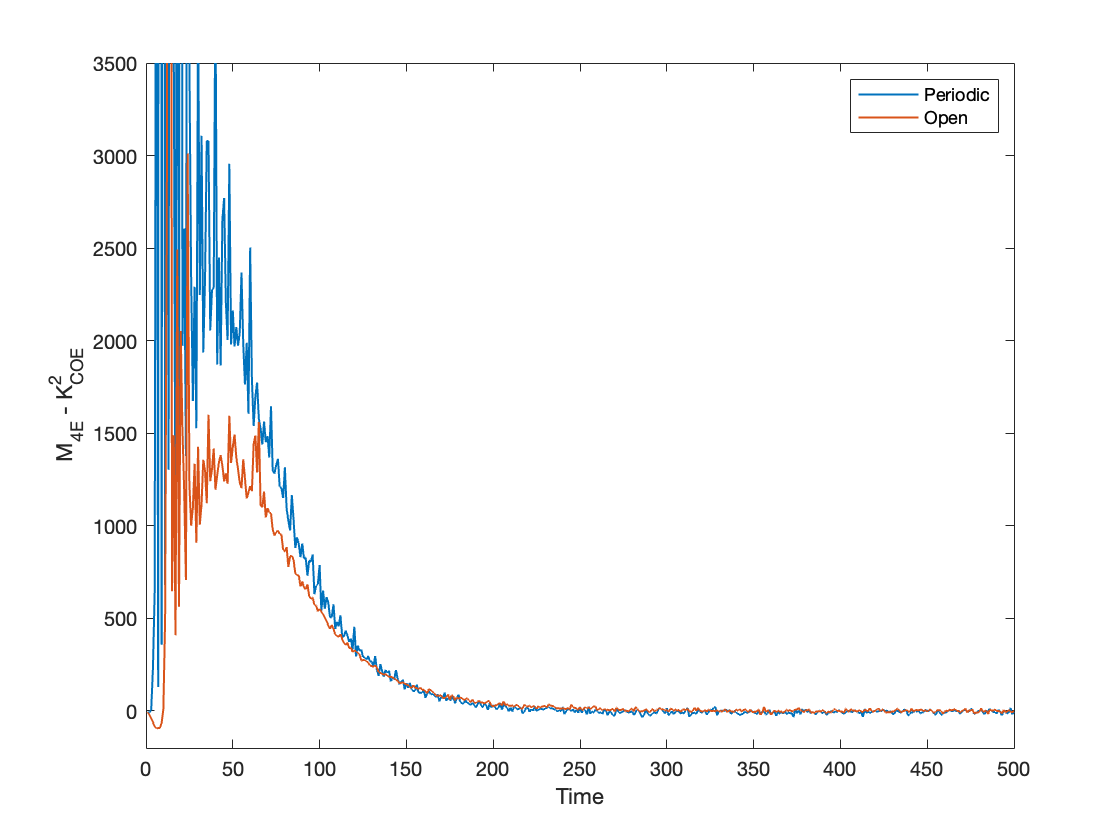}
    \end{subfigure}
    \begin{subfigure}{0.49\textwidth}
        \centering
        \includegraphics[width=\linewidth]{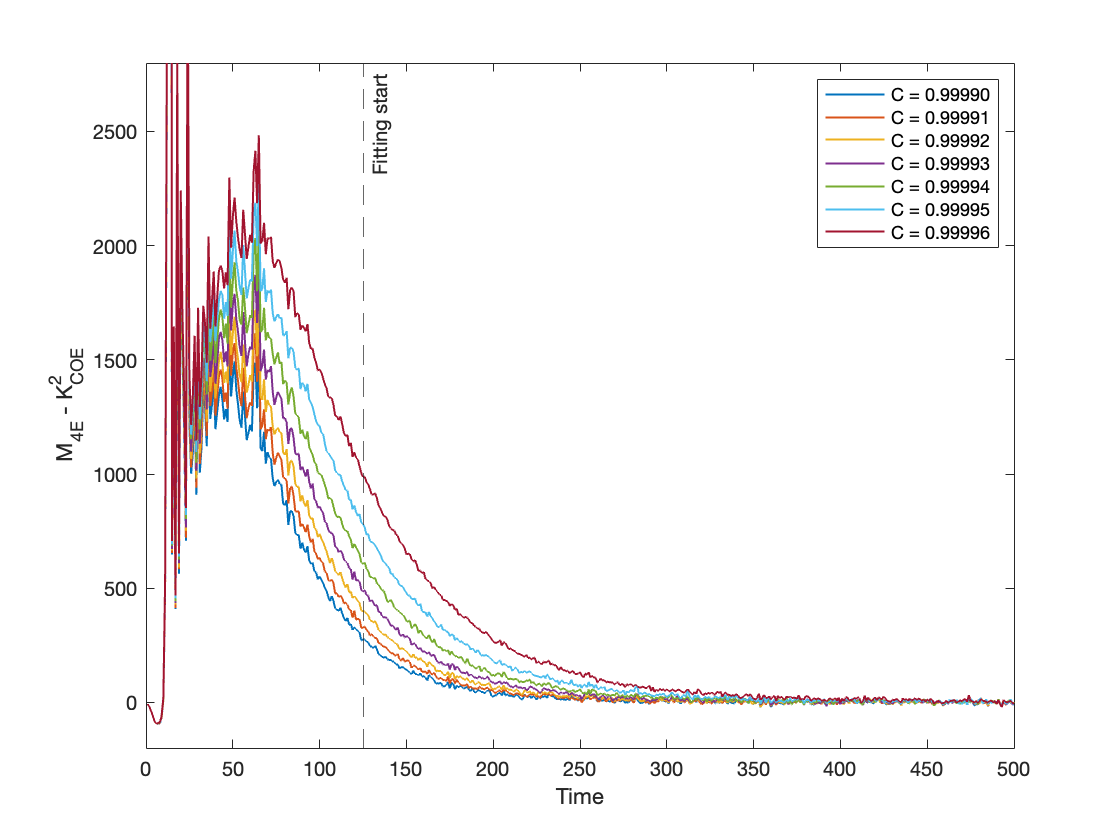}
    \end{subfigure}
    \caption{$M_{\rm 4E} - K_{\rm COE}^2$ for (left) periodic vs open boundary conditions, (right) open boundary conditions for various $\delta$ values. Averaged over $N = 2\cdot 10^6$}
    \label{fig:lsff4 open}
\end{figure}

\begin{figure}[ht]
    \centering
    \begin{subfigure}{0.49\textwidth}
        \centering
        \includegraphics[width=\linewidth]{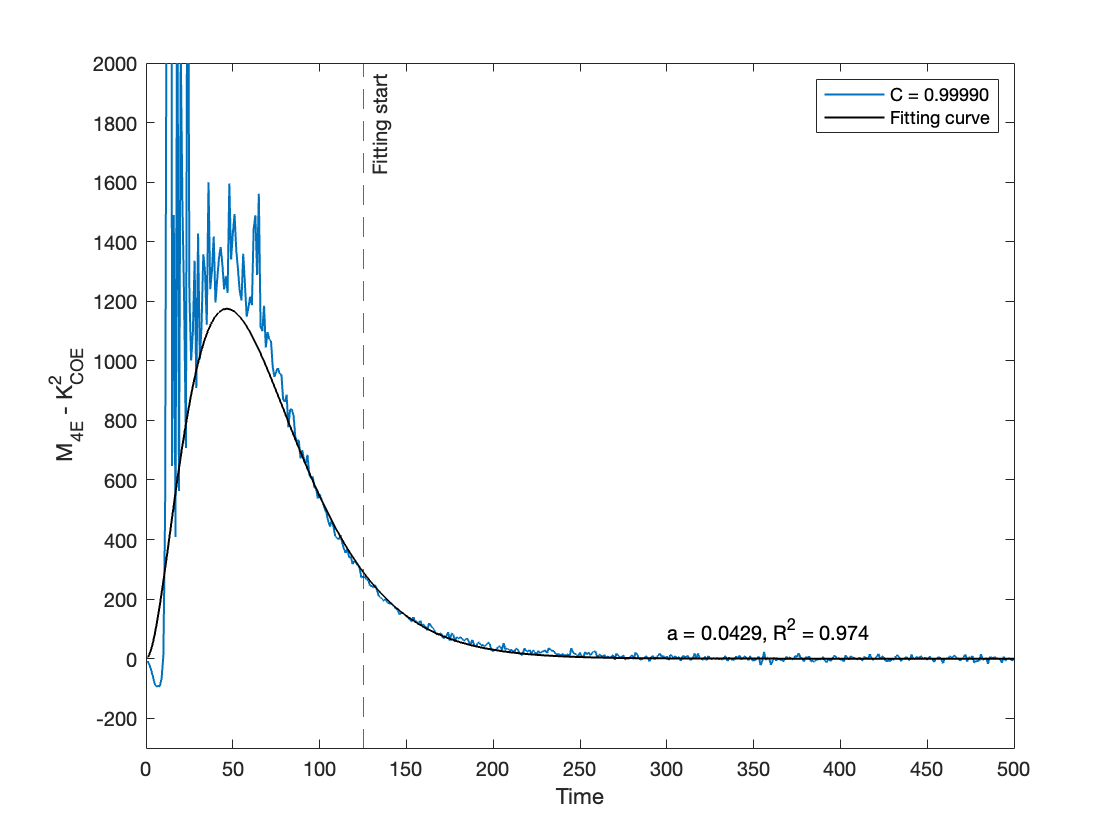}
    \end{subfigure}
    \begin{subfigure}{0.49\textwidth}
        \centering
        \includegraphics[width=\linewidth]{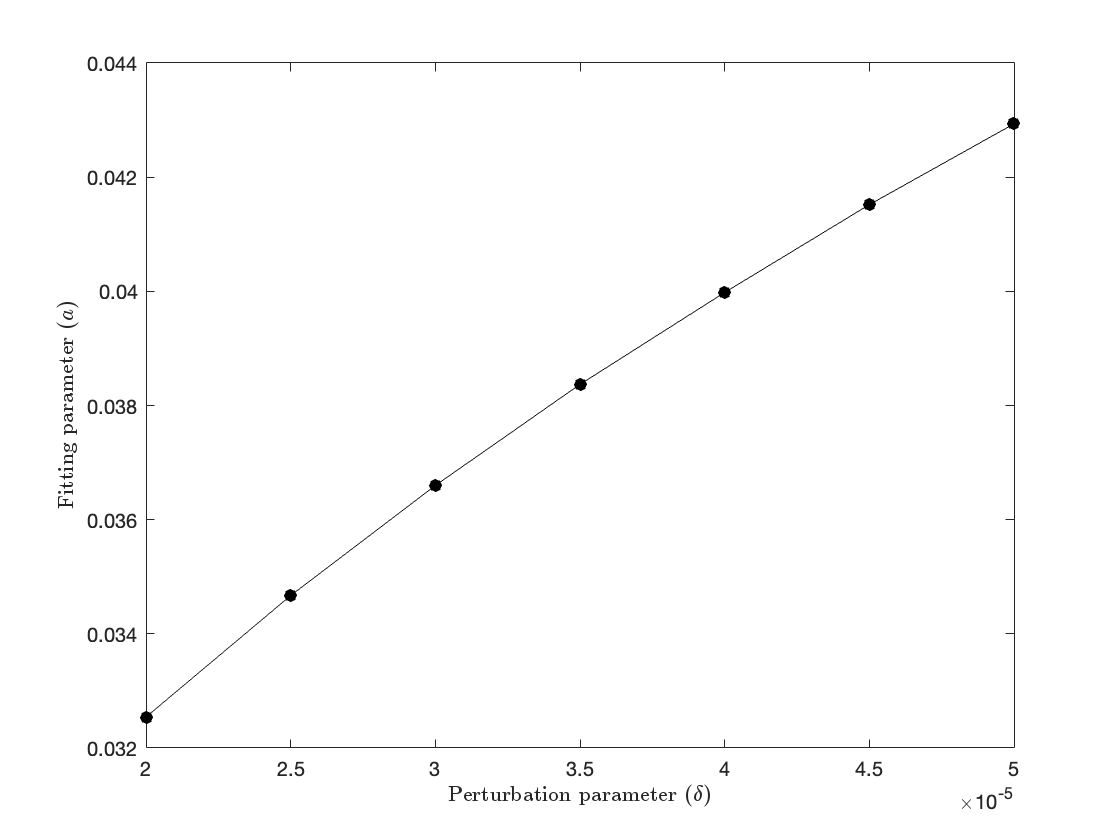}
    \end{subfigure}
    \caption{(Left) Fit shown for $\delta = 5\cdot10^{-5}$ along with value for fitting parameter $a$ and coefficient of determination $R^2$. (Right) Dependence of fitting parameter $a$ on perturbation parameter $\delta$}
    \label{fig:fitting C4 open}
\end{figure}

Note that the leading order behavior $(4t^2)$ remains the same regardless of boundary conditions; it is only the subleading behavior that changes. We see this is Fig. \ref{fig:lsff4 open}, where the ramp for the periodic case is steeper.\footnote{For early time, $M_{\rm 4E}-K_{\rm COE}^2$ for open boundary conditions is negative. This is likely due to early time fluctuations and is expected to become non-negative with sufficient averaging.}

We also vary $\delta$ (see Fig. \ref{fig:lsff4 open}) and fit to the sample curve $f(t) = 4t^2 e^{-at}$, noting once again the linear relationship between $a$ and $\delta$ (see Fig. \ref{fig:fitting C4 open}). Thus the pairings described in Sec. \ref{sec:periodic} provide a convenient method to compute the LSFF and its higher order variants for both open and periodic boundary conditions.

\section{Discussion}
\label{sec:discussion}

In this work we have considered higher order spectral statistics in the kicked Ising model for a variety of boundary conditions and couplings. Considering the trace of the Floquet operator ${\rm tr}\left(U^t_{\rm KI} [\boldsymbol h]\right)$ to be a random variable, we investigated higher order moments for both open and periodic boundary conditions on the Ising chain. Numerical analyses indicate the system is sensitive to boundary conditions, with the periodic chain obeying statistics consistent with the moments of a real random variable, which is confirmed by explicitly constructing unimodular eigenvectors of the 4th order transfer matrix $\mathbb{T}^{(4)}$ from unimodular eigenvectors of $\mathbb{T}$. This was done by pairing 2nd order eigenvectors, with each pairing contributing $K^2$ to the 4th moment, where $K$ is the 2nd moment, to recover the lower bound from \cite{Flack_2020}.  The open chain followed complex Gaussian statistics and lined up with the COE (circular orthogonal matrix) prediction. An analytic computation of the moments is proposed by altering the transfer matrix approach using boundary vectors. These results are then used to compute the leading and subleading terms for the perturbative expansion of the Loschmidt SFF and its higher order analogues for both open and periodic boundary conditions. Numerical data is fit to an exponential decay ansatz to find the predicted dependence of the decay parameter on the perturbation parameter.
\\\\
The results in this paper have all centered around the self-dual point. As shown in Figure~\ref{fig:non dual}, when away from the self-dual point, we find that both open and periodic boundary conditions reproduce the COE result. Hence, there is something special about the self-dual point as far the higher moments are concerned.
\begin{figure}[h!]
    \centering
    \begin{subfigure}[b]{0.49\textwidth}
        \centering
        \includegraphics[width=\textwidth]{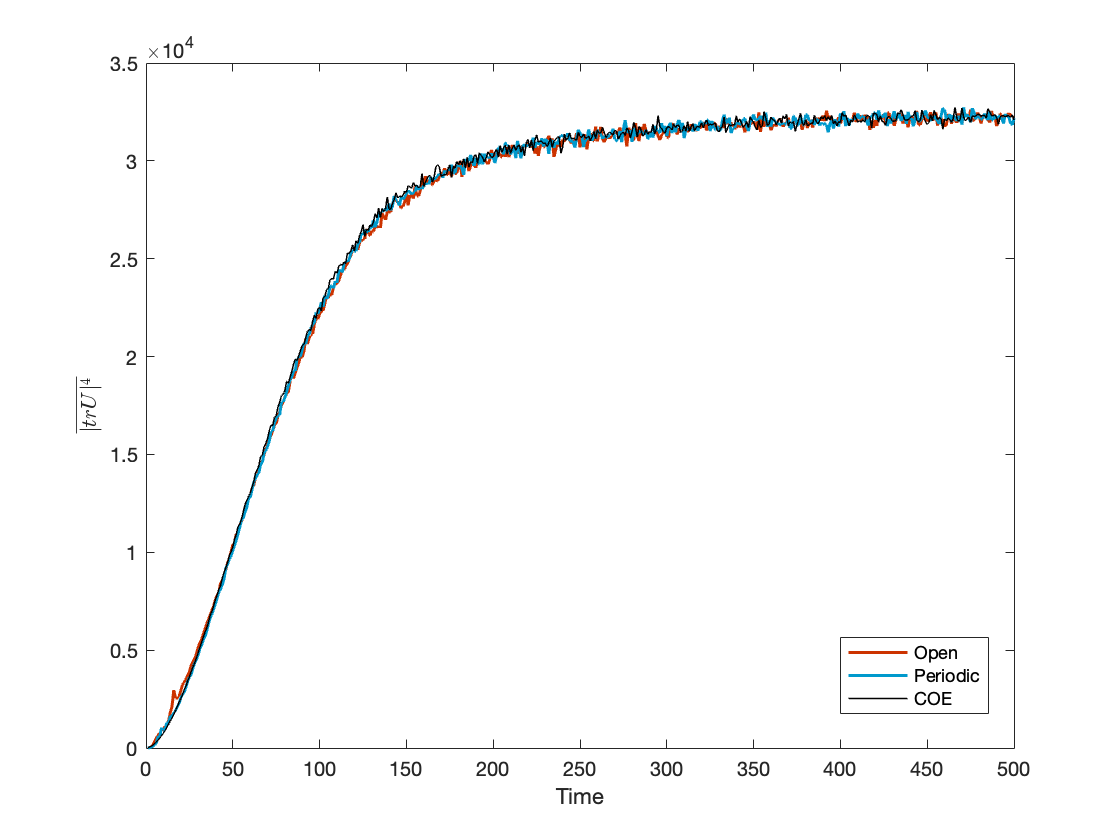}
        \subcaption{$J = (0.8)\pi/4, b = (0.9)\pi/4$}
    \end{subfigure}
    \begin{subfigure}[b]{0.49\textwidth}
        \centering
        \includegraphics[width=\textwidth]{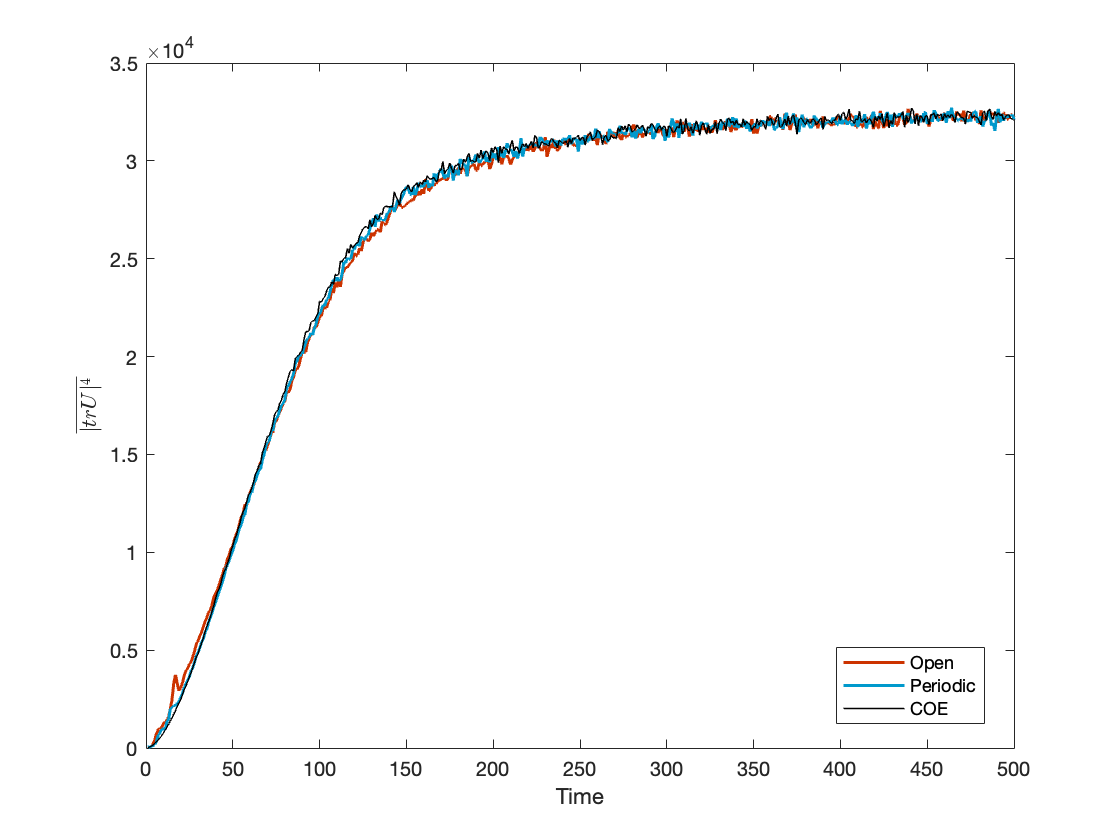}
        \subcaption{$J = (1.1)\pi/4, b = (-0.8)\pi/4$}
    \end{subfigure}
    \caption{4th moments for 7-qubit model away from the self dual points over $N = 10^5$ samples and $\bar{h} = 0.6, \sigma = 100\pi$}
    \label{fig:non dual}
\end{figure}

The proof of claim \ref{eq:claim3} from Sec. \ref{sec:periodic} helps us understand what was so special about the self-dual point(s). The original proof counting the eigenvectors of $\mathbb{T}$ did assume the system was at the self-dual point(s); however numerical results show that deviating slightly from the self-dual points does not affect the 2nd moment significantly \cite{PhysRevE.101.052201}. Thus if we assume (almost) the same 2nd moment and thus the same number of eigenvectors of $\mathbb{T}$ when we move slightly away from the self-dual point(s), one of our pairings of $\mathbb{T}^{(4)}$ must not contribute to the 4th moment since Fig. \ref{fig:non dual} shows that the 4th moment for both open and periodic boundary conditions matches the COE prediction. Our proofs of claims \ref{eq:claim1} and \ref{eq:claim2} made no reference to the self-dual point(s); however our proof of claim \ref{eq:claim3} required us to be at the self-dual point(s) for the eigenvector to exist. For general higher order moments, this means that self pairings of primed/unprimed indices among themselves would not contribute to the moment. Only pairings between a primed and unprimed index can contribute a factor of $K^n$, and there are $n!$ such pairings, giving us the expected moment that matches the COE prediction (i.e. trace behaves as a complex random variable) away from the self-dual point(s).
\\\\
A key limitation to highlight is that the result obtained for the 4th and higher order moments is only a lower bound. While we expect this bound to be saturated in the thermodynamic limit, a complete proof is absent from both our work and \cite{Flack_2020}. Furthermore, future work could focus on giving an analytic result for open boundary conditions using the boundary vector formalism and applying it to the Loschmidt SFF to obtain the results found in Sec. \ref{sec:loschmidt} via a different method. One technical limitation to highlight is the restriction to odd time. \cite{Bertini_2018} was able to prove that unimodular eigenvalues of $\mathbb{T}$ were all 1 only for odd time. This was necessary to reduce the problem of computing the SFF to counting eigenvectors, as otherwise eigenvalues of different phases may cancel when taking the trace. Additional insight can be gained for even time as \cite{Bertini_2018} did show that the only possible eigenvalues of $\mathbb{T}$ here are $\pm 1$. Furthermore, an eigenvalue of $-1$ is much `rarer', appearing in numerical analyses only for early times. Since we constructed higher order moments by pairing the 2nd order eigenvectors, this would suggest that eigenvalues of $-1$ would be just as rare for higher order moments as these would only occur if we paired 2nd order eigenvectors of eigenvalue $+1$ with $-1$. \cite{Flack_2020} suggests the lower bound computed in Sec. \ref{sec:periodic} should be saturated for even times $t \geq 6$; however, a complete analysis is required.
\\\\
Directions for future work include proving that the bound in Sec. \ref{sec:periodic} is saturated in the TL, along with a complete analysis of even times. Further analysis of open boundary conditions could also involve providing analytic and further numerical evidence for the form of the boundary vectors proposed in Sec. \ref{sec:open}. This would also provide a method for an alternate computation of the open boundary condition Loschmidt SFF in Sec. \ref{sec:loschmidt}. Future work could also involve computing subleading corrections to Gaussianity for the SFF and higher order moments.

\begin{acknowledgments}
The authors would like to thank Michael Winer for collaboration on related projects. DG is supported by the Loomis Laboratory of Physics at the University of Illinois, Urbana-Champaign. The MATLAB code and datasets used to generate the figures are available at \cite{swingle_2025_18928776}.

\end{acknowledgments}

\bibliography{bibliography}

\end{document}